\documentclass{article}      % use "amsart" instead of "article" for AMSLaTeX format
\usepackage[a4paper]{geometry}                        % See geometry.pdf to learn the layout options. There are lots.
\usepackage{graphicx}                    
\usepackage{amssymb}
\usepackage{indentfirst}
\usepackage{amsmath}
\usepackage{amsthm}
\usepackage{subcaption}
\usepackage{multirow}
\usepackage{array}
\newcolumntype{C}[1]{>{\centering\arraybackslash}p{#1}}
\usepackage{natbib}

\usepackage{tikz}
\usetikzlibrary{arrows.meta,calc,backgrounds}
\usepackage{pgfplots}

\newtheorem{theorem}{Theorem}
\newtheorem{lemma}[theorem]{Lemma}

\newcommand{\tr}{\mathrm{tr}}
\newcommand{\iid}{\textrm{i.i.d.\ }}

\usepackage{tikz}
\usetikzlibrary{arrows.meta,calc,backgrounds}
\tikzset{
    right angle quadrant/.code={
    \pgfmathsetmacro\quadranta{{1,1,-1,-1}[#1-1]}     % Arrays for selecting quadrant
    \pgfmathsetmacro\quadrantb{{1,-1,-1,1}[#1-1]}
    },
    right angle quadrant=1, % Make sure it is set, even if not called explicitly
    right angle length/.code={\def\rightanglelength{#1}},   % Length of symbol
    right angle length=2ex, % Make sure it is set...
    right angle symbol/.style n args={3}{
        insert path={
            let \p0 = ($(#1)!(#3)!(#2)$) in     % Intersection
                let \p1 = ($(\p0)!\quadranta*\rightanglelength!(#3)$), % Point on base line
                \p2 = ($(\p0)!\quadrantb*\rightanglelength!(#2)$) in % Point on perpendicular line
                let \p3 = ($(\p1)+(\p2)-(\p0)$) in  % Corner point of symbol
            (\p1) -- (\p3) -- (\p2)
        }
    }
}

\providecommand{\keywords}[1]{\textbf{\textit{Keywords: }} #1}

\usepackage{authblk}
\usepackage{hyperref}

\title{Adaptive Smoothing Spline for Trajectory Reconstruction}

\author[$1$]{Zhanglong Cao}
\author[$1$]{David Bryant}
\author[$2$]{Tim Molteno}
\author[$2$]{Colin Fox}
\author[$1$]{Matthew Parry}

\affil[$1$]{Department of Mathematics \& Statistics\\ University of Otago\\ Dunedin 9016\\ New Zealand}
\affil[$2$]{Department of Physics\\ University of Otago\\ Dunedin 9016\\ New Zealand}

\date{}  
\begin{document}
\maketitle

\begin{abstract}
Trajectory reconstruction is the process of inferring the path of a moving object between successive observations. In this paper, we propose a smoothing spline -- which we name the V-spline -- that incorporates position and velocity information and a penalty term that controls acceleration. We introduce a particular adaptive V-spline designed to control the impact of irregularly sampled observations and noisy velocity measurements. A cross-validation scheme for estimating the V-spline parameters is given and we detail the performance of the V-spline on four particularly challenging test datasets. Finally, an application of the V-spline to vehicle trajectory reconstruction in two dimensions is given, in which the penalty term is allowed to further depend on known operational characteristics of the vehicle.

\end{abstract}


\section{Introduction}

Global Positioning System (GPS) devices are widely used to track vehicles and other moving objects. The devices generate a time series of noisy measurements of position and velocity that can then be used in batch and on-line reconstruction of trajectories. 

GPS receivers are used to obtain trajectory information for a wide variety of reasons. The \textit{TracMap} company, located in New Zealand and USA, produces GPS display units to aid precision farming in agriculture, horticulture and viticulture. With these units, operational data are collected and sent to a remote server for further analysis. GPS units also guide drivers of farm vehicles to locations on the farm that require specific attention.

Given a sequence of position vectors in a tracking system, the simplest way of constructing the complete trajectory of a moving object is by connecting positions with a sequence of lines, known as line-based trajectory representation \citep{agarwal2003indexing}. Vehicles with an omnidirectional drive or a differential drive can actually follow such a path in a drive-and-turn fashion, though it is highly inefficient \citep{gloderer2010spline}. This kind of non-smooth motion can cause slippage and over-actuation \citep{magid2006spline}. By contrast, most vehicles typically follow smooth trajectories without sharp turns. 

Several methods have been investigated to incorporate smoothness constraints into trajectory reconstruction. One approach uses the minimal length path that is continuously differentiable and consists of line segments or arcs of circles, with no more than three segments or arcs between successive positions \citep{dubins1957curves}. This method is called Dubins curve and has been extended to other more complex vehicle models but is still limited to line segments and arcs of circles \citep{yang2010analytical}. However, there are still curvature discontinuities at the junctions between lines and arcs, leading to yaw angle errors \citep{wang2017curvature}. 

Spline methods have been developed to construct smooth trajectories. \cite{magid2006spline} propose a path-planning algorithm based on splines. The main objective of the method is the smoothness of the path, not a shortest or minimum-time path. \cite{yu2004curve} give a piecewise cubic reconstruction found by matching the observed position and velocity at the endpoints of each interval; this is essentially a Hermite spline. More generally, a B-spline gives a closed-form expression for the trajectory with continuous second derivatives and goes through the position points smoothly while ignoring outliers \citep{komoriya1989trajectory, ben2004geometric}. B-splines are flexible and have minimal support with respect to a given degree, smoothness, and domain partition. \cite{gasparetto2007new} use fifth-order B-splines to model the overall trajectory, allowing one to set kinematic constraints on the motion, expressed as the velocity, acceleration, and jerk. In the context of computer (or computerized) numerical control (CNC), \cite{erkorkmaz2001high} presented a quintic spline trajectory generation algorithm connecting a series of reference knots that produces continuous position, velocity, and acceleration profiles. \cite{yang2010analytical} proposed an efficient and analytical continuous curvature path-smoothing algorithm based on parametric cubic B\'{e}zier curves that can fit sequentially ordered points. 

A parametric approach only captures features contained in the preconceived class of functions \citep{yao2005functional} and increases model bias. To avoid this problem, nonparametric methods, such as smoothing splines, have been developed \citep{craven1978smoothing}. Let $\{t_i\}_{i=1}^n$ be a sequence of observation times in the interval $[a,b]$ satisfying $a \leq t_1<t_2< \cdots < t_n \leq b$, and let $\mathbf{y}=\left\lbrace y_1,\ldots,y_n\right\rbrace$ be the associated position observations. Smoothing spline estimates of $f(t)$ appear as a solution to the following minimization problem: find $\hat{f} \in \mathcal{C}^{(2)}[a,b]$ that minimizes the penalized residual sum of squares,
\begin{equation}\label{smoothingob}
\mbox{RSS}=\sum_{i=1}^{n}\left(  y_i-f(t_i)\right)^2+\lambda\int_{a}^{b} \left(f''(t)\right)^2dt
\end{equation}
for a pre-specified value $\lambda>0$ \citep{aydin2012smoothing}. In equation  \eqref{smoothingob}, the first term is a residual sum of squares and penalizes lack of fit. The second term is a roughness penalty term where the smoothing parameter $\lambda$ varies from 0 to $+\infty$. (The roughness penalty term is a formalization of a mechanical device: if a thin piece of flexible wood, called a spline, is bent to the shape of the curve $f$, then the leading term in the strain energy is proportional to $\int f''(t)^2dt$ \citep{green1993nonparametric}.) The reconstruction cost, equation \eqref{smoothingob}, is determined not only by its goodness-of-fit to the data quantified by the residual sum of squares but also by its roughness \citep{schwarz2012geodesy}. For a given $\lambda$, minimizing equation \eqref{smoothingob} will give the best compromise between smoothness and goodness-of-fit. When $\lambda=0$ the reconstruction is a smooth interpolation of the observation points; when $\lambda=\infty$ the reconstruction is a straight line. Notice that the first term in equation \eqref{smoothingob} depends only on the values of $f(t)$ at $t_i, i=1, \ldots, n$. \cite{green1993nonparametric} show that the function that minimizes the objective function for fixed values of $f(t_i)$ is a cubic spline: an interpolation of points via a continuous piecewise cubic function, with continuous first and second derivatives. The continuity requirements uniquely determine the interpolating spline, except at the boundaries \citep{sealfon2005smoothing}.

\cite{zhang2013cubic} propose Hermite interpolation on each interval to fit position, velocity and acceleration with kinematic constraints. Their default trajectory formulation is a combination of several cubic splines on every interval or, alternatively, is a single function found by minimizing 
\begin{equation}
p\sum_{i=1}^{n}\lvert y_i-f(t_i) \rvert^2+(1-p)\int_a^b \lvert f''(t) \rvert^2dt,
\end{equation}
where $p$ is a smoothing parameter \citep{castro2006geometric}.

A conventional smoothing spline is controlled by one single parameter, which controls the smoothness of the spline on the whole domain. A natural extension is to allow the smoothing parameter to vary as a function of the independent variable, adapting to the change of roughness in different domains \citep{silverman1985some, donoho1995wavelet}. The objective function is now of the form 
\begin{equation}\label{objective}
\sum_{i=1}^{n}\left(y_i-f(t_i) \right)^2+\int_a^b\lambda(t) \left( f''(t)\right)^2 dt.
\end{equation}

Similar to the conventional smoothing spline problem, one has to choose the penalty function $\lambda(t)$. The fundamental idea of nonparametric smoothing is to let the data choose the amount of smoothness, which consequently decides the model complexity \citep{gu1998model}. When $\lambda$ is constant, most methods focus on data-driven criteria, such as cross-validation (CV), generalized cross-validation (GCV) \citep{craven1978smoothing} and generalized maximum likelihood (GML) \citep{wahba1985comparison}. Allowing the smoothing parameter to be a function poses additional challenges, though \cite{liu2010data} were able to extend GML to adaptive smoothing splines.

In this paper, we propose a smoothing spline -- which we name the V-spline -- that is obtained from noisy paired position data $\mathbf{y}=\left\lbrace y_1,\ldots,y_n\right\rbrace$ and velocity data $\mathbf{v}=\left\lbrace v_1,\ldots,v_n\right\rbrace$. In Section \ref{SectionTractorSpline}, the objective function for the V-spline is introduced that depends on velocity residuals $v_i-f'(t_i)$, as well as position residuals $y_i-f(t_i)$, and a parameter $\gamma$ that controls the degree to which the velocity information is used in the reconstruction. We show that the V-spline can be written in terms of modified Hermite spline basis functions. We also introduce a particular adaptive V-spline that seeks to control the impact of irregularly sampled observations and noisy velocity measurements. A cross-validation scheme for estimating the V-spline parameters is given in section \ref{sect:CV}. Section \ref{sect:sim} details the performance of the V-spline on simulated data based on the \textit{Blocks}, \textit{Bumps}, \textit{HeaviSine} and \textit{Doppler} test signals \citep{donoho1994ideal}. Finally, an application of the V-spline to a real 2-dimensional dataset is given in section \ref{splineapplication}.

\section{The V-spline}\label{SectionTractorSpline}

We consider the situation of paired position data $\mathbf{y}=\left\lbrace y_1,\ldots,y_n\right\rbrace$ and velocity data $\mathbf{v}=\left\lbrace v_1,\ldots,v_n\right\rbrace$ at a sequence of times satisfying $a \leq t_1<t_2< \cdots < t_n \leq b$. For $f\in {\mathcal C}^{(2)}_{\text{piecewise}}[a,b]$, we define the objective function 
\begin{equation}\label{tractorsplineObjective}
J[f]= \frac{1}{n} \sum_{i=1}^{n} \left( y_i-f(t_i) \right)^2 + \frac{\gamma}{n} \sum_{i=1}^{n} \left( v_i-f'(t_i) \right)^2 +\sum_{i=0}^{n}\lambda_i \int_{t_i}^{t_{i+1}}  \left(f''(t)\right)^2 dt,
\end{equation}
where $t_0:=a$, $t_{n+1}:=b$, $\gamma>0$, and we have chosen the penalty function $\lambda(t)$ to be a piecewise constant function, i.e for $i=0,\ldots, n$,
\begin{equation}\label{penalty}
\lambda(t)=\lambda_i,\qquad t\in[t_i,t_{i+1}).
\end{equation} 
From now on, we will understand $\lambda(t)$ to be given by eq. (\ref{penalty}) and we will often use $\lambda$ to refer to the set of $\lambda_i$.

\begin{theorem}\label{TractorSplineTheorem}
For $n\geq2$, the objective function $J[f]$ is uniquely minimized by a cubic spline, piecewise on the intervals $[t_i,t_{i+1})$, $i=1,\ldots,n-1$, and linear on $[a,t_1]$ and $[t_n,b]$.
\end{theorem}

We term the minimizer of \eqref{tractorsplineObjective} the \textit{V-spline} because it incorporates velocity information and because of its application to vehicle and vessel tracking. The proof of Theorem \ref{TractorSplineTheorem} is in Appendix \ref{AppendixTractorSplineProof}. 

{\em Remark:} In the language of splines, the points $t_1,\ldots,t_n$ are the interior knots of the V-spline, and $t_0, t_{n+1}$ are the exterior or boundary knots. Since the V-spline is linear on $[a,t_1]$ and $[t_n,b]$, without loss of generality we let $t_1=a$ and $t_n=b$ from now on.

\subsection{Constructing basis functions}

It is convenient to construct basis functions from cubic Hermite splines \citep{hintzen2010improved}. If $y_0$ and $v_0$ are the position and velocity at time $0$, and $y_1$ and $v_1$ are the position and velocity at time $1$, then for $s\in[0,1]$, the cubic Hermite spline is defined as
\begin{equation*}
f(s)=\left(2s^3-3s^2+1\right)y_0+\left(s^3-2s^2+s\right)v_0+\left(-2s^3+3s^2\right)y_1+\left(s^3-s^2\right)v_1.
\end{equation*}
This implies that for an arbitrary interval $[t_i, t_{i+1})$, the relevant cubic spline basis functions are
\begin{align}\label{hermitebasis1}
& h_{00}^{\left(i\right)}\left(t\right)=
\begin{cases}
2\left(\frac{t-t_{i}}{t_{i+1}-t_{i}}\right)^3-3\left(\frac{t-t_{i}}{t_{i+1}-t_{i}}\right)^2+1 & t_i\leq t<t_{i+1} \\ 
0 & \mbox{otherwise}
\end{cases}, \\
& h_{10}^{\left(i\right)}\left(t\right)=\begin{cases}
\frac{\left(t-t_{i}\right)^3}{\left(t_{i+1}-t_{i}\right)^2}-2\frac{\left(t-t_{i}\right)^2}{t_{i+1}-t_{i}}+\left(t-t_{i}\right)  \hphantom{.}  & t_i\leq t<t_{i+1} \\ 
0 &   \mbox{otherwise}
\end{cases},\\
& h_{01}^{\left(i\right)}\left(t\right)=
\begin{cases}
-2\left(\frac{t-t_i}{t_{i+1}-t_i}\right)^3+3\left(\frac{t-t_i}{t_{i+1}-t_i}\right)^2 \hphantom{+} & t_i\leq t<t_{i+1} \\ 
0 &   \mbox{otherwise}
\end{cases},\\
& h_{11}^{\left(i\right)}\left(t\right)=\begin{cases}
\frac{\left(t-t_i\right)^3}{\left(t_{i+1}-t_i\right)^2}-\frac{\left(t-t_i\right)^2}{t_{i+1}-t_i}  \hphantom{+t-+123-}  & t_i\leq t<t_{i+1} \\ 
0 &   \mbox{otherwise}
\end{cases}.
\end{align}
Consequently, the cubic Hermite spline $f^{(i)}(t)$ on an arbitrary interval $[t_i,t_{i+1})$ with two successive points $\{y_i,v_i\}$ and $\{y_{i+1},v_{i+1}\}$ is expressed as
\begin{equation}\label{cubicHermitesplineform}
f^{(i)}(t)=h_{00}^{(i)}(t)y_i+h_{10}^{(i)}(t)  v_i+h_{01}^{(i)}(t)y_{i+1} +h_{11}^{(i)}(t) v_{i+1}.
\end{equation}

For V-splines, a slightly more convenient basis is given by $\{N_k(t)\}_{k=1}^{2n}$, where $N_1(t) = h^{(1)}_{00}(t)$, $N_2(t) = h^{(1)}_{10}(t)$, and for all $i=1,2,\ldots,n-2$, 
\begin{align*}
N_{2i+1}(t)&=h_{01}^{(i)}(t)+h_{00}^{(i+1)}(t), \\
N_{2i+2}(t)&= h_{11}^{(i)}(t)+h_{10}^{(i+1)}(t),
\end{align*}
and
\begin{align*}
N_{2n-1}(t) &= 
\begin{cases}
h_{01}^{(n-1)}(t) & \mbox{if $t<t_n$}\\ 
1 & \mbox{if $t=t_n$}
\end{cases},\\
N_{2n}(t) &= h_{11}^{(n-1)}(t).
\end{align*}
Any $f\in{\mathcal C}^{(2)}_{\text{piecewise}}[a,b]$ can then be represented in the form
\begin{equation}
f(t)=\sum_{k=1}^{2n}N_k(t)\theta_k,
\end{equation}
where $\left\lbrace \theta_k\right\rbrace_{k=1}^{2n}$ are parameters.

\subsection{Computing the V-spline}

In terms of the basis functions in the previous section, the objective function \eqref{tractorsplineObjective} is given by
\begin{equation}\label{tractormse}
nJ[f](\theta, \lambda,\gamma) = \left(\mathbf{y}-B\theta\right)^\top \left(\mathbf{y}-B\theta\right) +\gamma \left(\mathbf{v}-C\theta\right)^\top \left(\mathbf{v}-C\theta\right)+n \theta^\top\Omega_{\lambda}\theta,
\end{equation}
where $B$ and $C$ are $n\times 2n$ matrices with components
\begin{align}
[B]_{ij}&=N_j(t_i)=\begin{cases}
1, & j=2i-1\\
0, & \mbox{otherwise}
\end{cases}\\
[C]_{ij}&=N'_j(t_i)=\begin{cases}
1, & j=2i\\
0, & \mbox{otherwise}
\end{cases}
\end{align}
and $\Omega{_\lambda}$ is a $2n\times 2n$ matrix with components $[\Omega_{\lambda}]_{jk}=\int_{a}^{b}\lambda(t) N''_j(t)N''_k(t)dt$. In the following, we reserve the use of boldface for $n\times 1$ vectors and $n\times n$ matrices.

The detailed structure of $\Omega{_\lambda}$ is considered in Appendix \ref{PenaltyTermDetails}. It is convenient to write $\Omega_{\lambda}=\sum_{i=1}^{n-1}\lambda_i\Omega^{(i)}$, where $[\Omega^{(i)}]_{jk}=\int_{t_i}^{t_{i+1}} N''_j(t)N''_k(t)dt$. It is then evident that $\Omega{_\lambda}$ is a bandwidth four matrix.

Since equation \eqref{tractormse} is a quadratic form in terms of $\theta$, it is straightforward to establish that the objective function is minimized at
%\begin{equation}
%\left(B^\top B+\gamma C^\top C+n\Omega_{\lambda}\right)\hat{\theta}=\left(B^\top\mathbf{y}+\gamma C^\top\mathbf{v}\right).
%\end{equation}
%Therefore, the solution is 
\begin{equation}
\hat{\theta}=\left(B^\top B+\gamma C^\top C+n\Omega_{\lambda}\right)^{-1}\left(B^\top\mathbf{y}+\gamma C^\top\mathbf{v}\right),\label{thetahat}
\end{equation}
which can be identified as a generalized ridge regression. The fitted V-spline is then given by
$\hat{f}(t)=\sum_{k=1}^{2n}N_k(t)\hat{\theta}_k$. 

The V-spline is an example of a linear smoother \citep{esl2009}. This is because the estimated parameters in equation \eqref{thetahat} are a linear combination of $\mathbf{y}$ and $\mathbf{v}$. Denoting by $\hat{\mathbf{f}}$ and $\hat{\mathbf{f}'}$ the vector of fitted values $\hat{f}(t_i)$ and $\hat{f'}(t_i)$ at the training points $t_i$, we have
\begin{equation}\label{eq:fittedy}
%\begin{split}
\hat{\mathbf{f}} =B\left(B^\top B+\gamma C^\top C+n\Omega_{\lambda}\right)^{-1}\left(B^\top\mathbf{y}+\gamma C^\top\mathbf{v}\right)=: \mathbf{S}_{\lambda,\gamma}\mathbf{y}+\gamma\mathbf{T}_{\lambda,\gamma}\mathbf{v} 
%\end{split}
\end{equation}
\begin{equation}\label{eq:fittedv}
%\begin{split}
\hat{\mathbf{f}'}
=C\left(B^\top B+\gamma C^\top C+n\Omega_{\lambda}\right)^{-1}\left(B^\top\mathbf{y}+\gamma C^\top\mathbf{v}\right)=:\mathbf{U}_{\lambda,\gamma}\mathbf{y}+\gamma\mathbf{V}_{\lambda,\gamma}\mathbf{v}
%\end{split}
\end{equation}
where $\mathbf{S}_{\lambda,\gamma}, \mathbf{T}_{\lambda,\gamma}, \mathbf{U}_{\lambda,\gamma}$ and $\mathbf{V}_{\lambda,\gamma}$ are smoother matrices that depend only on $t_i,\lambda(t)$ and $\gamma$. It is not hard to show that $\mathbf{S}_{\lambda,\gamma}$ and $\mathbf{V}_{\lambda,\gamma}$ are symmetric, positive semidefinite matrices. Note that  $\mathbf{T}_{\lambda,\gamma}= \mathbf{U}_{\lambda,\gamma}^\top$.
%Additionally, by Cholesky decomposition,
%\begin{equation}
%\left(B^\top B+\gamma C^\top C+n\Omega_{\lambda}\right)^{-1}=\mathbf{R}\mathbf{R}^\top,
%\end{equation}
%hence $\mathbf{T}_{\lambda,\gamma}=B\mathbf{R}\mathbf{R}^\top C^\top$ and $\mathbf{U}_{\lambda,\gamma}=C\mathbf{R}\mathbf{R}^\top B^\top$, so that 
% $\mathbf{T}_{\lambda,\gamma}= \mathbf{U}_{\lambda,\gamma}^\top$. 

\begin{theorem}\label{TractorsplineCorollary}
If $f(t)$ is a V-spline then, for almost all $\mathbf{y}$ and $\mathbf{v}$, $f''(t)$ is continuous at the knots if and only if $\gamma=0$ and $\lambda_i=\lambda_0$, for all $i=1,\ldots,n-1$.
\end{theorem}

The proof of Theorem \ref{TractorsplineCorollary} is in Appendix \ref{proofofCorollary}. The V-spline with $\gamma=0$ and $\lambda_i=\lambda_0$ is simply the conventional smoothing spline.

\subsection{The adaptive V-spline}

Until now, we have not explicitly considered the impact of irregularly sampled observations or of noisy measurements of velocity on trajectory reconstruction. In order to do this, it is instructive to evaluate the contribution to the  penalty term from the interval $[t_i,t_{i+1})$. Using \eqref{cubicHermitesplineform}, it is relatively straightforward to show that
\begin{equation}
f''(t)=\frac{1}{t_{i+1} -t_i}\left\{ 6\left(\varepsilon^{+}_i+\varepsilon^{-}_{i}\right)\frac{t-t_i}{t_{i+1}-t_i}-2\left(2\varepsilon^{+}_i+\varepsilon^{-}_{i}\right)\right\},
\end{equation}
where $\varepsilon^+_i=v_i-\bar{v}_i$, $\varepsilon^-_{i}=v_{i+1}-\bar{v}_i$ and $\bar{v}_i=(y_{i+1}-y_i)/(t_{i+1}-t_i)$ is the average velocity over the interval. The $\varepsilon^{\pm}_i$ can be interpreted as the difference at time $t_i$ and $t^-_{i+1}$ respectively between the velocity implied by an interpolating Hermite spline and the velocity implied by a straight line reconstruction.

The contribution to the penalty term is then
\begin{equation}\label{pen_contrib}
4\lambda_i\frac{(\varepsilon^+_i)^2+\varepsilon^+_i\varepsilon^-_{i}+(\varepsilon^{-}_{i})^2}{\Delta T_i},
\end{equation}
where $\Delta T_i=t_{i+1}-t_i$. We call the quantity $(\varepsilon^+_i)^2+\varepsilon^+_i\varepsilon^-_{i}+(\varepsilon^{-}_{i})^2$, the square of the {\em discrepancy} of the velocity on the interval $[t_i,t_{i+1})$.% Note that the discrepancy is an irreducible component of the penalty. 

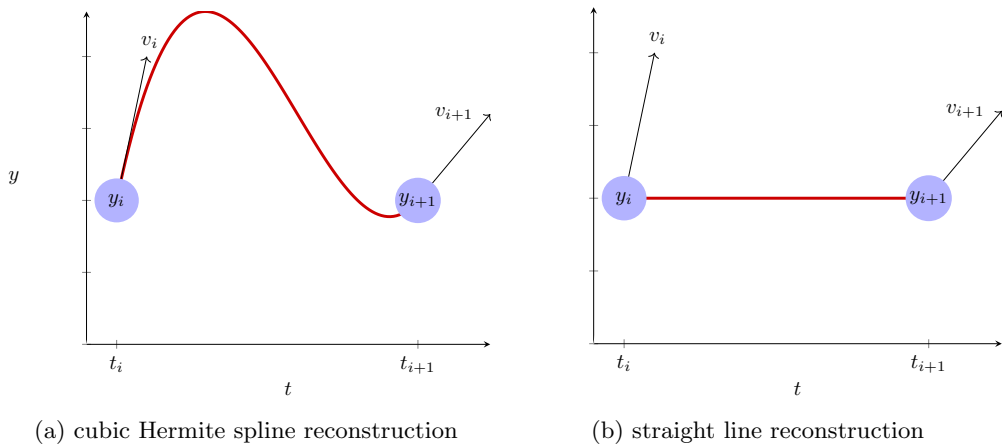
\begin{figure}[h]
\centering
\begin{subfigure}[b]{0.45\textwidth}
\centering
\resizebox{\linewidth}{!}{
	  \begin{tikzpicture}
	    \begin{axis}[
	    	ylabel style={rotate=-90},
	        axis x line=bottom,
	        axis y line=left,
	        xlabel={$t$},
	        ylabel={$y$},
	        domain=-0.1:1.11,
	        xtick={0, ..., 1},
	        samples=100,
	        xticklabels={$t_i$,$t_{i+1}$},
	        yticklabels=\empty,
	    ]
       \addplot[line width=0.5mm, red!80!black][domain=0:1]{1.25*pow(x,3)-2.25*pow(x,2)+x};
       %\addplot [black] [domain=-0.1:0]{1.25*pow(x,3)-2.25*pow(x,2)+x};
       \addplot [white] [domain=-0.1:0]{x};
	   \addplot[->] coordinates {(0,0) (0.1,0.1)};
	   \addplot[->] coordinates {(1,0) (1.24,0.06)};
	   \node[circle,fill=blue!30,inner sep=2pt] at (axis cs:0,0) {$\mbox{ }y_i\mbox{ }$};
	   \node[circle,fill=blue!30,inner sep=1pt] at (axis cs:1,0) {$y_{i+1}$};
	   \node[] at (axis cs:0.11,0.11) {$v_i$};
	   \node[] at (axis cs:1.12,0.06) {$v_{i+1}$};
	   \end{axis}
	  \end{tikzpicture}
      }
      \caption{cubic Hermite spline reconstruction}\label{subhermite}
\end{subfigure}
\begin{subfigure}[b]{0.45\textwidth}
       \centering
       \resizebox{\linewidth}{!}{
       \begin{tikzpicture}
       	    \begin{axis}[
       	        axis x line=bottom,
       	        axis y line=left,
       	        xlabel={$t$},
       	        ylabel=\empty,
       	        domain=-0.1:1.11,
       	        xtick={0, ..., 1},
       	        samples=100,
       	        xticklabels={$t_i$,$t_{i+1}$},
       	        yticklabels=\empty,
       	      ]
       	    \addplot [line width=0.1mm, white][domain=0:1]{1.25*pow(x,3)-2.25*pow(x,2)+x};
       	    \addplot [white] [domain=-0.1:0]{x};
			\addplot [line width=0.5mm, red!80!black][domain=0:1] {0};
	   		\addplot[->] coordinates {(0,0) (0.1,0.1)};
	   		\addplot[->] coordinates {(1,0) (1.24,0.06)};
	   		\node[circle,fill=blue!30,inner sep=2pt] at (axis cs:0,0) {$\mbox{ }y_i\mbox{ }$};
	   		\node[circle,fill=blue!30,inner sep=1pt] at (axis cs:1,0) {$y_{i+1}$};
	   		\node[] at (axis cs:0.11,0.11) {$v_i$};
	   		\node[] at (axis cs:1.12,0.06) {$v_{i+1}$};
       	  \end{axis}
       \end{tikzpicture}
       }
	\caption{straight line reconstruction}
\end{subfigure}
\caption{Comparing cubic Hermite spline reconstruction and straight line reconstruction. When $\Delta T_i=t_{i+1}-t_i$ is large or $\bar{v}_i \Delta T_i=y_{i+1}-y_i$ is small, the adaptive V-spline favours straighter reconstructions.
%On the left side, a genuine cubic Hermite spline is cooperating with noisy velocities. Even though the vehicle is not moving, the reconstruction is following the directions of $P_1$ and $P_2$ and gives a wiggle between the two points. On the right side, it is an expected reconstruction between two not-moving points after a long time gap.
}\label{figureAPT}
\end{figure}

%We address this issue by introducing an adjusted penalty term $\frac{\left(\Delta T_i\right)^3}{\left(\Delta d_i\right)^2}$ found in equation \eqref{APTintroequation} to the penalty function $\lambda(t)$, in which the V-spline is penalized by its real differences of $\Delta d_i$ and $\Delta T_i$ for each interval $[t_i, t_{i+1}]$. With this term, either the measurement in velocity becomes unreliable comparing to average speed or after a long time gap, the adjusted penalty term will works on the penalty function and forces it to return a straight line rather than a curve on this particular domain. From the physical point of view, the term is the reciprocal of the product of velocity and acceleration. Either velocity or acceleration goes to zero, the vehicle should either stop, which returns a straight line through time on $y$ axis, or keep moving with the same speed, which returns a linear interpolation instead of a curved path. 

%Therefore, the final form of the penalty function $\lambda(t,\eta)$ is piecewise constant having the following form on each interval 

As a consequence of \eqref{pen_contrib}, larger time intervals will tend to contribute less to the penalty term (other things being equal). But this is exactly when we would expect the velocity at the endpoints of the interval to provide less useful information about the trajectory over the interval. In the case when the observed change in position is small, i.e. when $y_{i+1}-y_i=\bar{v}_i \Delta T_i\approx 0$, over-reliance on noisy measurements of velocity will result in ``wiggly'' reconstructions. In these two instances -- graphically depicted in Figure \ref{subhermite} -- we would like the V-spline to adapt and to favour straighter reconstructions; this is a deliberate design choice. We can achieve this by choosing
\begin{equation}\label{adjustedpenalty}
\lambda_i=\eta \frac{\Delta T_i}{\bar{v}_i^2},
\end{equation}
where $\eta$ is a parameter to be estimated. The penalty term then takes a particularly compelling form: the contribution from the interval $[t_i,t_{i+1})$ \eqref{pen_contrib} is proportional to
\begin{equation}\label{compelling}
\left(\frac{\mbox{discrepancy in velocity}}{\mbox{average velocity}}\right)^2
\end{equation}
for all $i$. We call the resulting spline the {\em adaptive V-spline}. The spline when $\lambda_i=\lambda_0$, we call the {\em non-adaptive V-spline}.

%Eventually, in the objective function, there are two unknown parameters: $\eta$ controlling the curvature of V-spline on different domains and $\gamma$ controlling the residuals of velocity. The piecewise constant function $\lambda(t,\eta)$ is using a data driven method to model the penalty function in the adaptive V-splines. 

\section{Parameter selection and cross-validation}\label{sect:CV}

The issue of choosing the smoothing parameter is ubiquitous in curve estimation and there are two different philosophical approaches to the problem. The first is to regard the free choice of smoothing parameter as an advantageous feature of the procedure. The second is to let the data determine the parameter \citep{green1993nonparametric}. We prefer the latter and use data to train our model and find the best parameters. The most well-known method for this is cross-validation.

In standard regression, which assumes the mean of the observation errors is zero, the true regression curve $f(t)$ has the property that if an observation $y_i$ is omitted at time point $t_i$, the value $f(t_i)$ is the best predictor of $y_i$ in terms of returning the least value of $\left(y_i-f(t_i)\right)^2$. We use this observation to motivate a leave-one-out cross-validation scheme to estimate $\lambda$ and $\gamma$ for both the non-adaptive and the adaptive V-splines.

%Denote by $\hat{f}^{(-i)}(t,\lambda)$ the estimated function from the remaining data, where $\lambda$ is the smoothing parameter. Then $\hat{f}^{(-i)}\left(t,\lambda\right)$ is the minimizer of  
%\begin{equation}\label{originalcv}
%\frac{1}{n}\sum_{j \neq i}\left(y_j-f(t_j) \right)^2+\lambda\int (f'')^2dt,
%\end{equation}
%and can be quantified by the cross-validation score function
%\begin{equation*}
%\mbox{CV}(\lambda)=\frac{1}{n}\sum_{i=1}^{n}\left(  y_i-\hat{f}^{(-i)}(t_i,\lambda)\right) ^2.
%\end{equation*}
%The basis idea of the cross-validation is to choose the value of $\lambda$ that minimizes $\mbox{CV}(\lambda)$ \citep{green1993nonparametric}. 
%
%%An efficient way to calculate the cross-validation score is introduced by \citep{green1993nonparametric}. 
%Through the equation \eqref{fhy}, it is known that the value of the smoothing spline $\hat{f}$ depends linearly on the data $y_1,\ldots,y_n$. Define the matrix $A(\lambda)$, which is a map vector of observed values $y_i$ to predicted values $\hat{f}(t_i)$. Then we have
%\begin{equation*}\label{crossvalidationmatrixA}
%\hat{\mathbf{f}}=A(\lambda)\mathbf{y}
%\end{equation*}
%and the following lemma.
%\begin{lemma}\citep{green1993nonparametric}\label{cvlema}
%The cross-validation score satisfies
%\begin{equation*}
%\mbox{CV}(\lambda)=\frac{1}{n} \sum_{i=1}^n \left(\frac{y_i-\hat{f}(t_i)}{1-A_{ii}(\lambda)}\right)^2
%\end{equation*}
%where $\hat{f}$ is the spline smoother calculated from the full data set $\left\lbrace (t_i,y_i)\right\rbrace$ with smoothing parameter $\lambda$.
%\end{lemma}

Let $\hat{f}^{(-i)}(t,\lambda,\gamma)$ be the minimizer of
\begin{align}
\frac{1}{n}\sum_{j \neq i}\left( y_j-f(t_j) \right)^2+\frac{\gamma}{n}\sum_{j \neq i} \left( v_j-f'(t_j) \right)^2+ \sum_{i=1}^{n-1}\lambda_i \int_{t_i}^{t_{i+1}} \left( f''(t) \right)^2dt,
\end{align}
and define the cross-validation score
\begin{align}
\mbox{CV}\left(\lambda,\gamma\right)=\frac{1}{n}\sum_{i=1}^{n}\left( y_i-\hat{f}^{(-i)}\left(t_i,\lambda,\gamma\right) \right) ^2.
\end{align}
We then choose $\lambda$ and $\gamma$ that jointly minimize $\mbox{CV}(\lambda,\gamma)$.

%Additionally, it is known that the parameter $\hat{\theta}=\left(B^\top B+\gamma C^\top C+n\Omega_\lambda\right)^{-1}\left(B^\top\mathbf{y}+\gamma C^\top\mathbf{v}\right)$ and will give us the following form \small
%\begin{equation}
%\begin{split}
% \hat{\mathbf{f}}&=B\hat{\theta}=B\left(B^\top B+\gamma C^\top C+n\Omega_\lambda\right)^{-1}B^\top\mathbf{y}+B\left(B^\top B+\gamma C^\top C+n\Omega_\lambda\right)^{-1} C^\top\mathbf{v}\\&=S\mathbf{y}+\gamma T\mathbf{v},
% \end{split}
% \end{equation}
% \begin{equation}
% \begin{split}
%\hat{\mathbf{f}}'&=C\hat{\theta}=C\left(B^\top B+\gamma C^\top C+n\Omega_\lambda\right)^{-1}B^\top\mathbf{y}+C\left(B^\top B+\gamma C^\top C+n\Omega_\lambda\right)^{-1}C^\top \mathbf{v}\\&=U\mathbf{y}+\gamma V\mathbf{v}.
% \end{split}
%\end{equation}\normalsize
%From Lemma \ref{cvlema}, we can prove the following theorem: 
The following theorem establishes that we can compute the cross-validation score without knowing the $\hat{f}^{(-i)}(t,\lambda,\gamma)$:
\begin{theorem}\label{tractorsplinecvscore}
The cross-validation score of a V-spline satisfies
\begin{equation}\label{tractorcv}
\mbox{CV}\left(\lambda,\gamma\right)=\frac{1}{n}\sum_{i=1}^{n} \left( \frac{y_i-\hat{f}(t_i)+\gamma \frac{T_{ii}}{1-\gamma V_{ii}}(v_i-\hat{f}'(t_i))}{1-S_{ii}-\gamma\frac{T_{ii}}{1-\gamma V_{ii}}U_{ii}} \right)^2
\end{equation}
where $\hat{f}$ is the V-spline smoother calculated from the full data set with smoothing parameter $\lambda$ and $\gamma$, and $S_{ii}=[\mathbf{S}_{\lambda,\gamma}]_{ii}$, etc.
\end{theorem}
The proof of Theorem \ref{tractorsplinecvscore} is in Appendix \ref{AppCVscore}.

\section{Simulation study}\label{sect:sim}%and Error Analysis}

In this section, we give an extensive comparison of methods for regularly sampled time series data followed by simulation of irregularly sampled data. The comparison is based on the ability to reconstruct trajectories derived from \textit{Blocks}, \textit{Bumps}, \textit{HeaviSine} and \textit{Doppler}, four functions which were used in \citep{donoho1994ideal, donoho1995adapting, abramovich1998wavelet} to mimic problematic features in imaging, spectroscopy and other types of signal processing. %Notice that the Blocks and Bumps functions have infinite first derivatives, and cannot be inferred by V-splines. Hence, we use these functions, denoted by $g(t)$, as models of velocity rather than position. 

Letting $g(t)$ denote any one of {\it Blocks}, {\it Bumps}, {\it HeavisSine} or {\it Doppler}, we treat $g(t)$ as the velocity of the trajectory $f(t)$, i.e. $f'(t)=g(t)$. Numerically, we calculate the position data via
\begin{equation}\label{generateVelocity}
f(t_{i+1})=f(t_i)+\frac{g(t_i)+g(t_{i+1})}{2}(t_{i+1}-t_i),
\end{equation}
where $f(t_1)=0$. The observed position and velocity are then found by adding \iid zero-mean Gaussian noise:
\begin{align}\label{tractorsplinegeneratefunctions}
\begin{split}
y_i &= f(t_i) + \varepsilon^{(f)}_i, \\
v_i &= g(t_i) + \varepsilon^{(g)}_i,
\end{split}
\end{align}
where $\varepsilon^{(f)}_i\sim N(0,\sigma_f/SNR)$, $\varepsilon^{(g)}_i\sim N(0,\sigma_g/SNR)$, $\sigma_f$ is the standard deviation of the positions $f(t_i)$, $\sigma_g$ is the standard deviation of the velocities $g(t_i)$, and SNR is the signal-to-noise ratio, which we take to be 3 or 7.

\subsection{Regularly sampled time series}

%A set of regularly sampled time series data has equal time difference between each pair of successive points. For example, denoted by $\Delta T_i = t_{i+1}-t_i$ for $i=1,\ldots,n-1$, then $\Delta T_1=\cdots = \Delta T_{n-1}$.

We compare the performance of the adaptive V-spline with two wavelet transform reconstructions \citep{donoho1995adapting, abramovich1998wavelet}, a spatially adaptive penalized spline known as the \textit{P-spline} \citep{krivobokova2008fast, ruppert2003semiparametric}, as well as the adaptive V-spline with $\gamma=0$ and the non-adaptive V-spline. It is important to note that only the non-adaptive and adaptive V-splines incorporate velocity information. For the wavelet transform approach, we use both the \textit{sure} threshold policy and \textit{BayesThresh} with levels $l=4, \ldots, 9$. The V-spline parameters are obtained by minimizing the cross-validation score \eqref{tractorcv}. Following \cite{nason2010wavelet}, we fix $n=1024$ in the simulations.

%In the V-spline, there are two parameters $\eta$ and $\gamma$ to optimize. To evaluate the performance of the velocity term in objective function \eqref{tractorsplineObjective} and the adjusted penalty term in \eqref{adjustedpenalty}, the parameter $\gamma$ is set as 0 in one reconstruction of V-spline, whose objective function and solution become
%\begin{equation}\label{ofgamma0}
%J[f]_{\gamma=0}= \frac{1}{n} \sum_{i=1}^{n} \left(y_i-f(t_i)\right)^2 +\sum_{i=1}^{n-1} \int_{t_i}^{t_{i+1}}\lambda(t) (f''(t))^2 dt,
%\end{equation}
%and
%\begin{equation}\label{thetahat0}
%\hat{\theta}_{\gamma=0}=\left(B^\top B+n\Omega_{\lambda}\right)^{-1}B^\top\mathbf{y}.
%\end{equation}
%In another reconstruction, the adjusted penalty term in \eqref{adjustedpenalty} is removed and the model is denoted by ``V-spline without APT''. 

\subsubsection{Numerical Examples}

\begin{figure}
    \centering
    \begin{subfigure}{0.45\textwidth}
    \centering
    \includegraphics[width=\linewidth,height=0.45\textwidth]{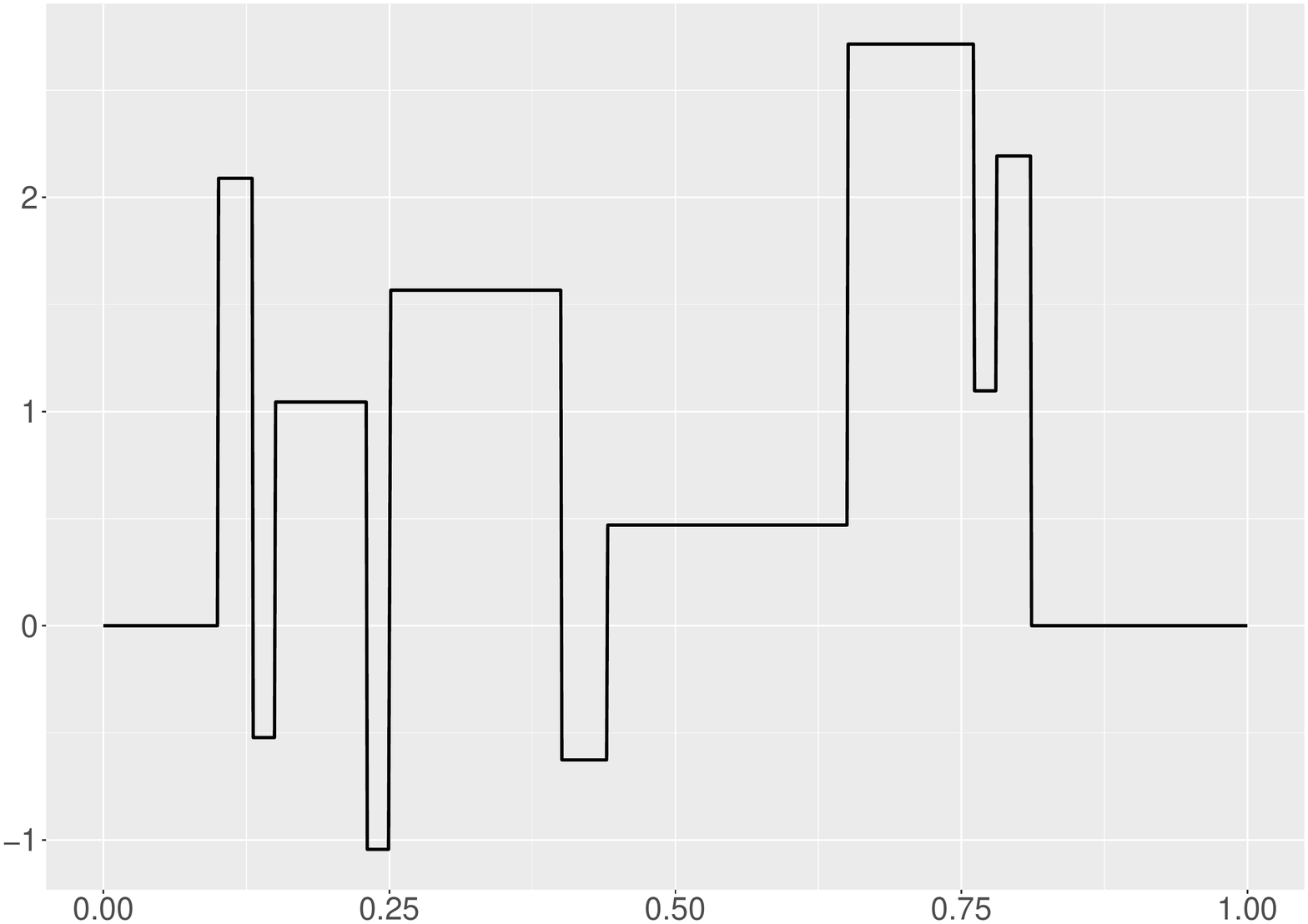}
    \caption{The \textit{Blocks} function}
    \end{subfigure}%
    \begin{subfigure}{0.45\textwidth}
    \centering
    \includegraphics[width=\linewidth,height=0.45\textwidth]{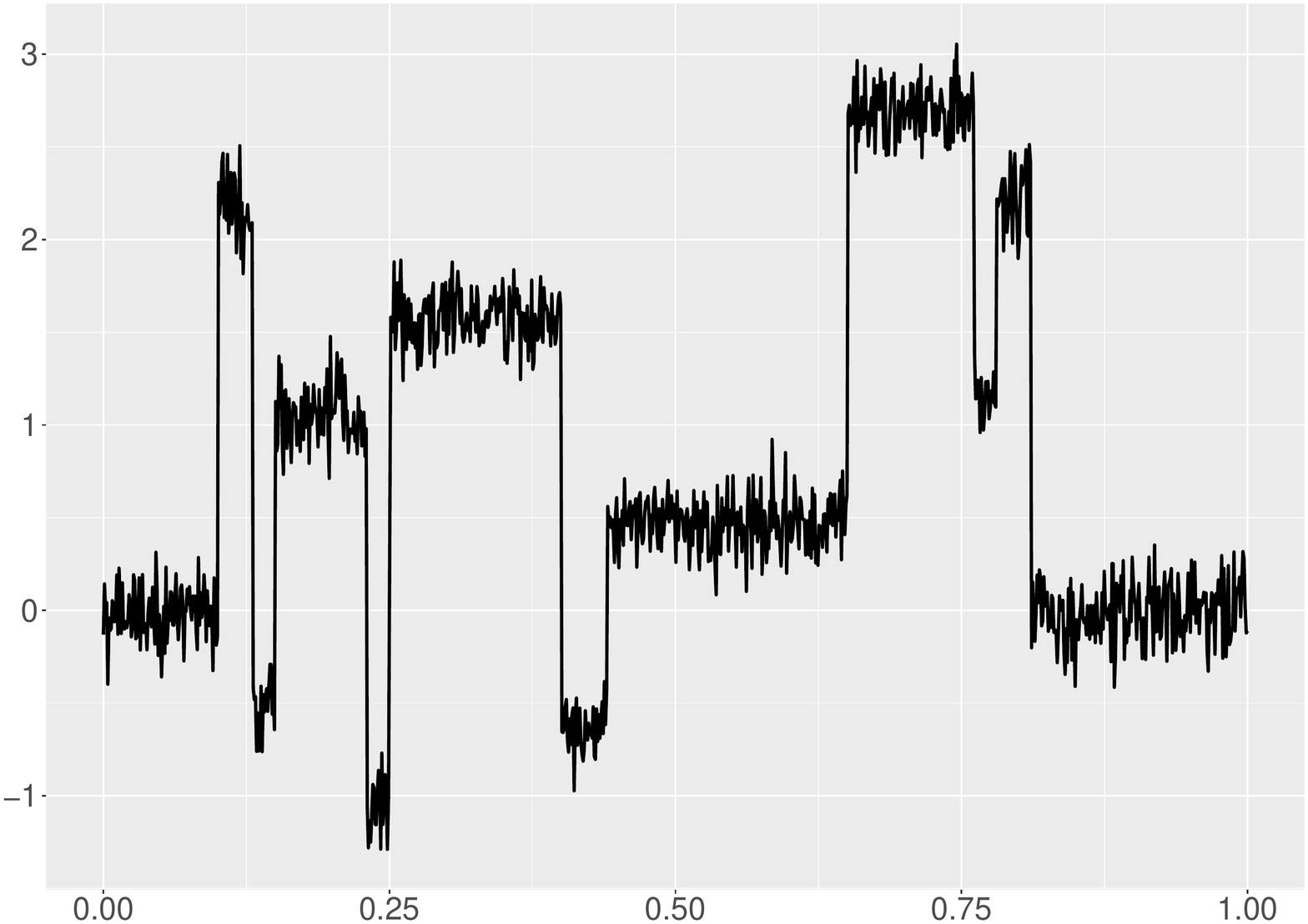}
    \caption{Noisy \textit{Blocks} at \textit{SNR}=7}
    \end{subfigure}
    \begin{subfigure}{0.45\textwidth}
    \centering
    \includegraphics[width=\linewidth,height=0.45\textwidth]{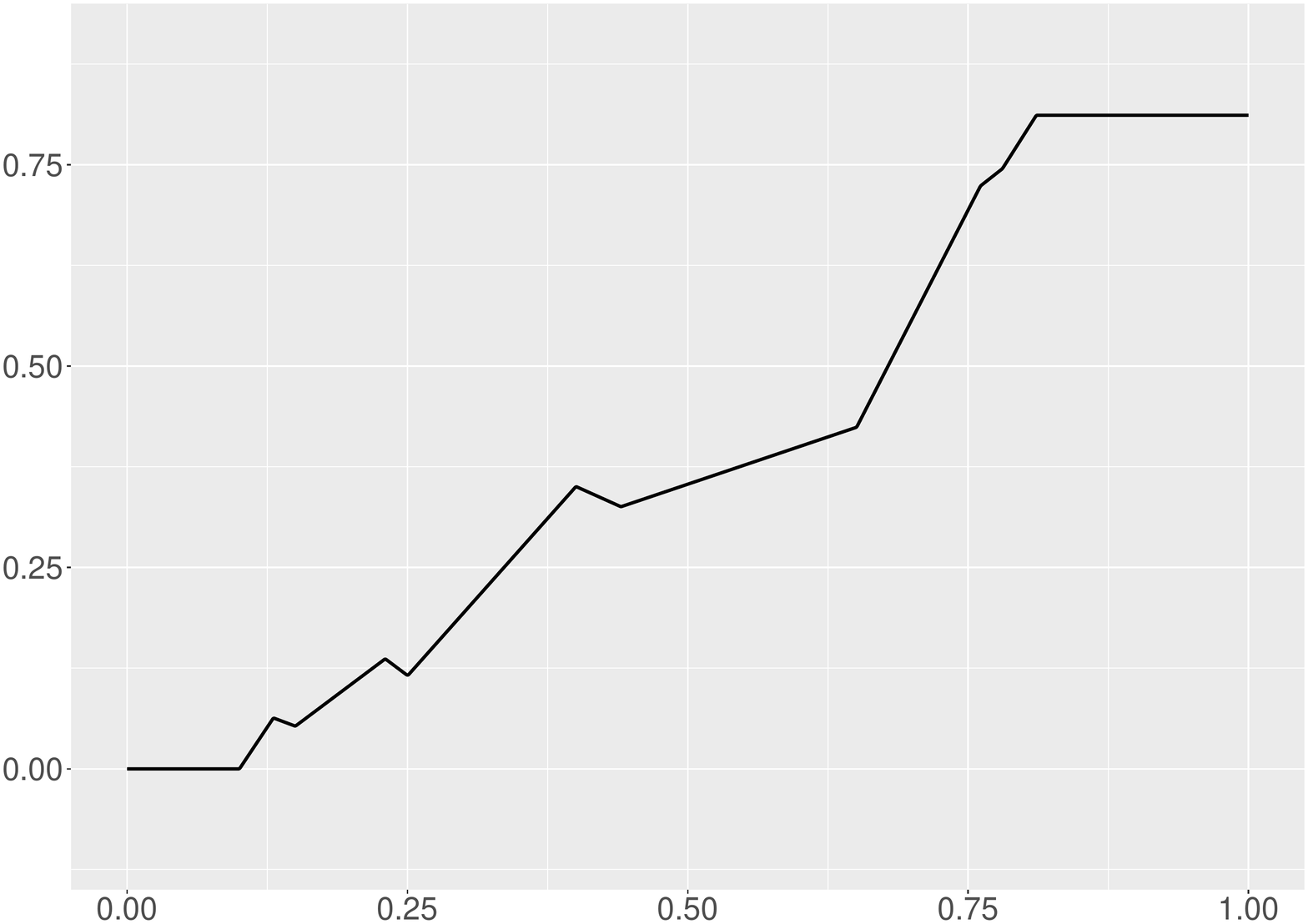}
    \caption{Generated positions}
    \end{subfigure}
    \begin{subfigure}{0.45\textwidth}
    \centering
    \includegraphics[width=\linewidth,height=0.45\textwidth]{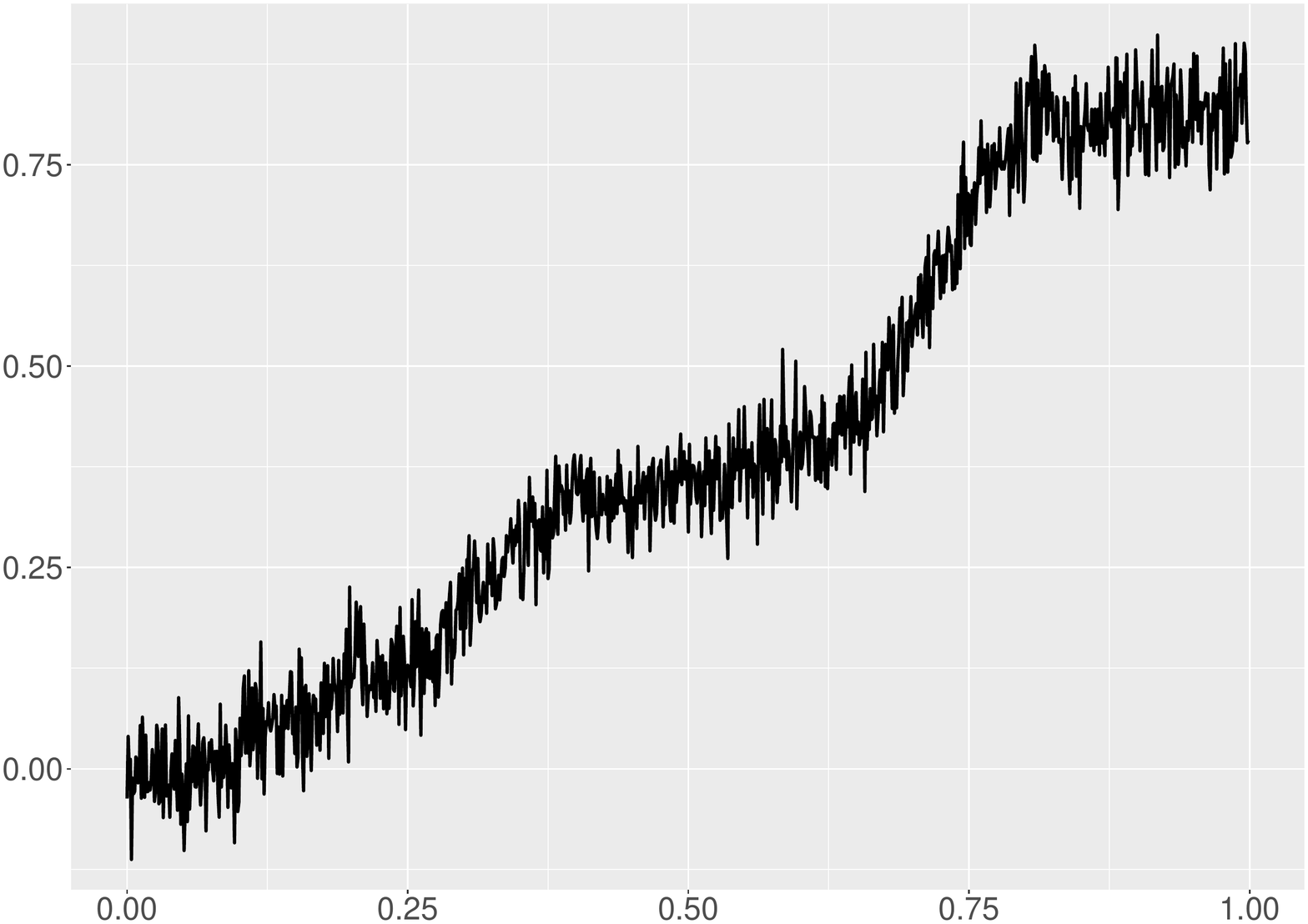}
    \caption{Noisy position at \textit{SNR}=7}
    \end{subfigure}
    \begin{subfigure}{0.45\textwidth}
    \centering
    \includegraphics[width=\linewidth,height=0.45\textwidth]{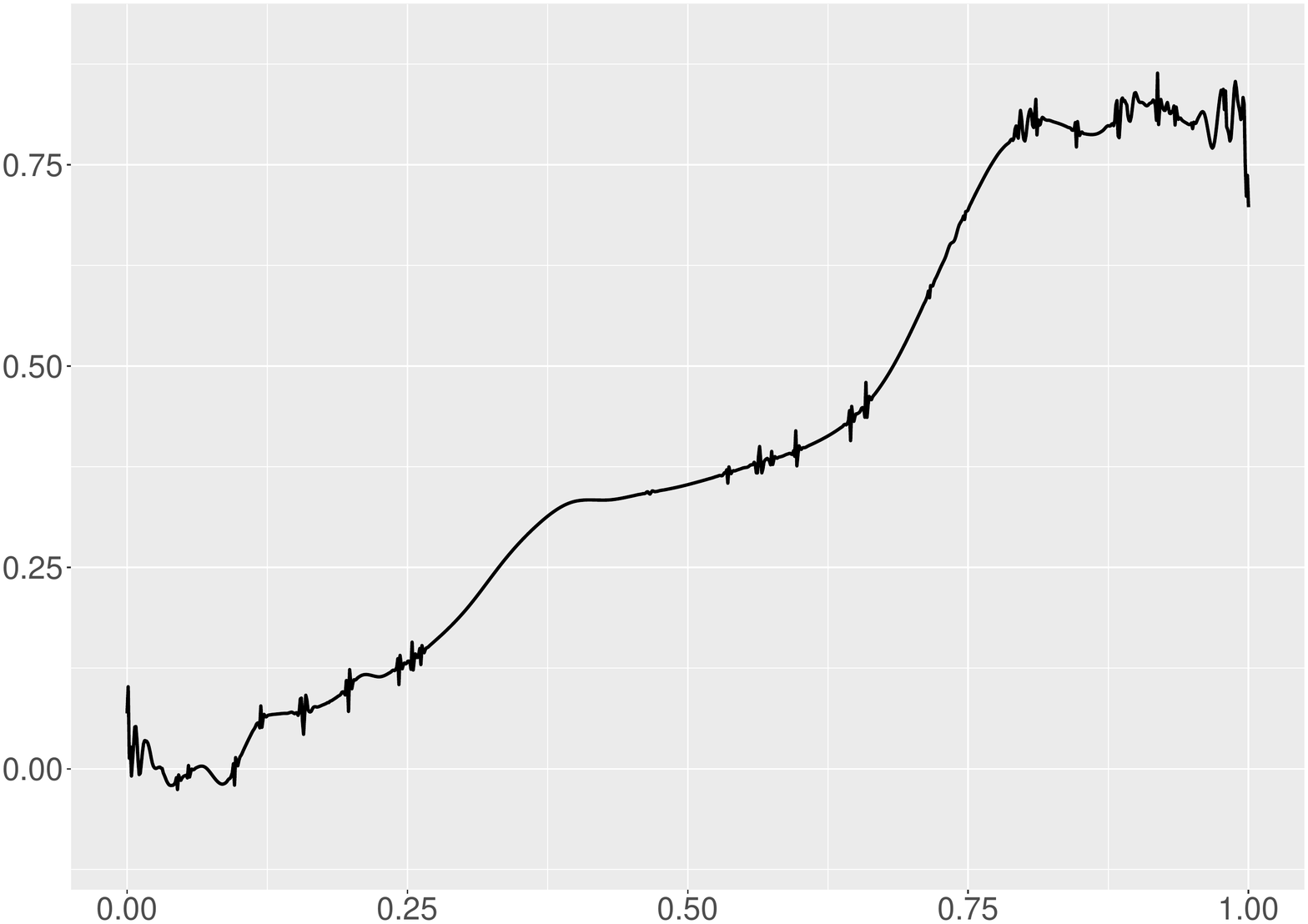}
    \caption{Reconstruction from Wavelet by sure threshold}
    \end{subfigure}
    \begin{subfigure}{0.45\textwidth}
    \centering
    \includegraphics[width=\linewidth,height=0.45\textwidth]{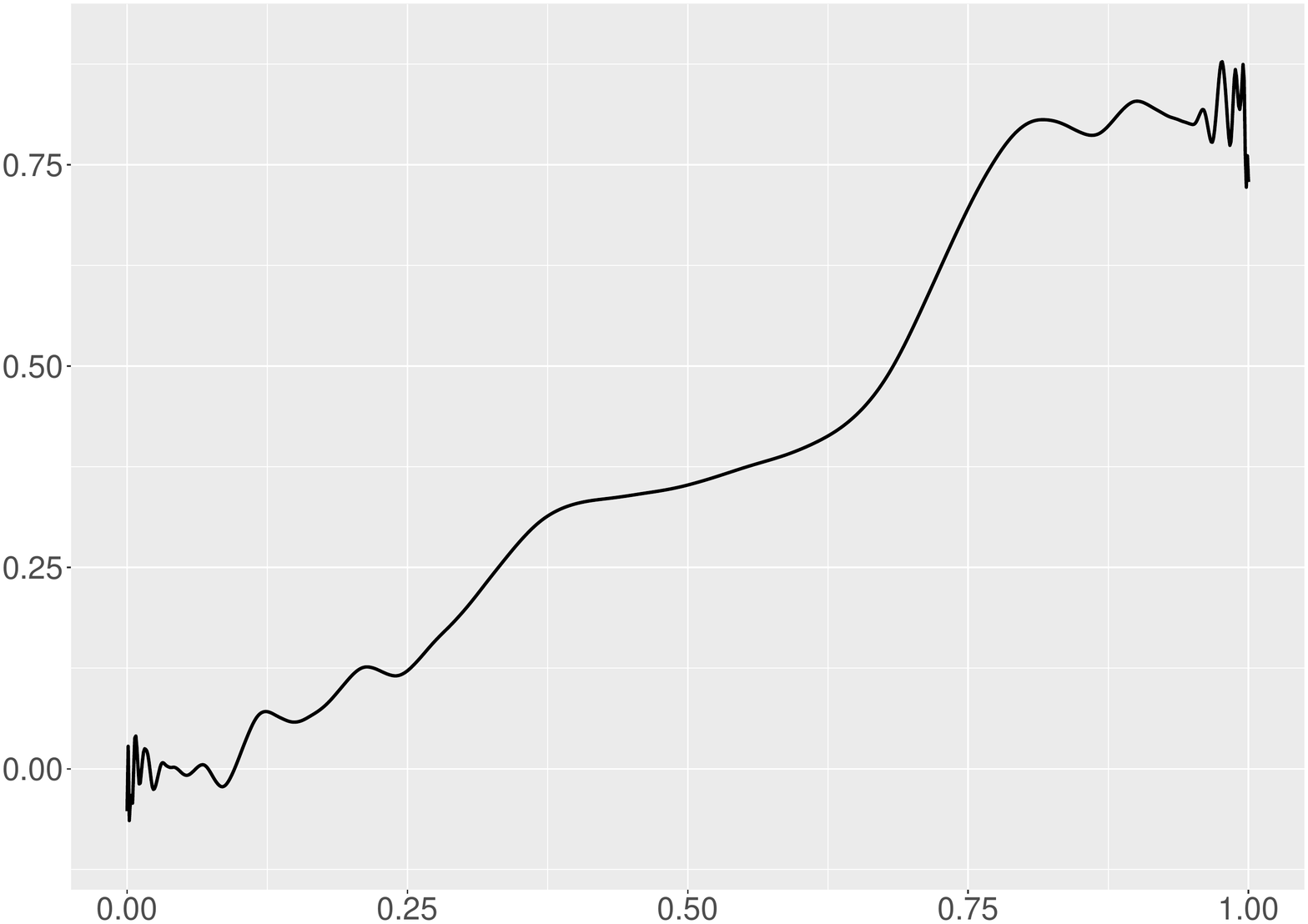}
    \caption{Reconstruction from Wavelet by BayesThresh approach}
    \end{subfigure}
    \begin{subfigure}{0.45\textwidth}
    \centering
    \includegraphics[width=\linewidth,height=0.45\textwidth]{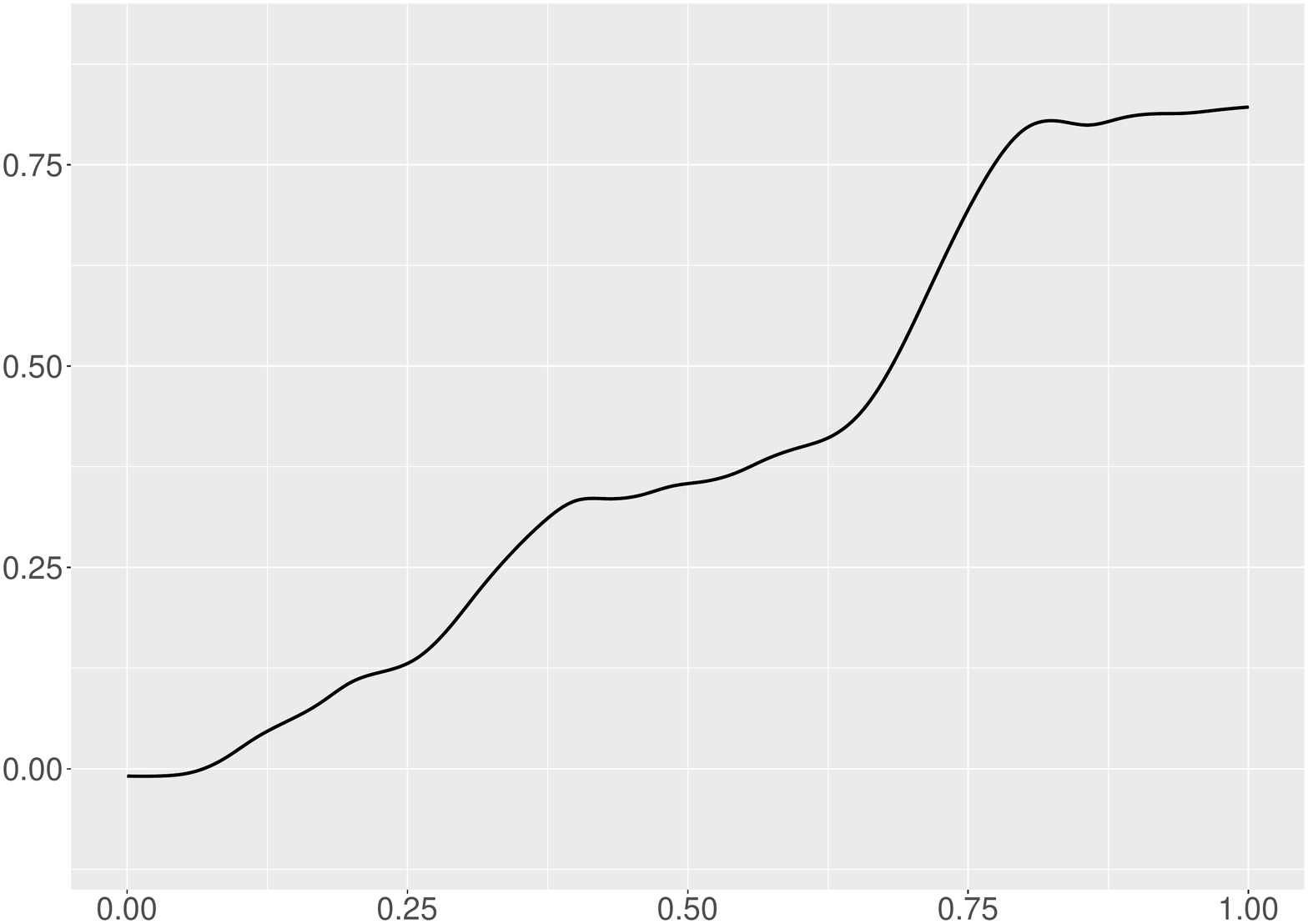}
    \caption{Reconstruction by P-spline\\{}}
    \end{subfigure}
    \begin{subfigure}{0.45\textwidth}
    \centering
    \includegraphics[width=\linewidth,height=0.45\textwidth]{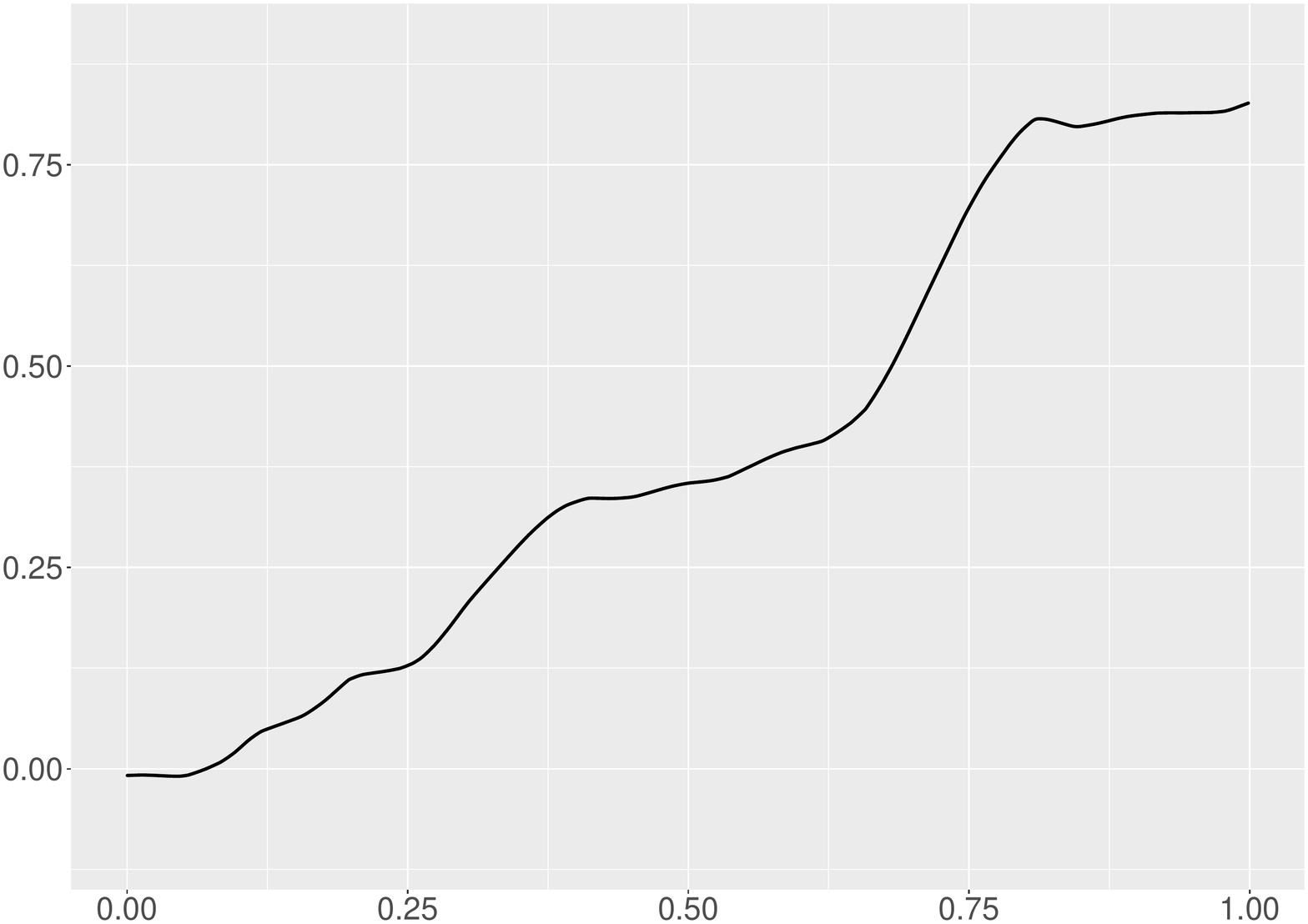}
    \caption{Reconstruction by adaptive V-spline, $\gamma=0$}
    \end{subfigure}
  \begin{subfigure}[t]{0.45\textwidth}
    \centering
    \includegraphics[width=\linewidth,height=0.45\textwidth]{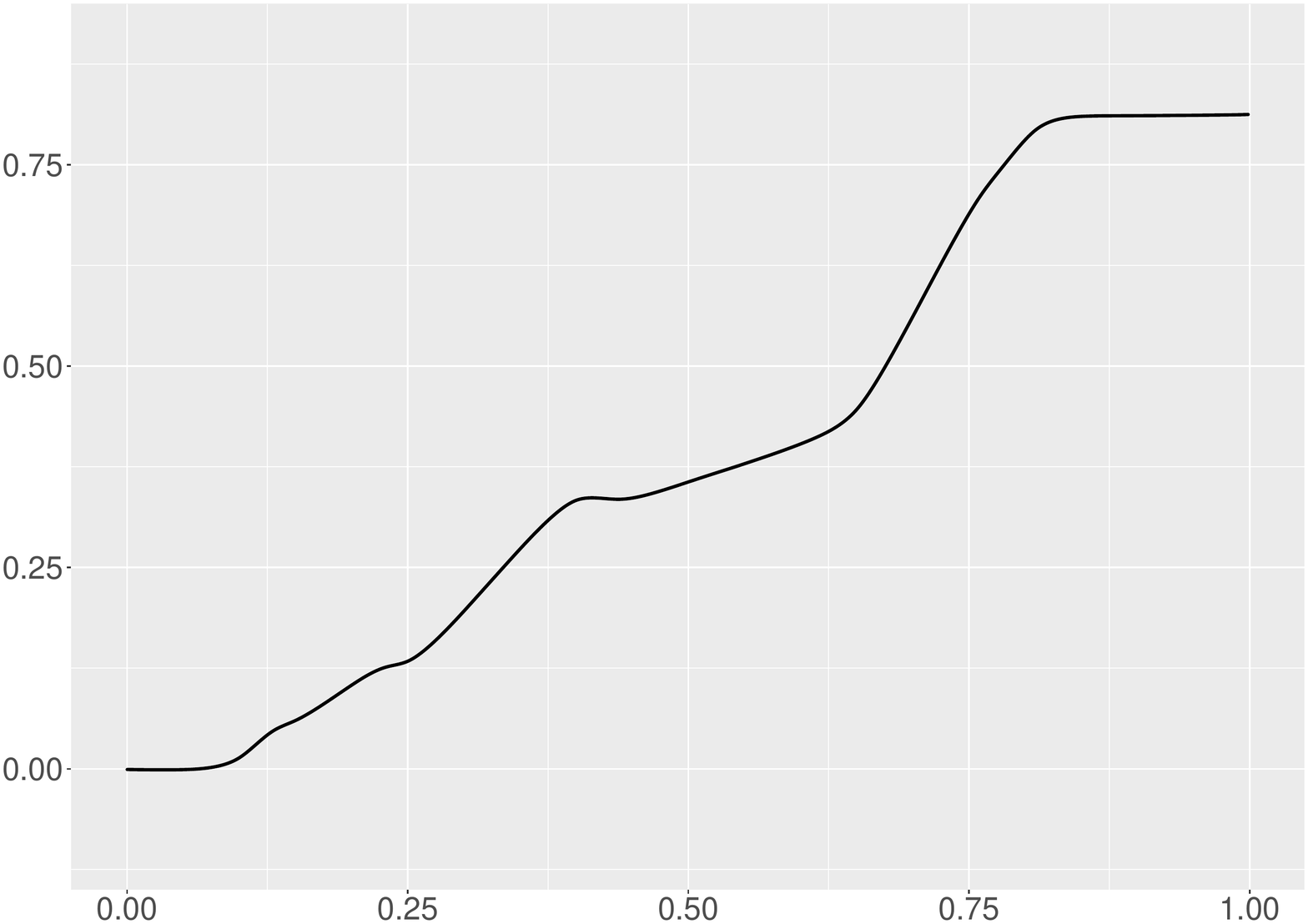}
    \caption{Reconstruction by non-adaptive V-spline}
    \end{subfigure}
    \begin{subfigure}[t]{0.45\textwidth}
    \centering
    \includegraphics[width=\linewidth,height=0.45\textwidth]{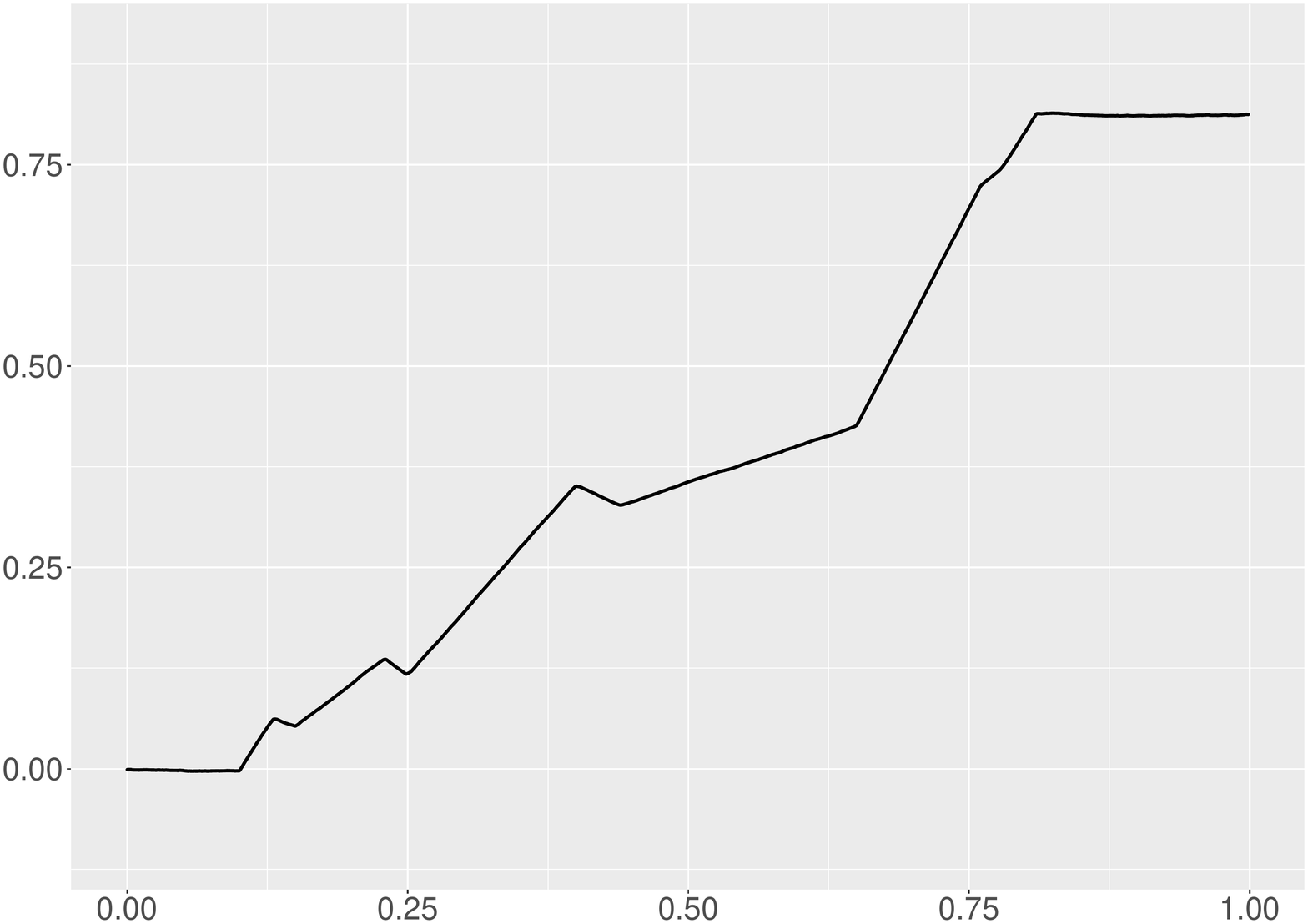}
    \caption{Reconstruction by adaptive V-spline\\\mbox{  } }
    \end{subfigure}
\caption{Numerical example: $\textit{Blocks}$. Comparison of different reconstruction methods with simulated data.}\label{num1}
 \end{figure}

\begin{figure}
    \centering
    \begin{subfigure}{0.45\textwidth}
    \centering
    \includegraphics[width=\linewidth,height=0.45\textwidth]{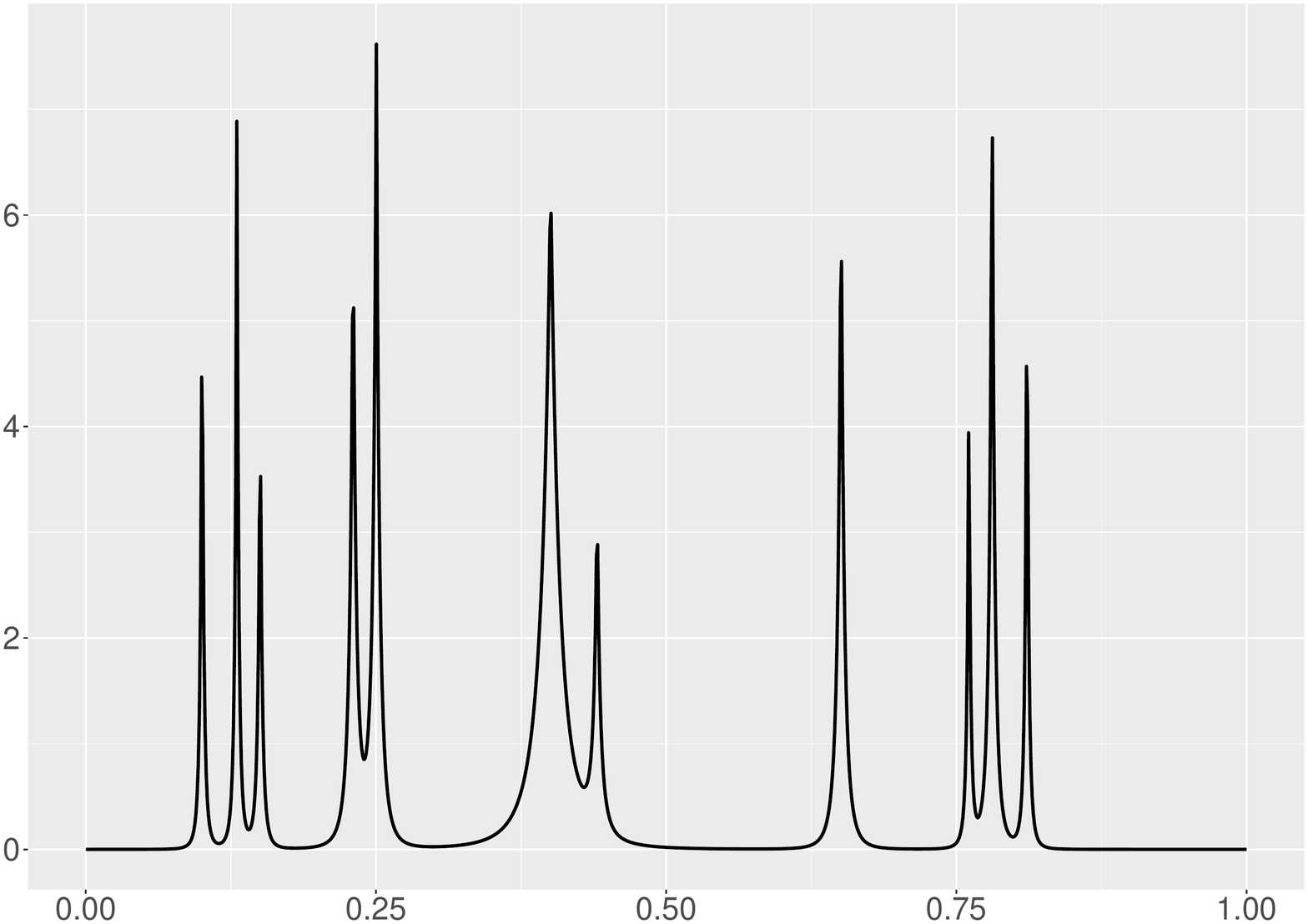}
    \caption{True \textit{Bumps} function}
    \end{subfigure}%
    \begin{subfigure}{0.45\textwidth}
    \centering
    \includegraphics[width=\linewidth,height=0.45\textwidth]{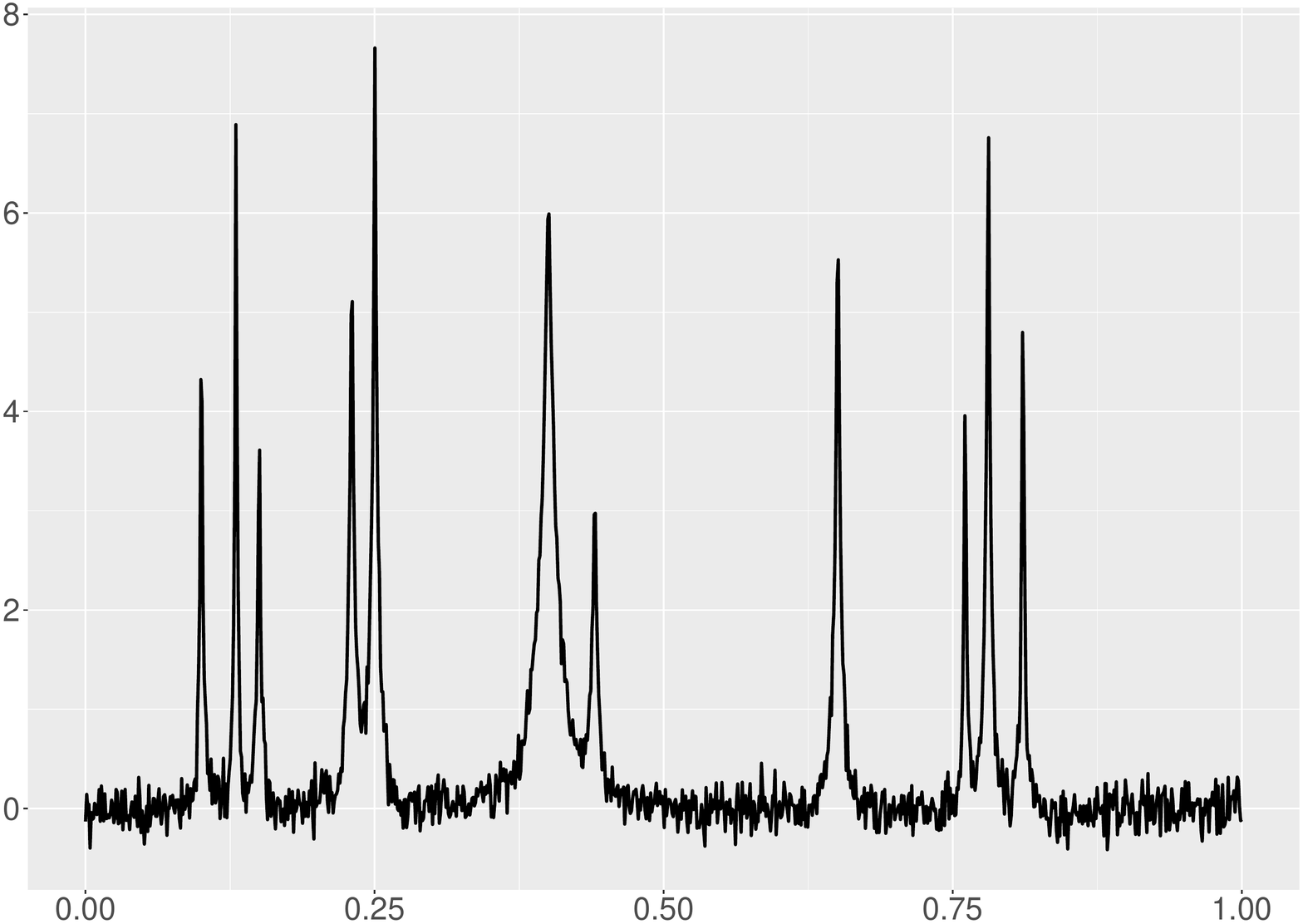}
    \caption{Noisy \textit{Bumps} at \textit{SNR}=7}
    \end{subfigure}
    \begin{subfigure}{0.45\textwidth}
    \centering
    \includegraphics[width=\linewidth,height=0.45\textwidth]{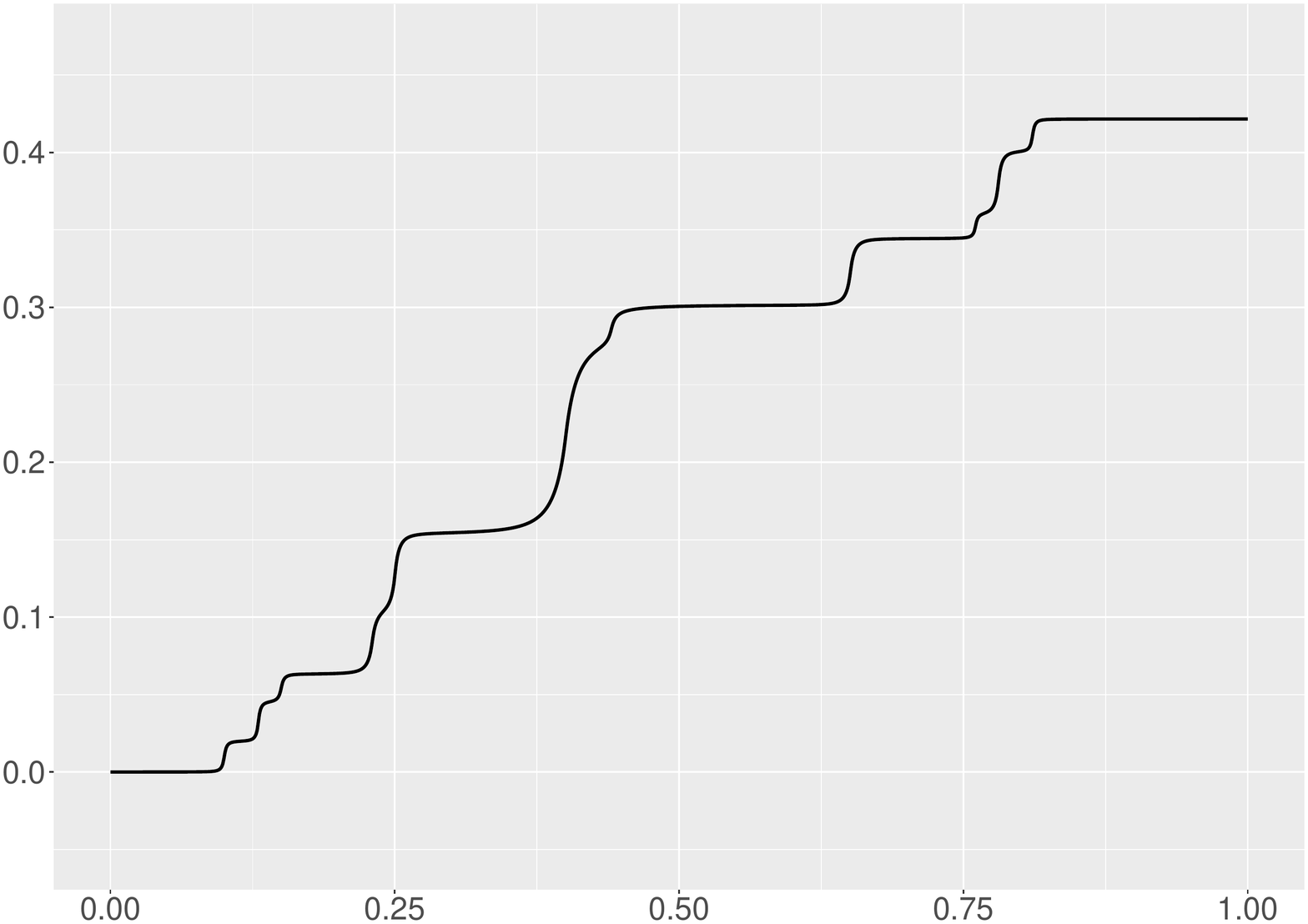}
    \caption{Generated positions}
    \end{subfigure}
    \begin{subfigure}{0.45\textwidth}
    \centering
    \includegraphics[width=\linewidth,height=0.45\textwidth]{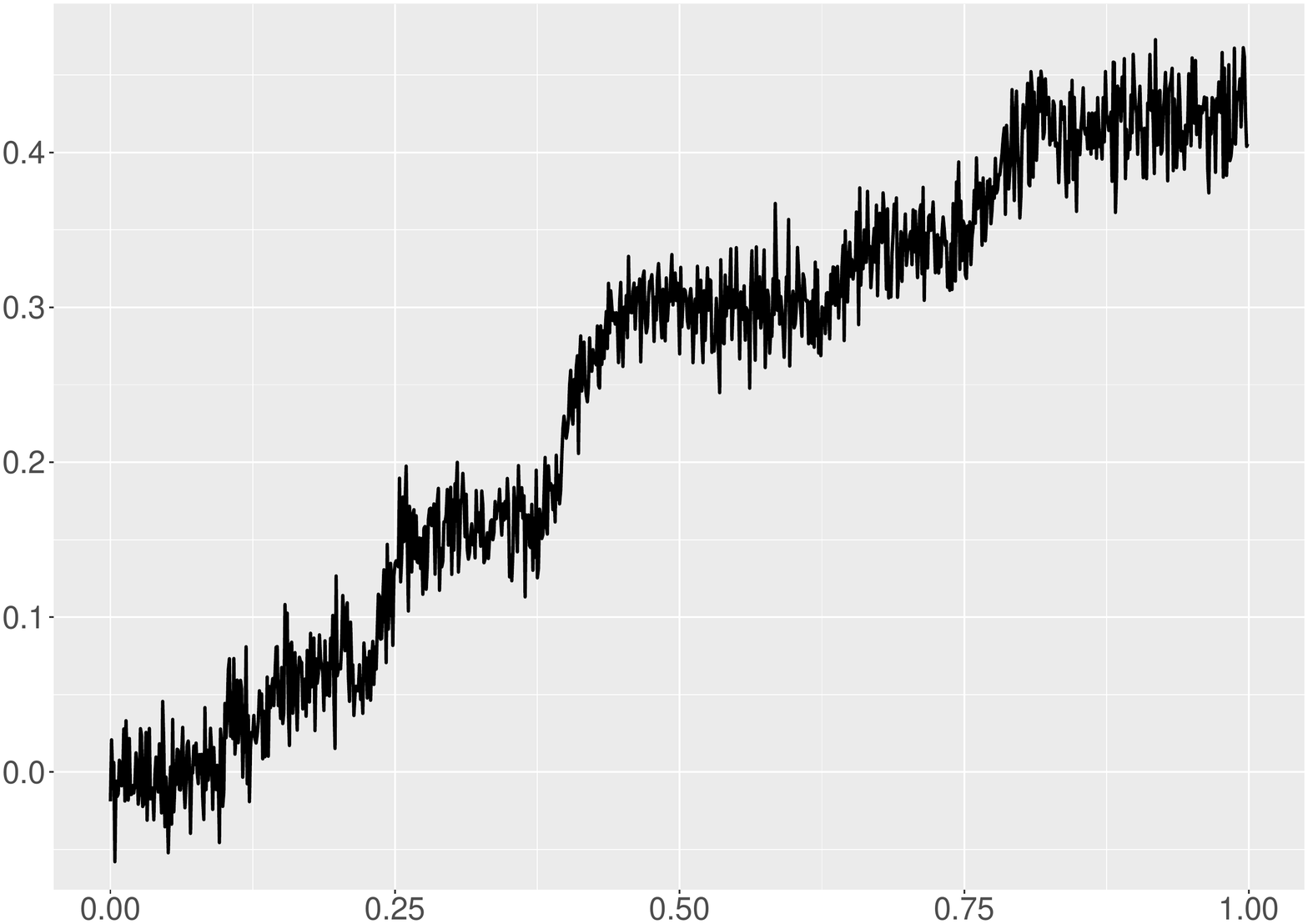}
    \caption{Noisy position at \textit{SNR}=7}
    \end{subfigure}
    \begin{subfigure}{0.45\textwidth}
    \centering
    \includegraphics[width=\linewidth,height=0.45\textwidth]{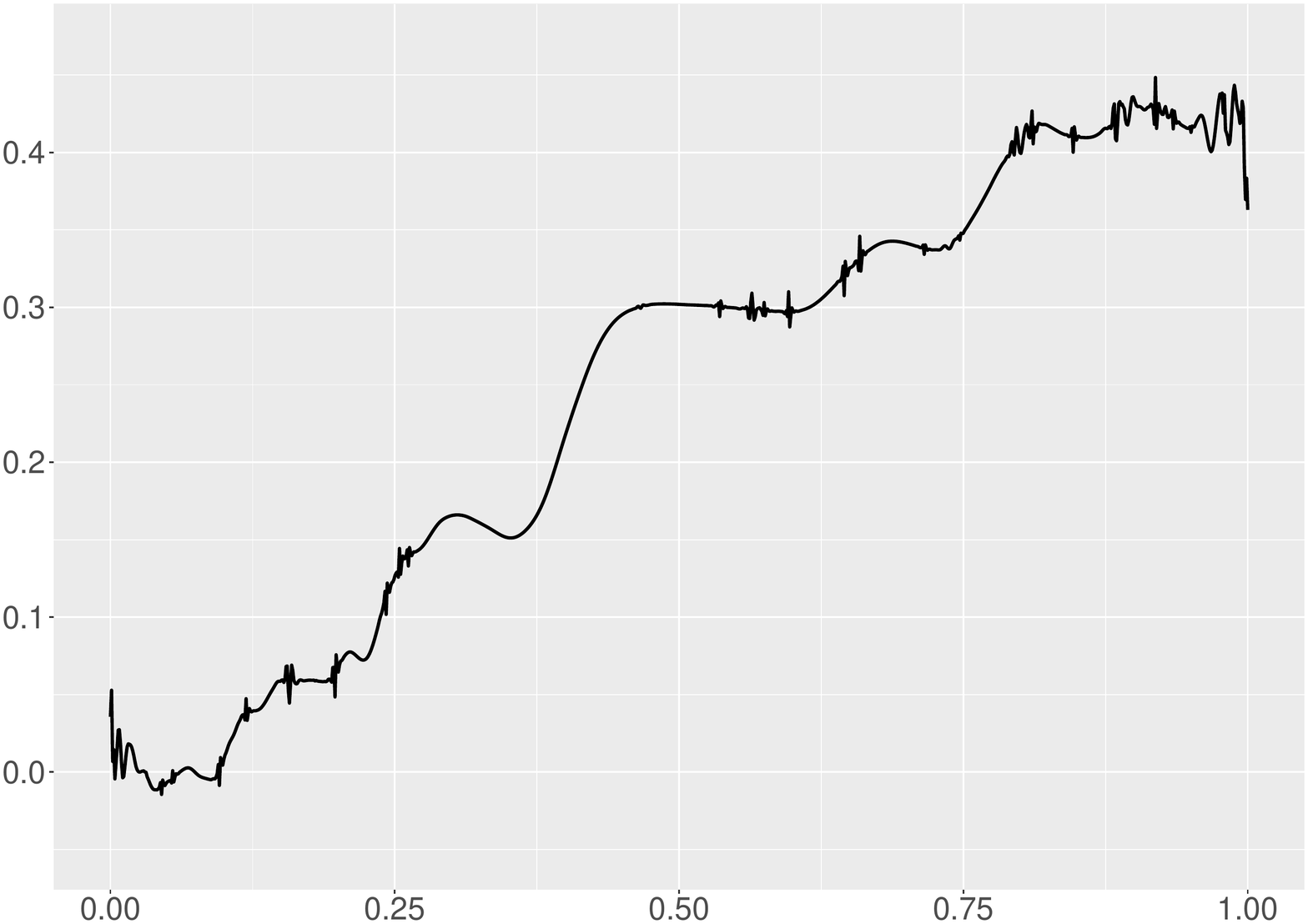}
    \caption{Reconstruction from Wavelet by sure threshold}
    \end{subfigure}
    \begin{subfigure}{0.45\textwidth}
    \centering
    \includegraphics[width=\linewidth,height=0.45\textwidth]{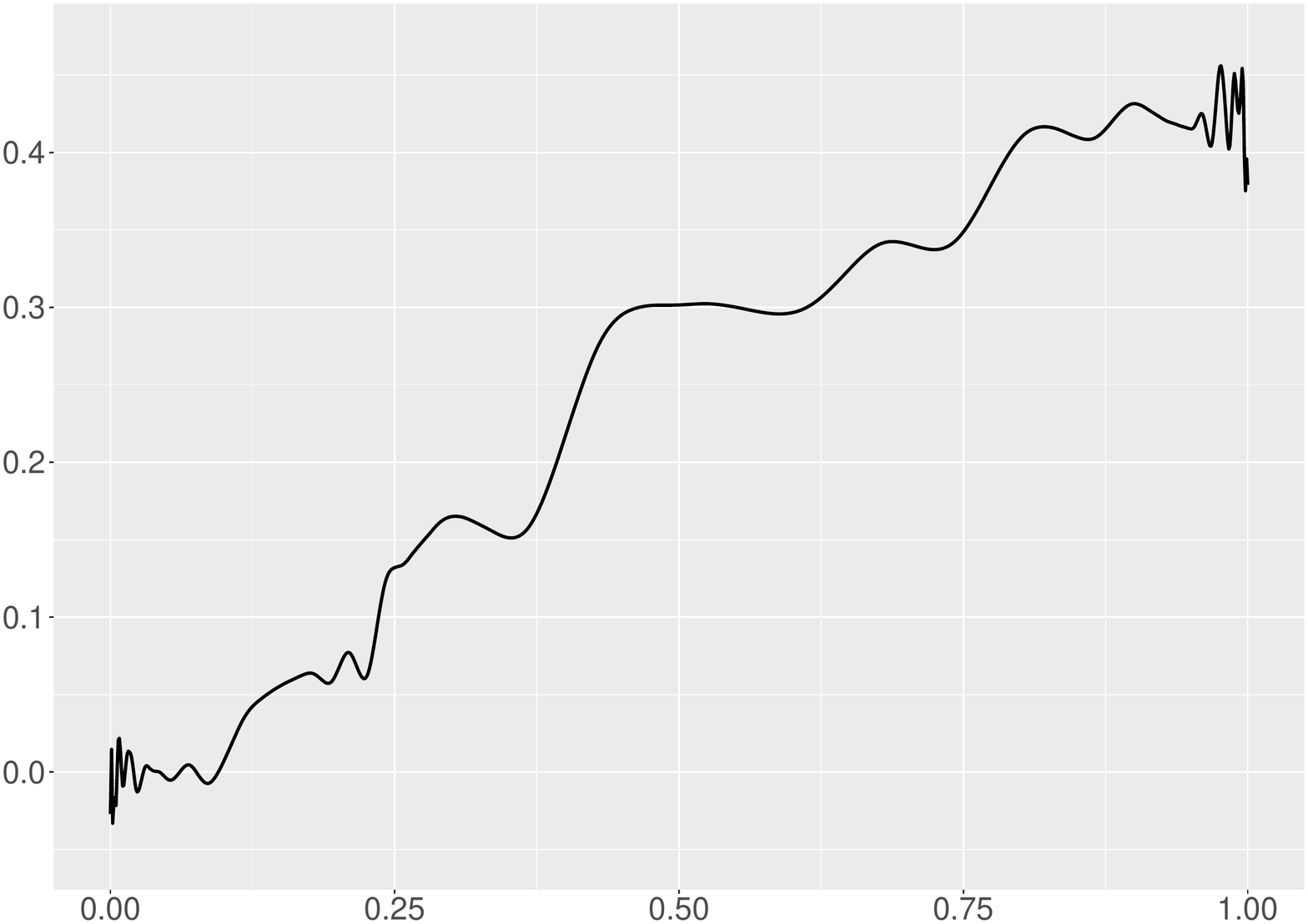}
    \caption{Reconstruction from Wavelet by BayesThresh approach}
    \end{subfigure}
    \begin{subfigure}{0.45\textwidth}
    \centering
    \includegraphics[width=\linewidth,height=0.45\textwidth]{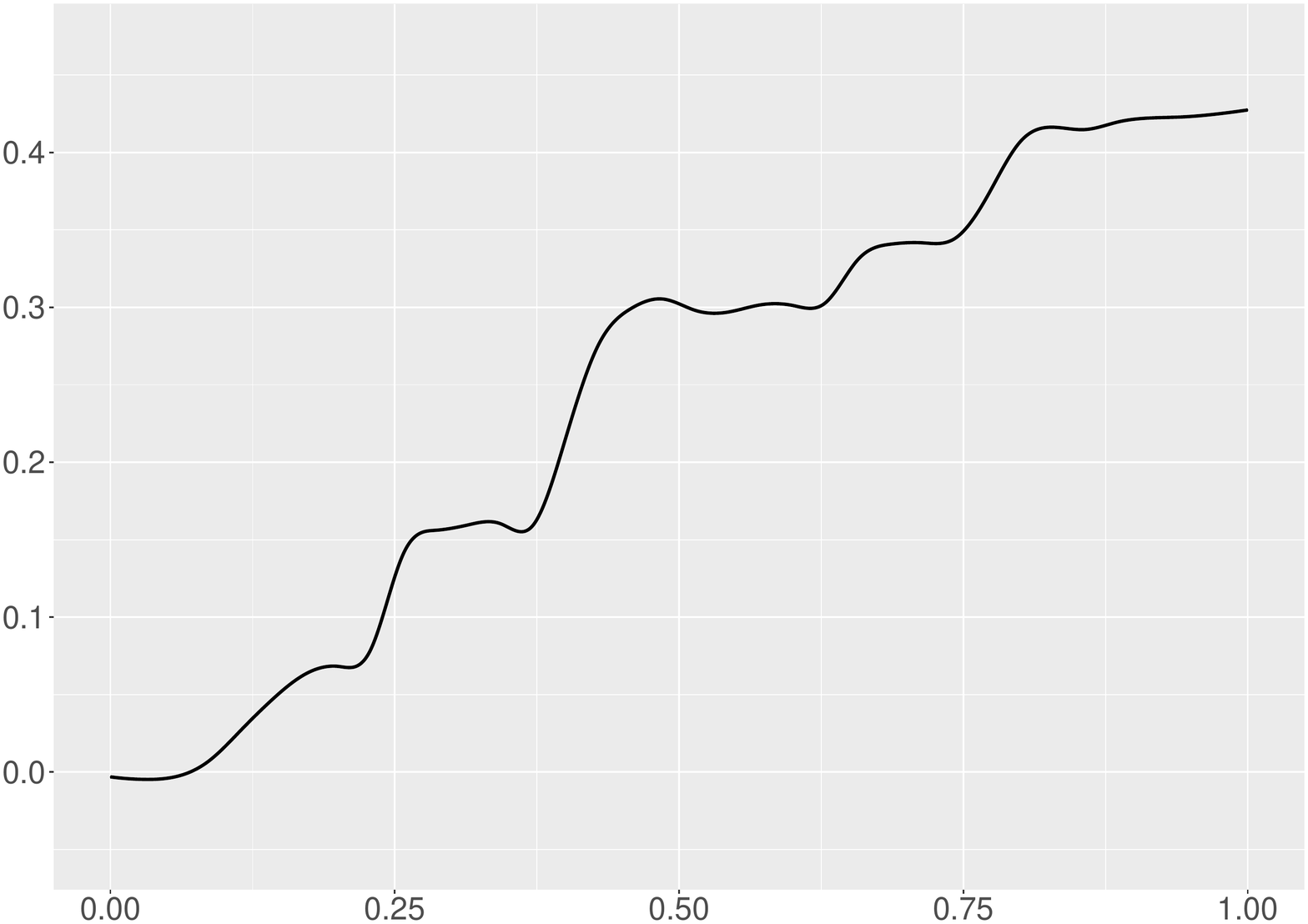}
    \caption{Reconstruction by P-spline }
    \end{subfigure}
    \begin{subfigure}{0.45\textwidth}
    \centering
    \includegraphics[width=\linewidth,height=0.45\textwidth]{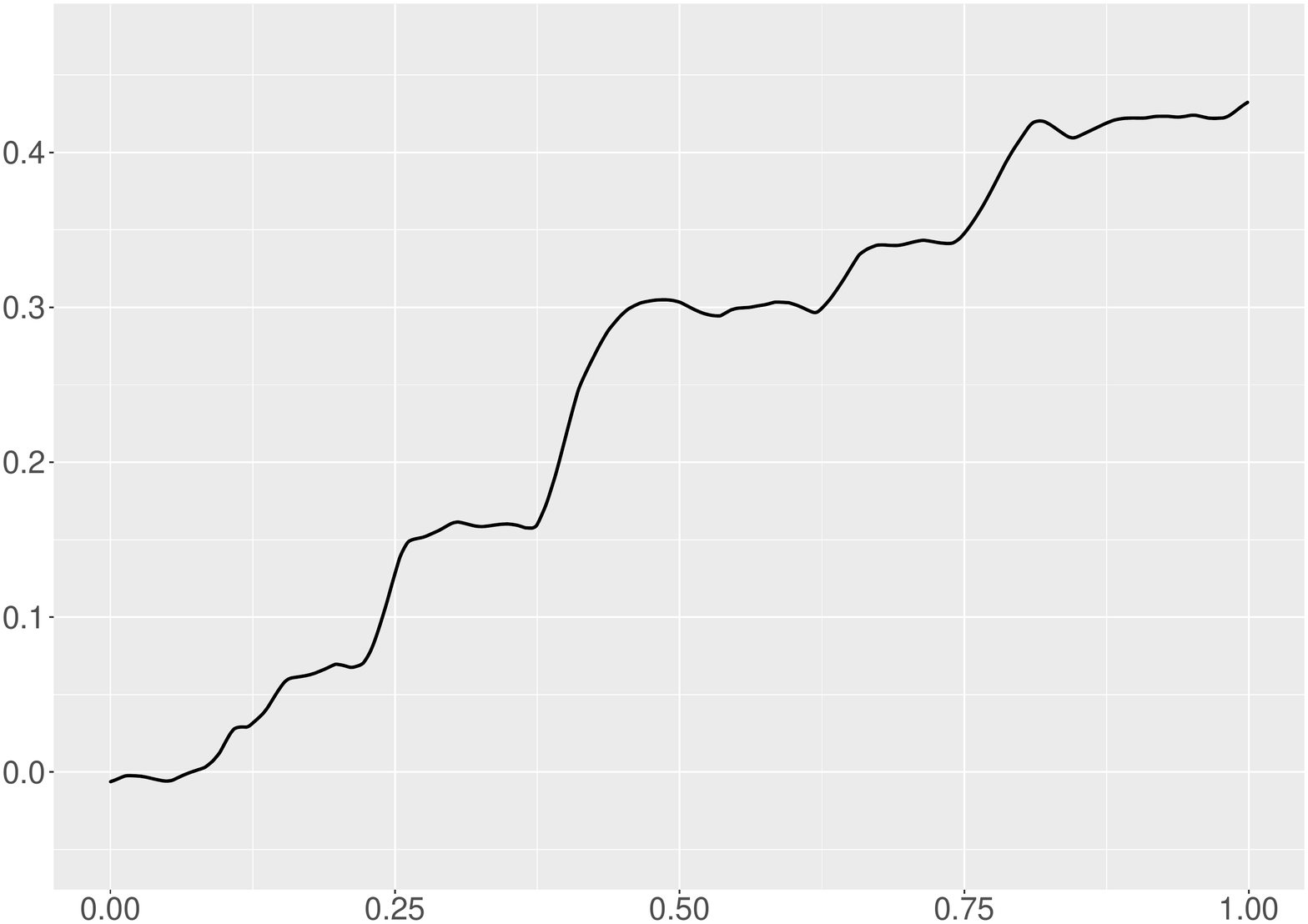}
    \caption{Reconstruction by adaptive V-spline , $\gamma=0$}
    \end{subfigure}
  \begin{subfigure}[t]{0.45\textwidth}
    \centering
    \includegraphics[width=\linewidth,height=0.45\textwidth]{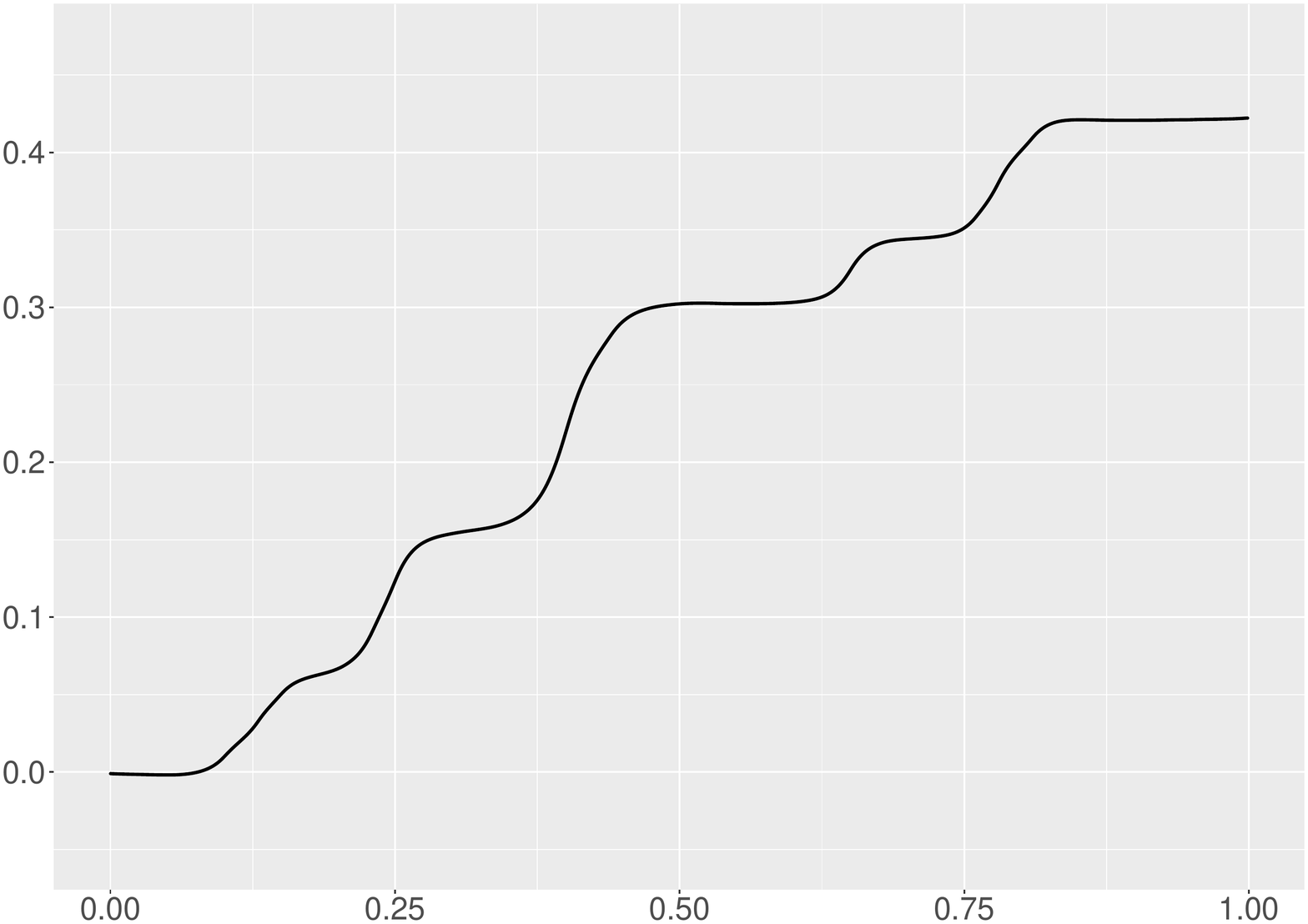}
    \caption{Reconstruction by non-adaptive V-spline}
    \end{subfigure}
    \begin{subfigure}[t]{0.45\textwidth}
    \centering
    \includegraphics[width=\linewidth,height=0.45\textwidth]{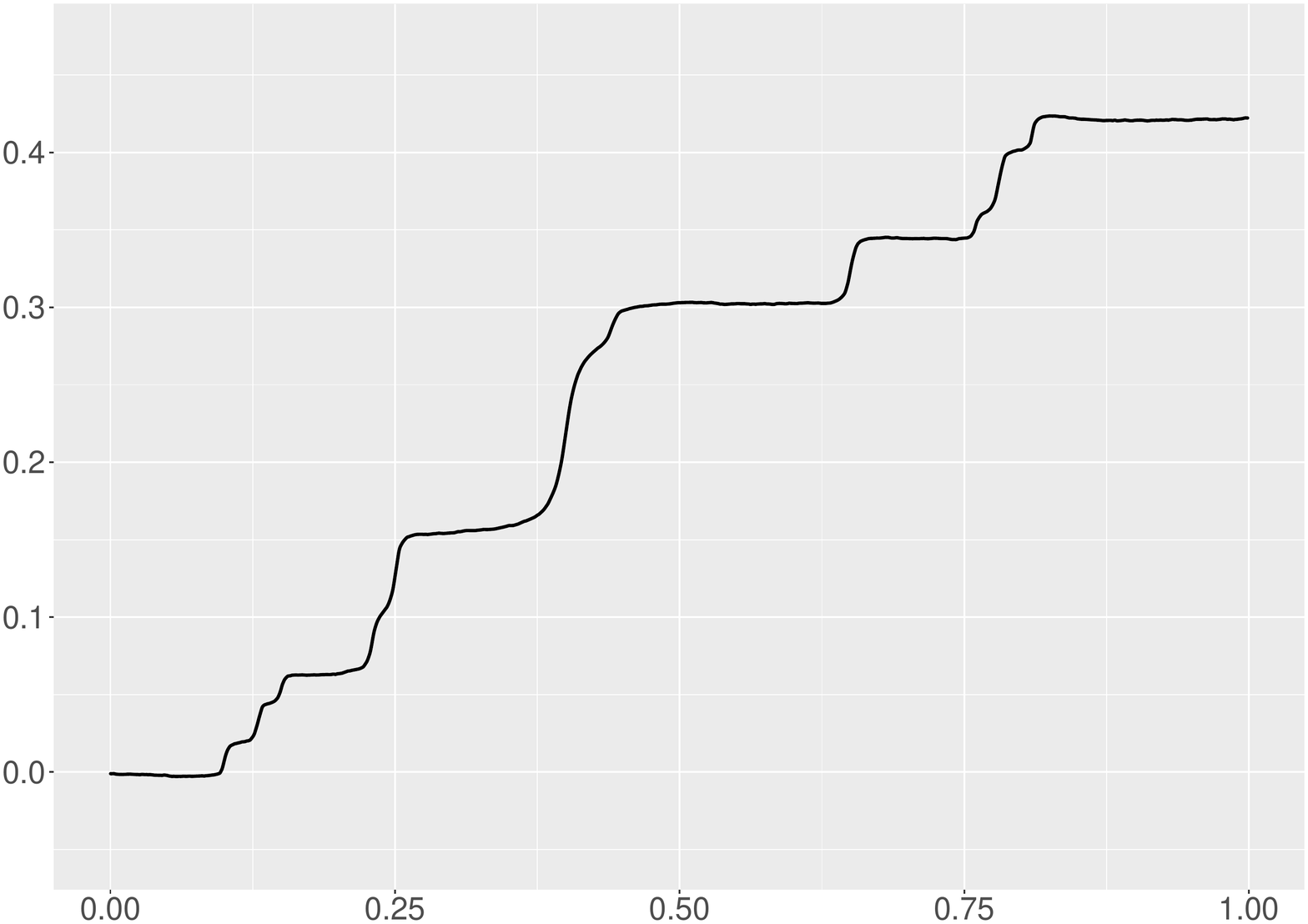}
    \caption{Reconstruction by adaptive V-spline \\\mbox{  } }
    \end{subfigure}
\caption{Numerical example: $\textit{Bumps}$. Comparison of different reconstruction methods with simulated data.}\label{num2}
 \end{figure}

\begin{figure}
    \centering
    \begin{subfigure}{0.45\textwidth}
    \centering
    \includegraphics[width=\linewidth,height=0.45\textwidth]{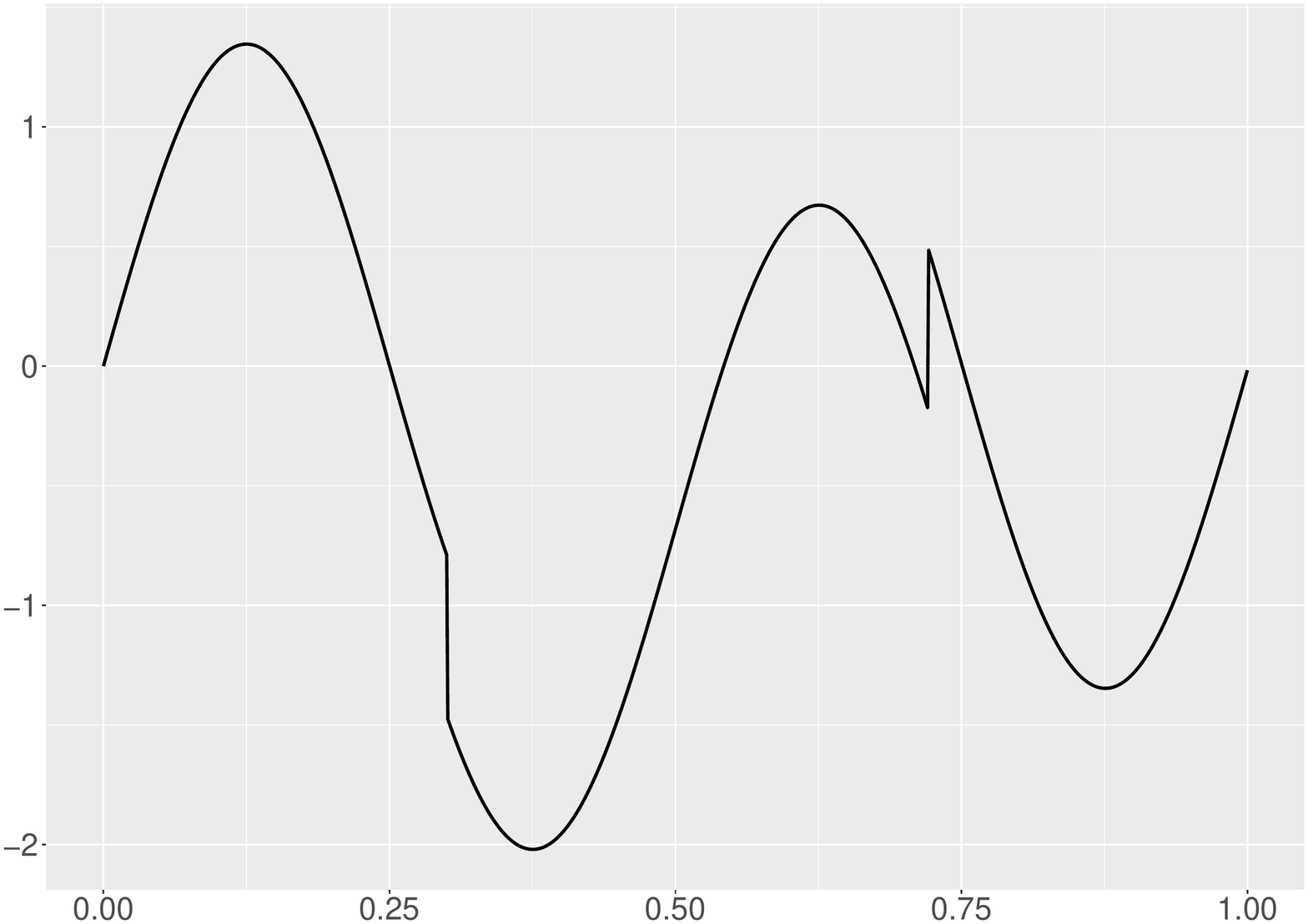}
    \caption{True \textit{HeaviSine} function}
    \end{subfigure}%
    \begin{subfigure}{0.45\textwidth}
    \centering
    \includegraphics[width=\linewidth,height=0.45\textwidth]{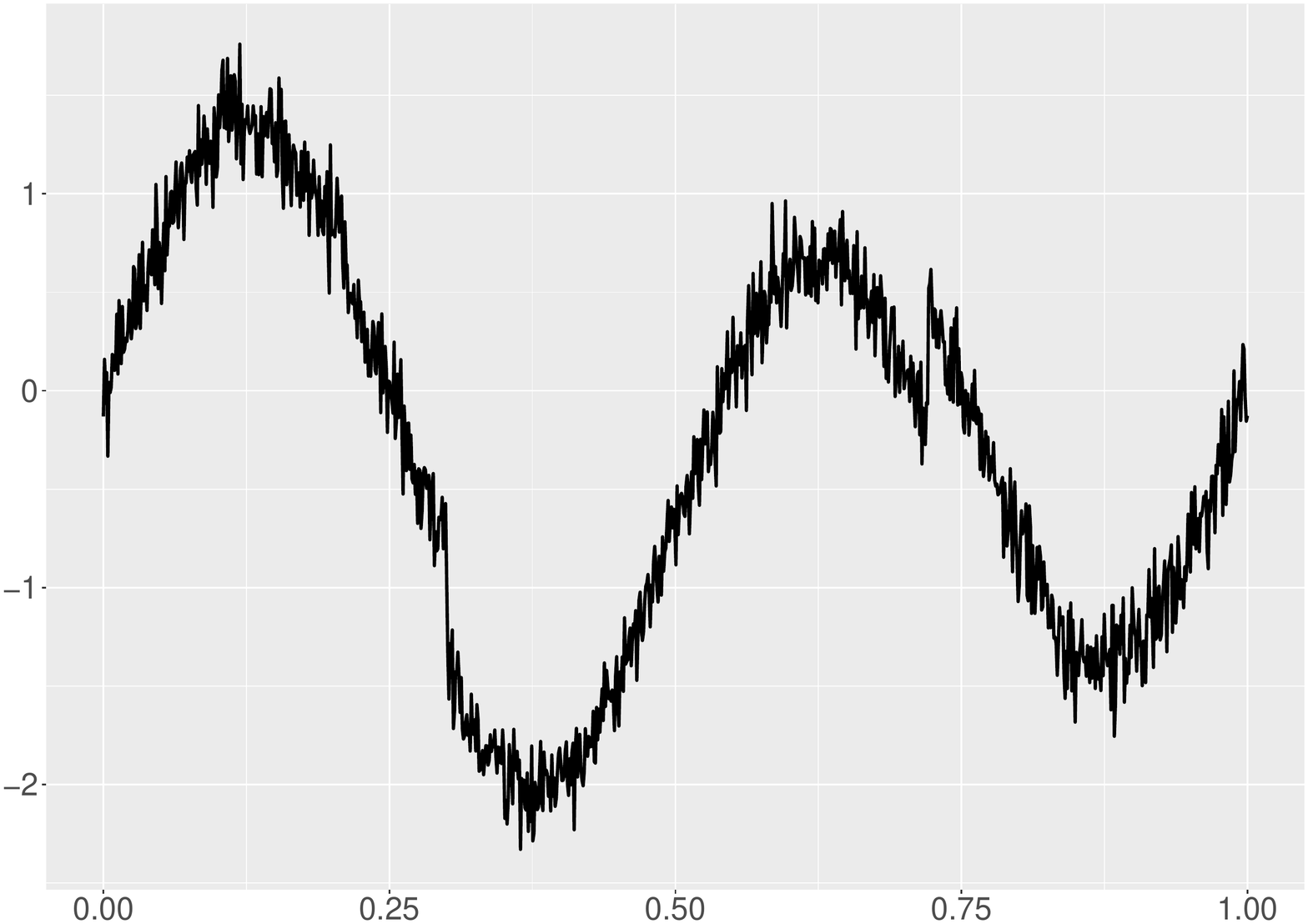}
    \caption{Noisy \textit{HeaviSine} at \textit{SNR}=7}
    \end{subfigure}
    \begin{subfigure}{0.45\textwidth}
    \centering
    \includegraphics[width=\linewidth,height=0.45\textwidth]{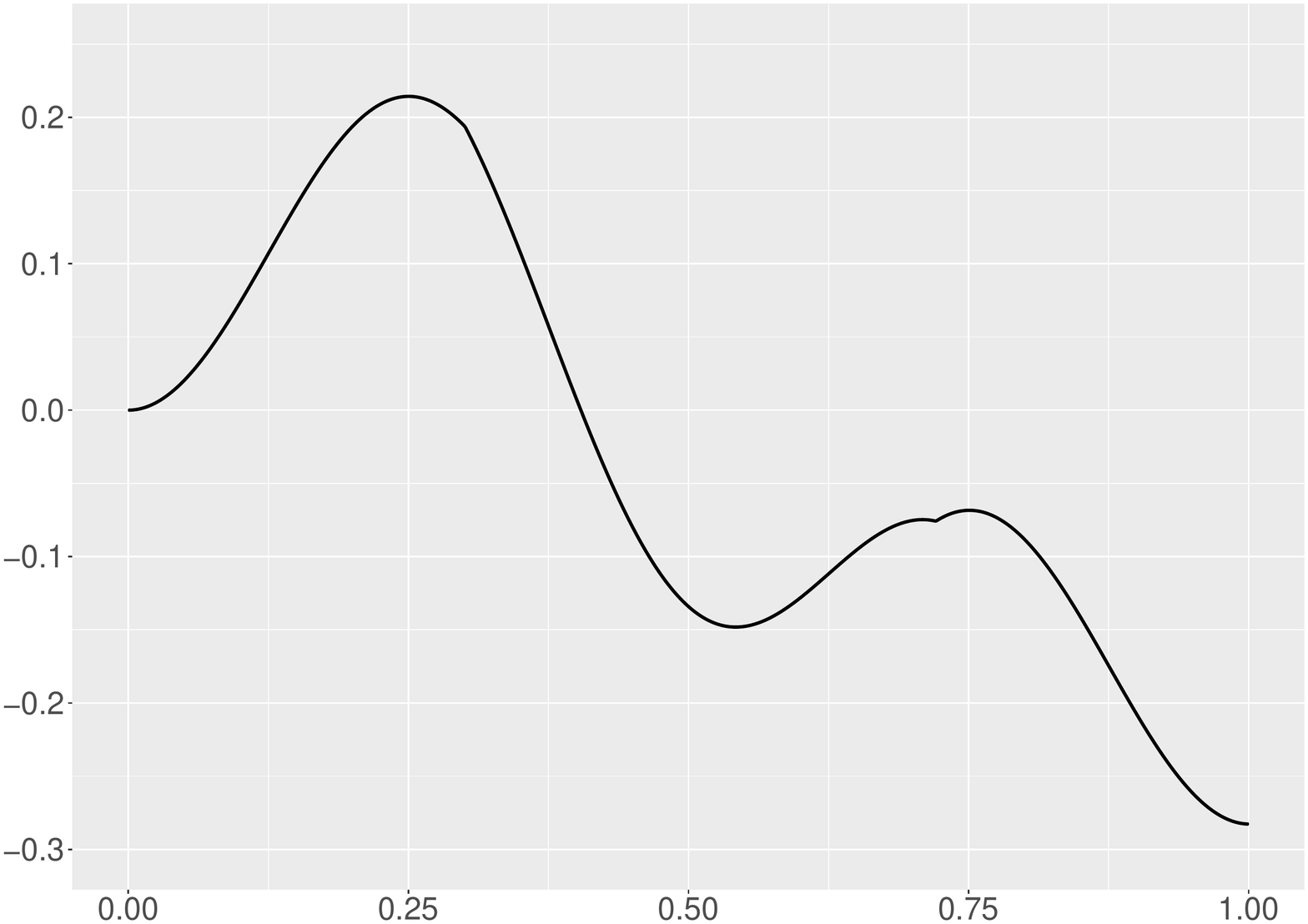}
    \caption{Generated positions}
    \end{subfigure}
    \begin{subfigure}{0.45\textwidth}
    \centering
    \includegraphics[width=\linewidth,height=0.45\textwidth]{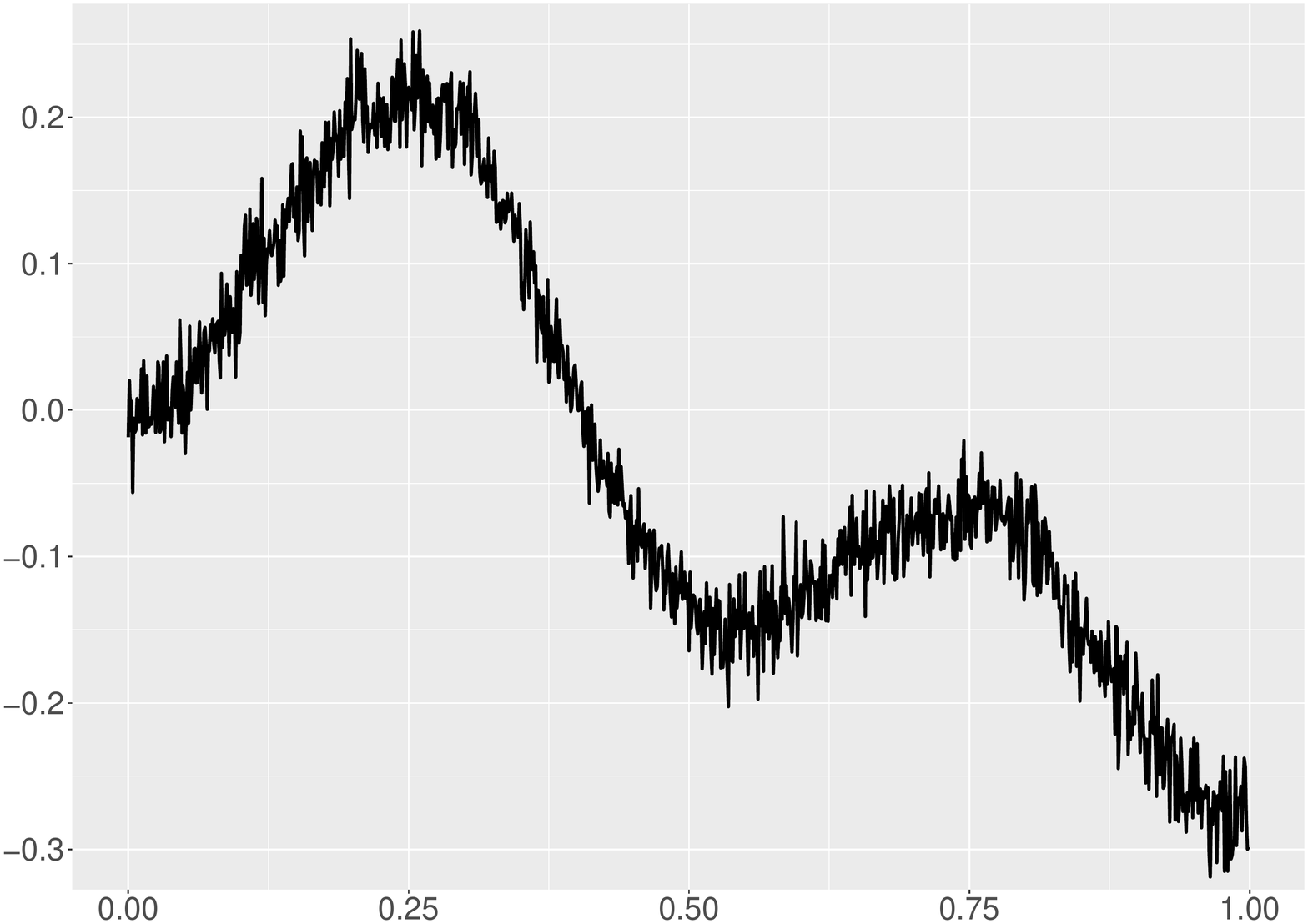}
    \caption{Noisy position at \textit{SNR}=7}
    \end{subfigure}
    \begin{subfigure}{0.45\textwidth}
    \centering
    \includegraphics[width=\linewidth,height=0.45\textwidth]{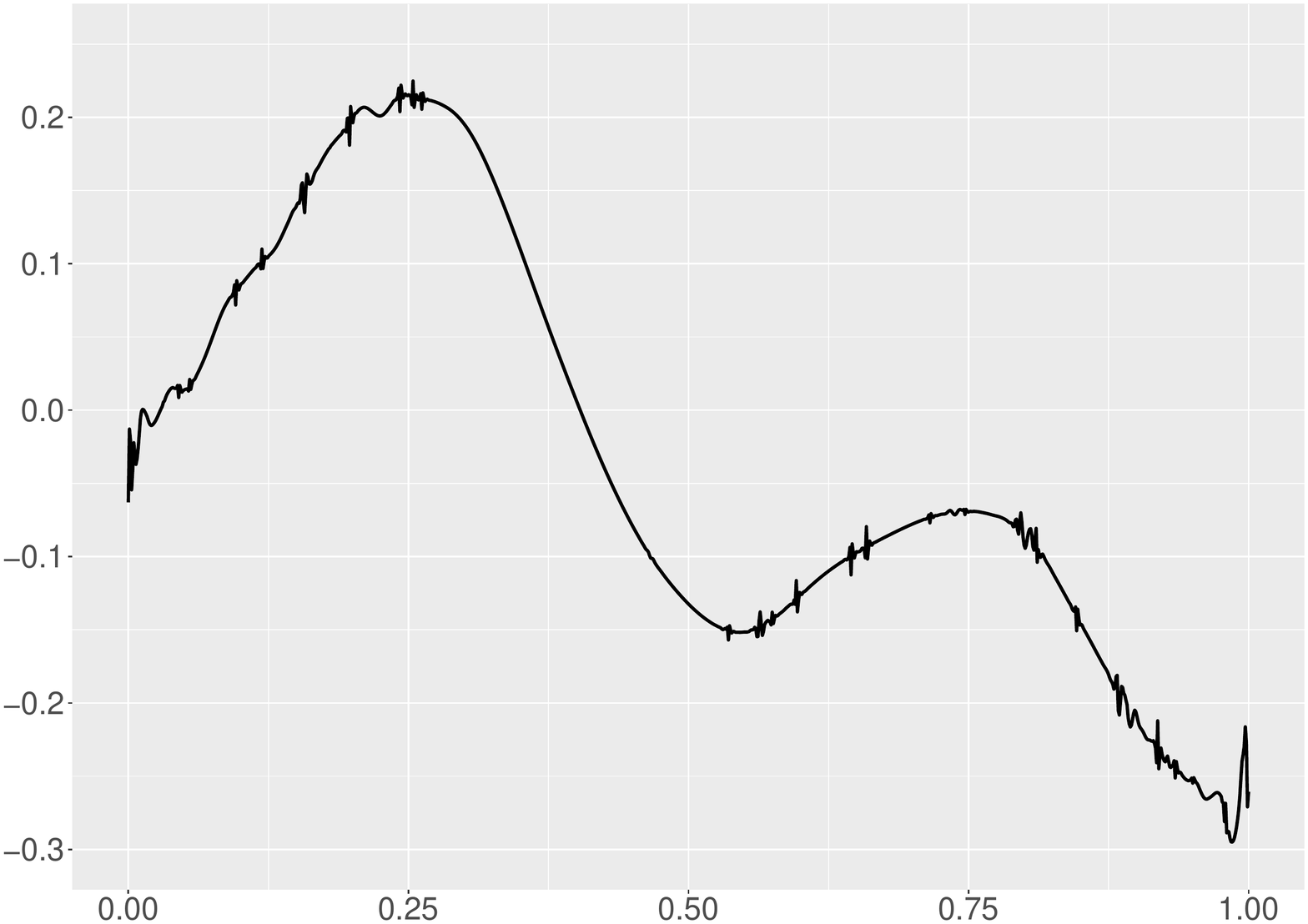}
    \caption{Reconstruction from Wavelet by sure threshold}
    \end{subfigure}
    \begin{subfigure}{0.45\textwidth}
    \centering
    \includegraphics[width=\linewidth,height=0.45\textwidth]{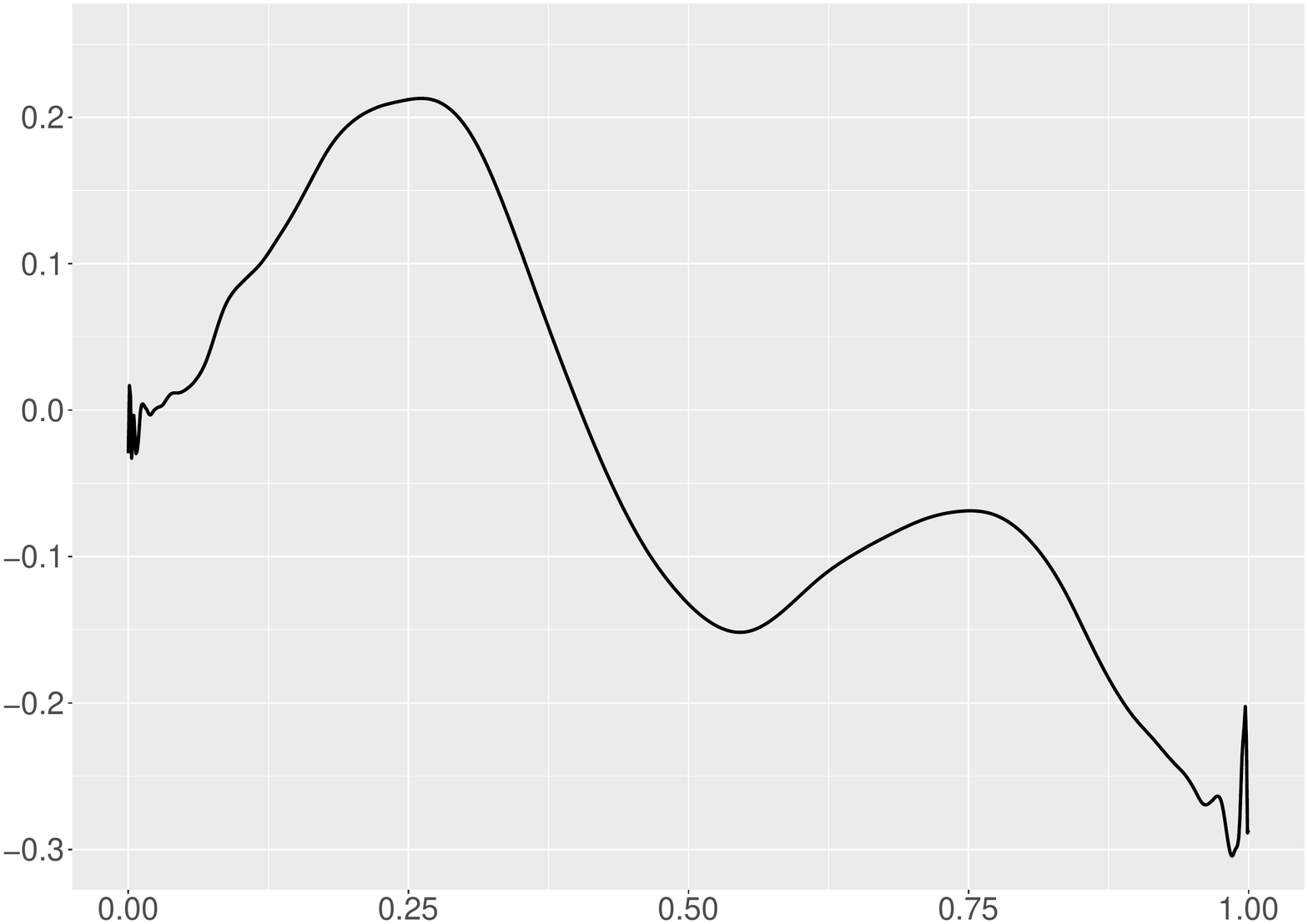}
    \caption{Reconstruction from Wavelet by BayesThresh approach}
    \end{subfigure}
    \begin{subfigure}{0.45\textwidth}
    \centering
    \includegraphics[width=\linewidth,height=0.45\textwidth]{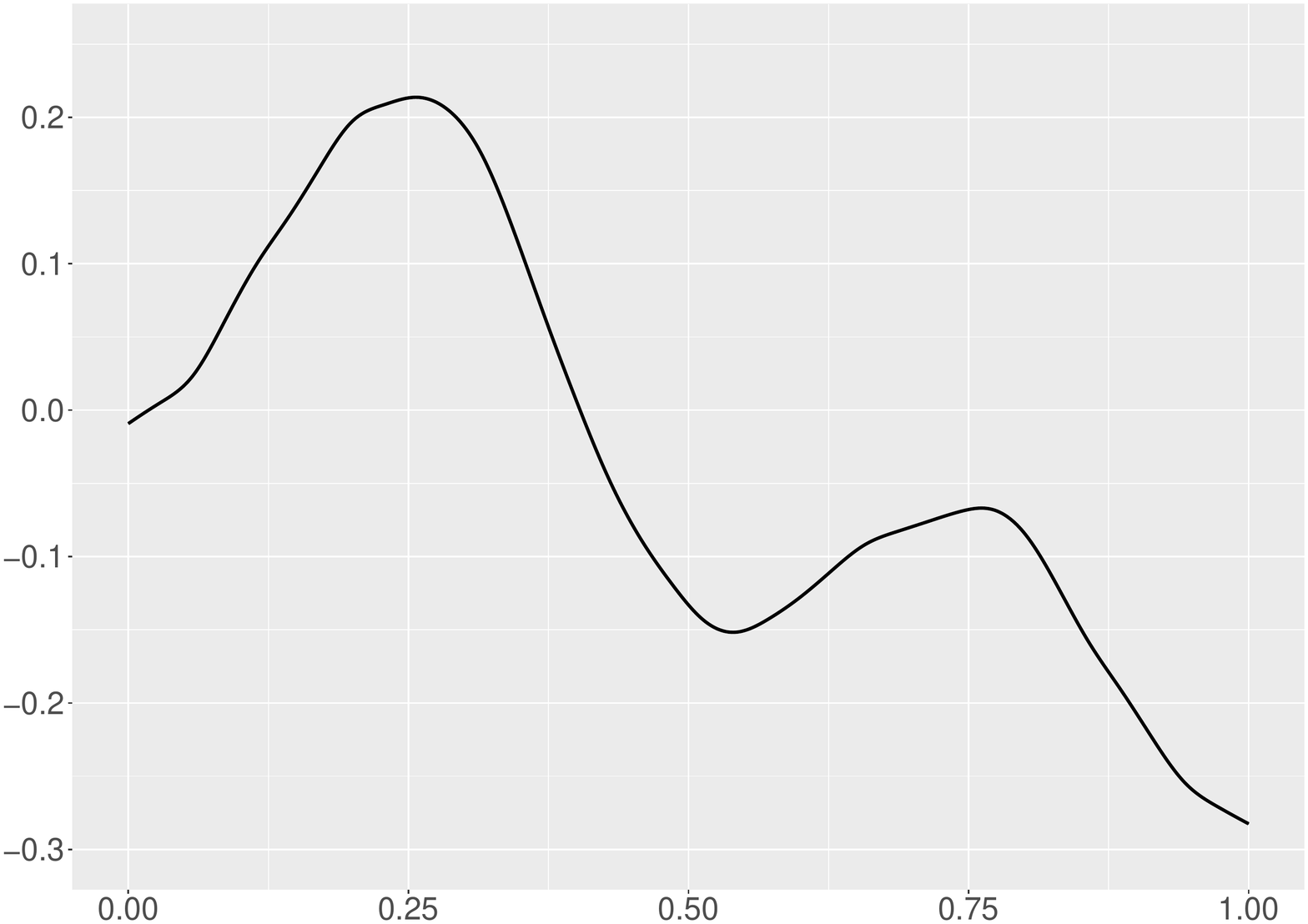}
    \caption{Reconstruction by P-spline}
    \end{subfigure}
    \begin{subfigure}{0.45\textwidth}
    \centering
    \includegraphics[width=\linewidth,height=0.45\textwidth]{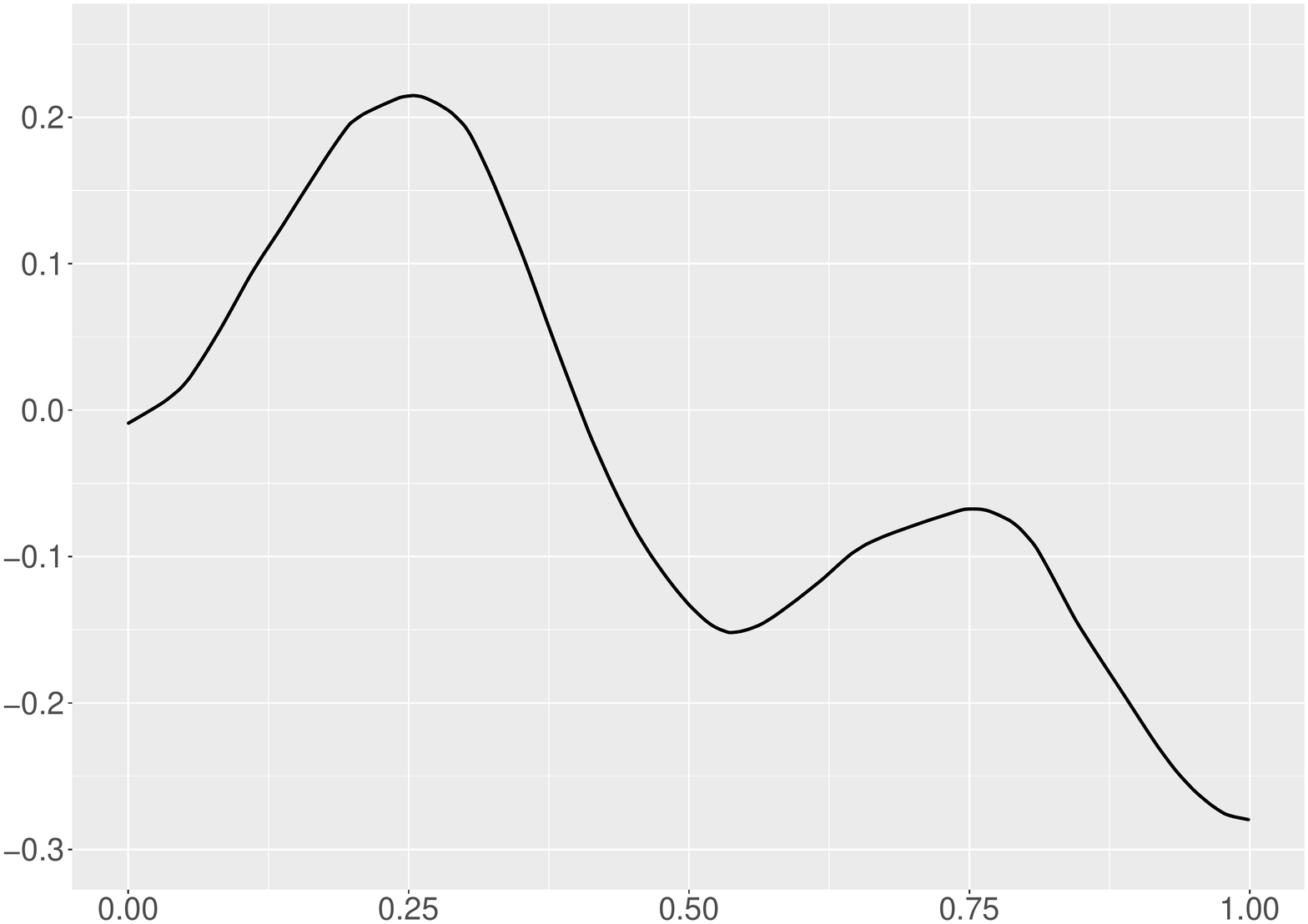}
    \caption{Reconstruction by adaptive V-spline , $\gamma=0$}
    \end{subfigure}
  \begin{subfigure}[t]{0.45\textwidth}
    \centering
    \includegraphics[width=\linewidth,height=0.45\textwidth]{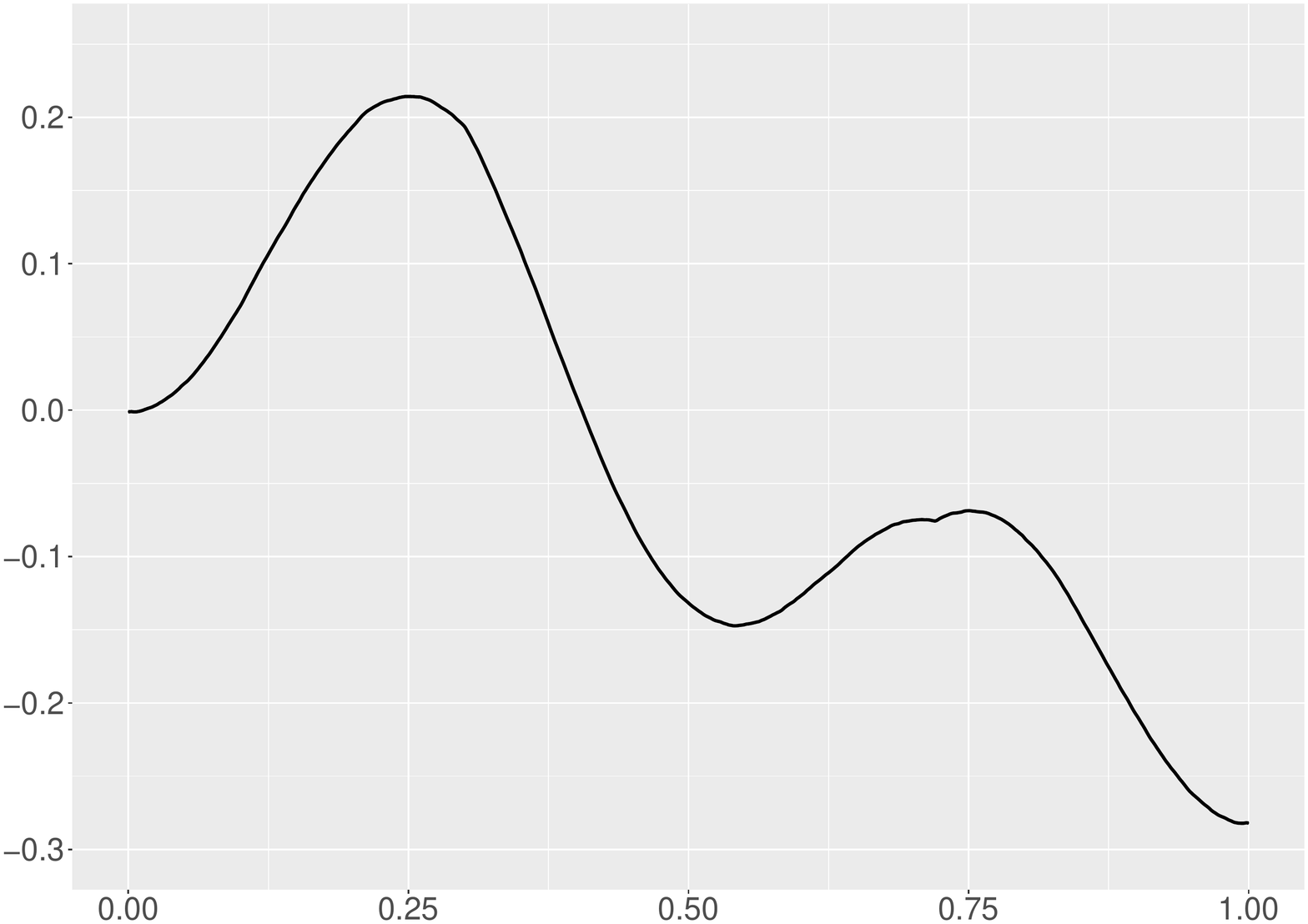}
    \caption{Reconstruction by non-adaptive V-spline}
    \end{subfigure}
    \begin{subfigure}[t]{0.45\textwidth}
    \centering
    \includegraphics[width=\linewidth,height=0.45\textwidth]{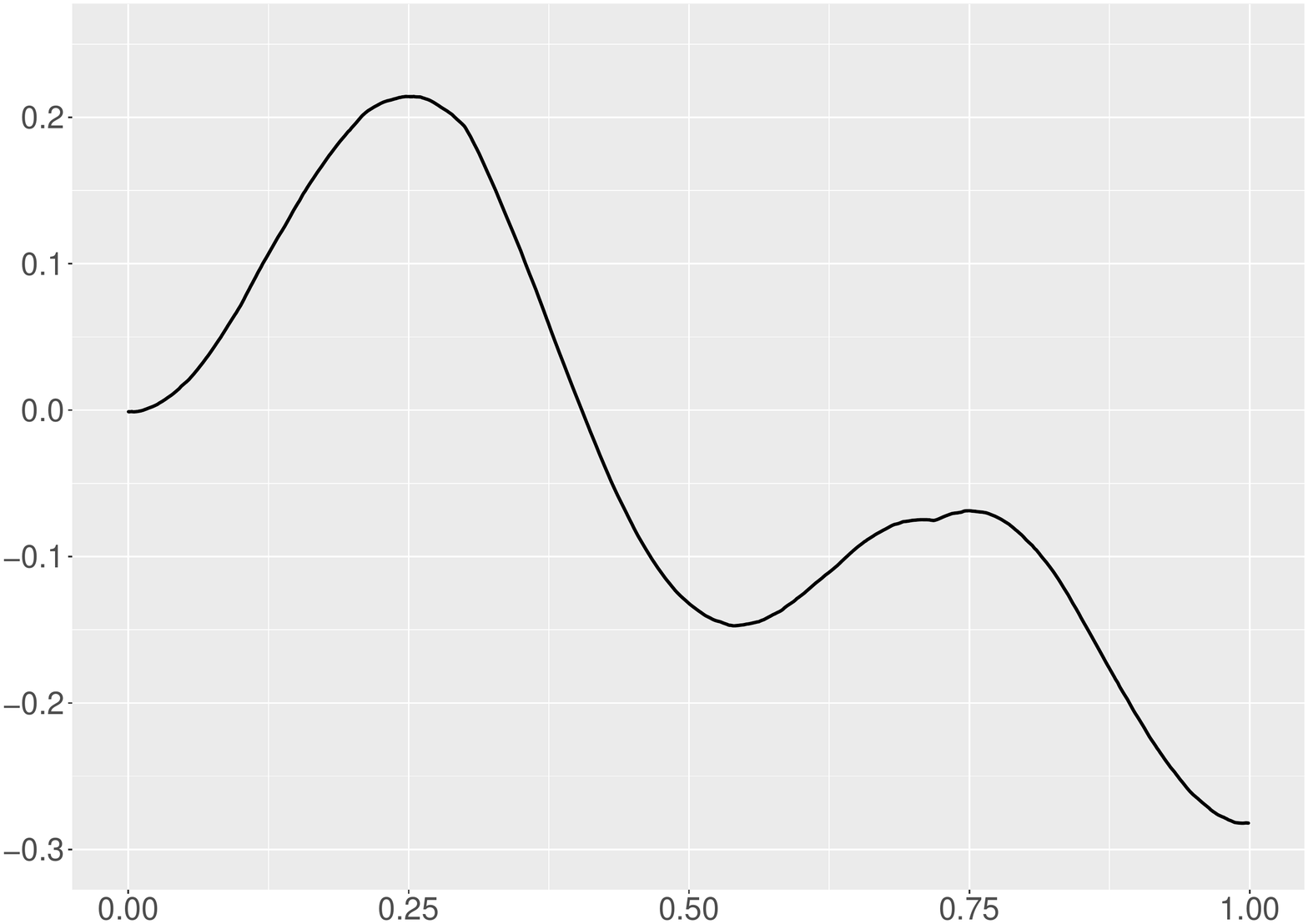}
    \caption{Reconstruction by adaptive V-spline \\ \mbox{  }}
    \end{subfigure}
\caption{Numerical example: $\textit{HeaviSine}$. Comparison of different reconstruction methods with simulated data.}\label{num3}
 \end{figure}

\begin{figure}
    \centering
    \begin{subfigure}{0.45\textwidth}
    \centering
    \includegraphics[width=\linewidth,height=0.45\textwidth]{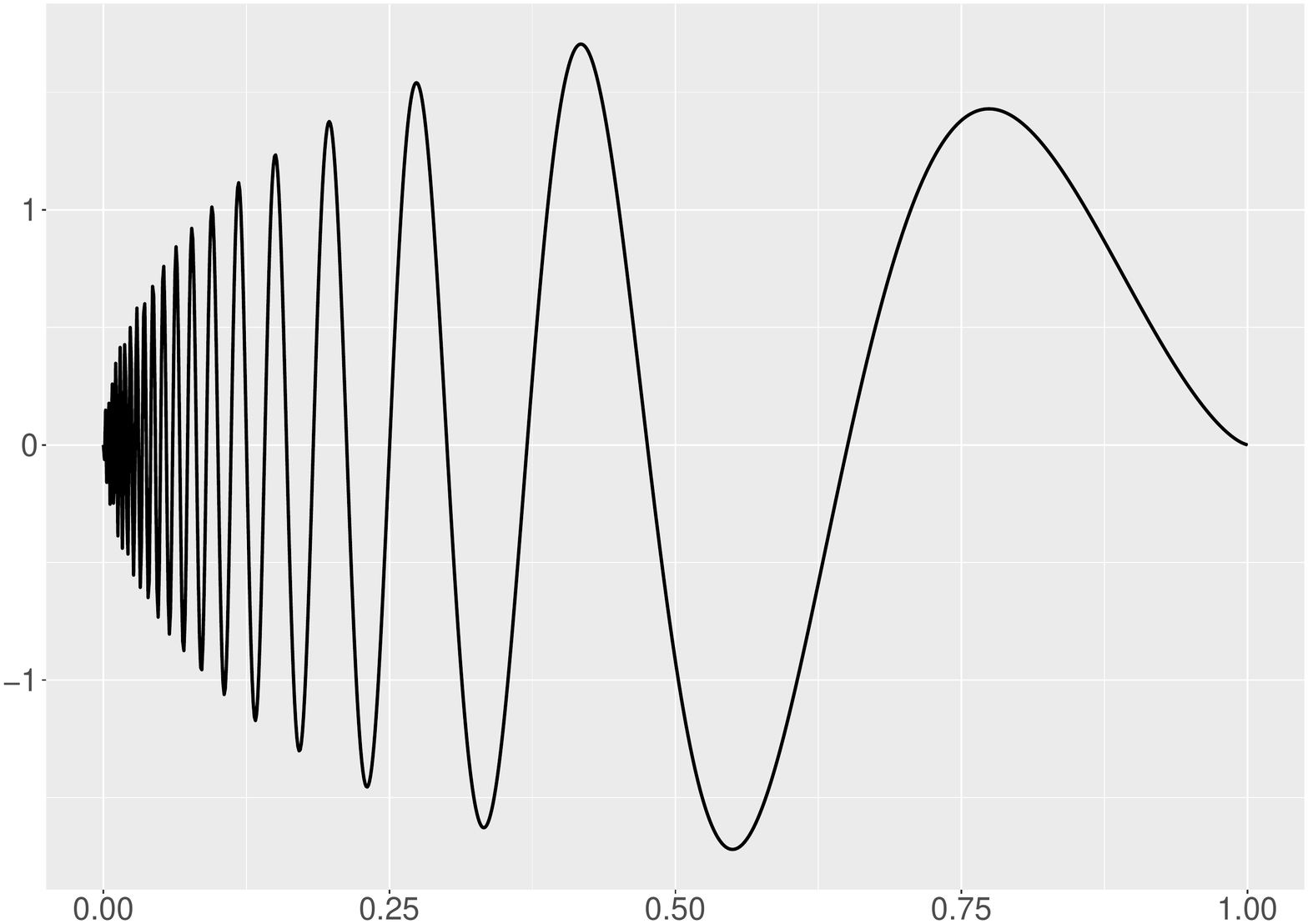}
    \caption{True \textit{Doppler} function}
    \end{subfigure}%
    \begin{subfigure}{0.45\textwidth}
    \centering
    \includegraphics[width=\linewidth,height=0.45\textwidth]{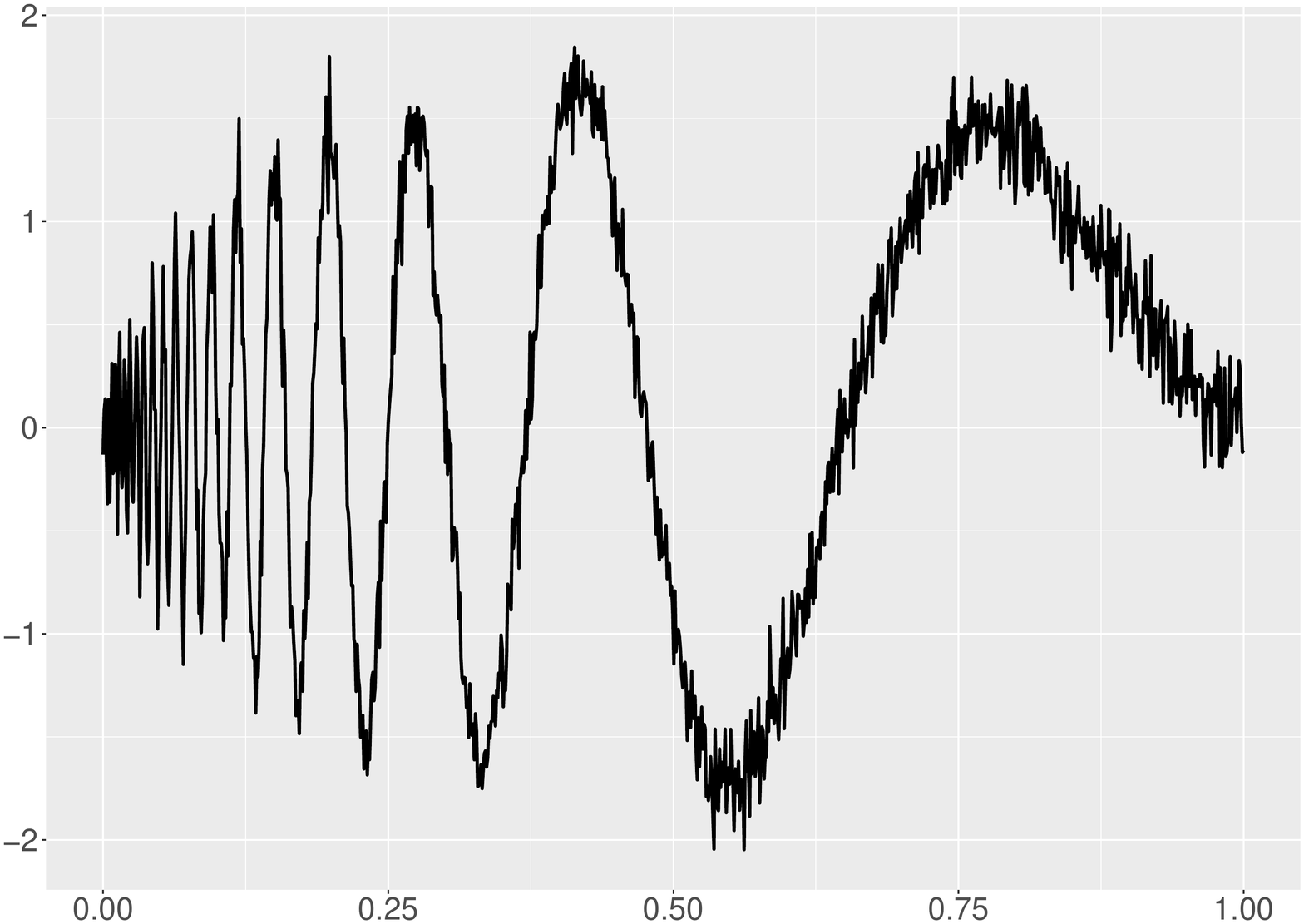}
    \caption{Noisy \textit{Doppler} at \textit{SNR}=7}
    \end{subfigure}
    \begin{subfigure}{0.45\textwidth}
    \centering
    \includegraphics[width=\linewidth,height=0.45\textwidth]{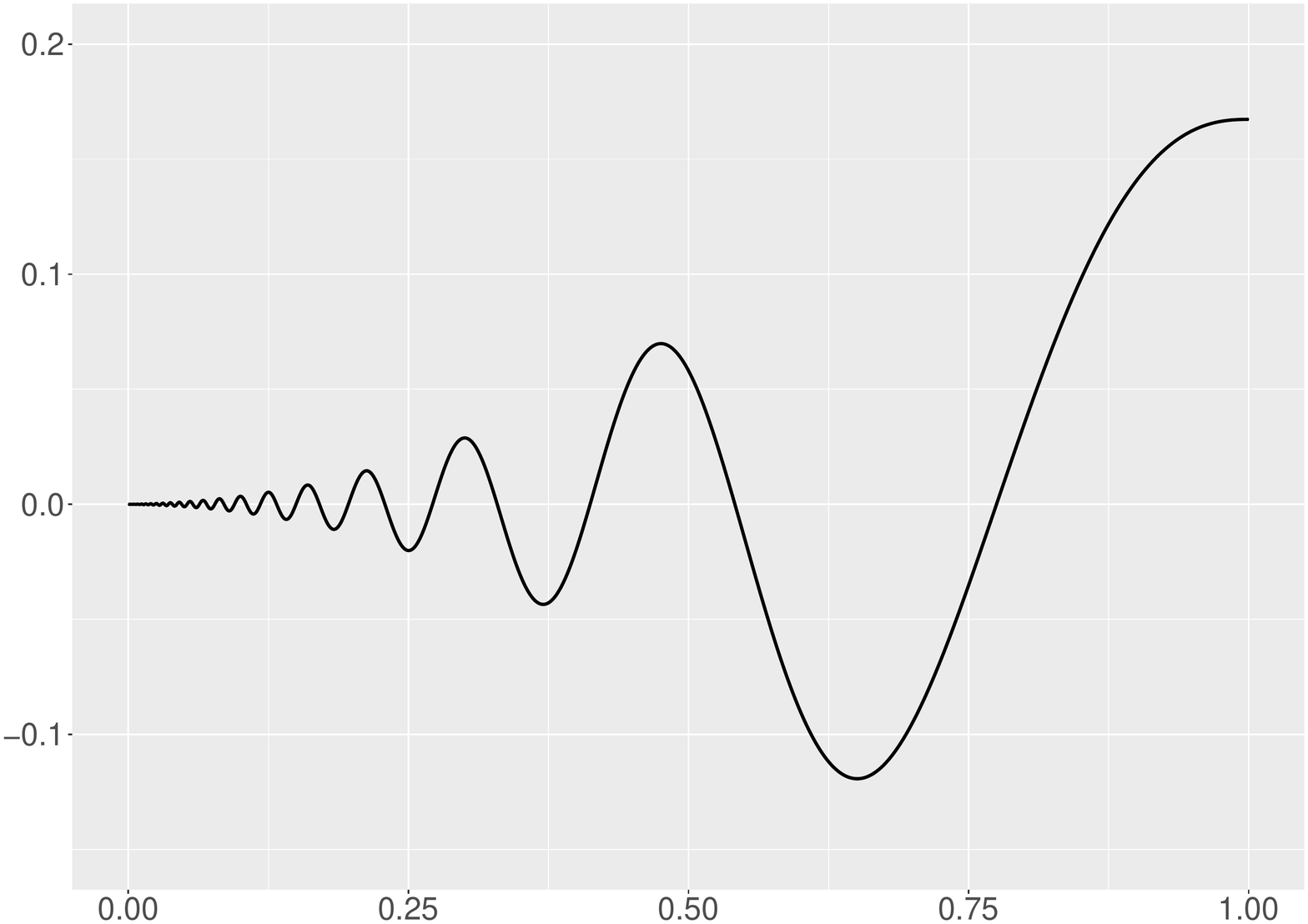}
    \caption{Generated positions}
    \end{subfigure}
    \begin{subfigure}{0.45\textwidth}
    \centering
    \includegraphics[width=\linewidth,height=0.45\textwidth]{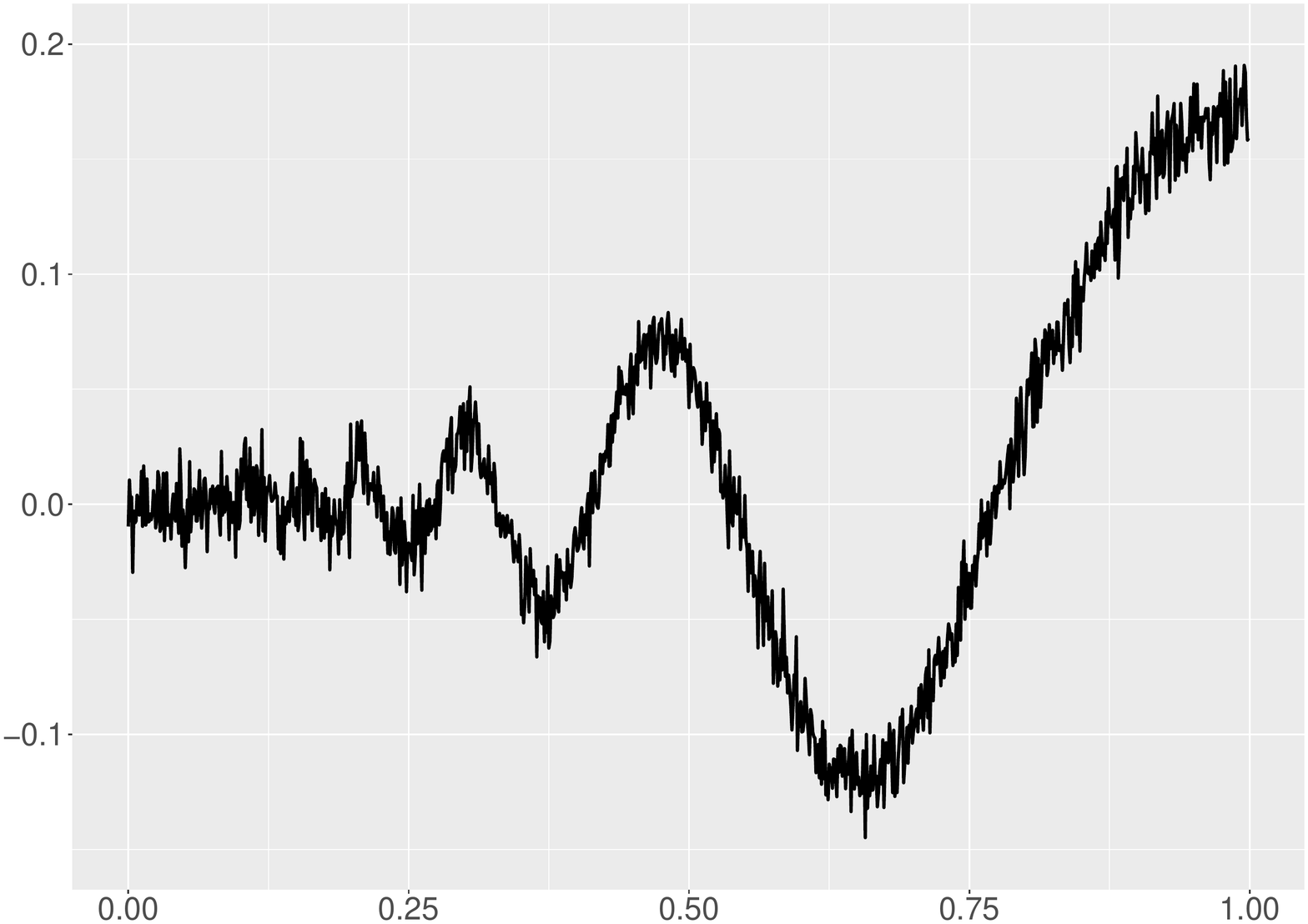}
    \caption{Noisy position at \textit{SNR}=7}
    \end{subfigure}
    \begin{subfigure}{0.45\textwidth}
    \centering
    \includegraphics[width=\linewidth,height=0.45\textwidth]{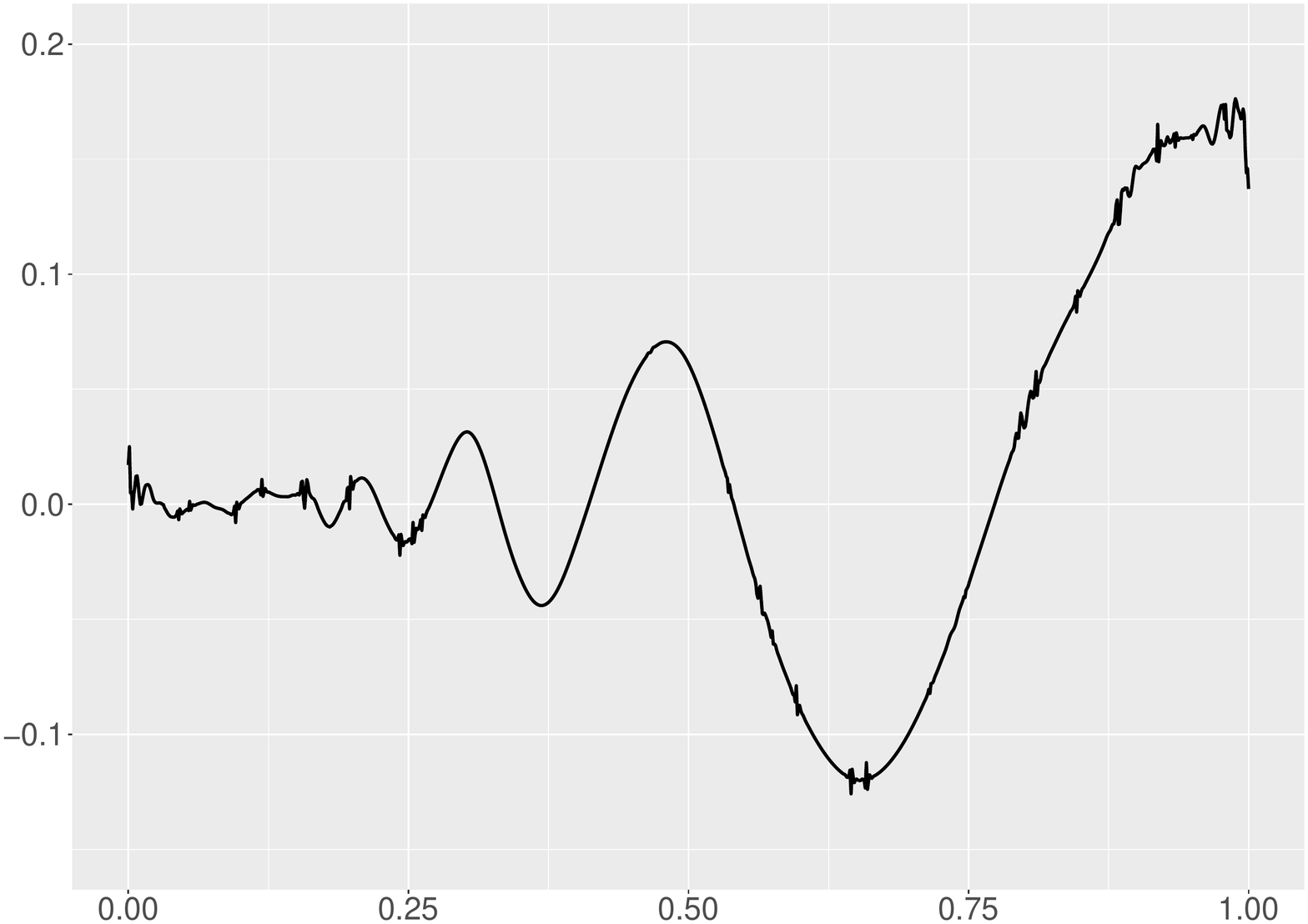}
    \caption{Reconstruction from Wavelet by sure threshold}
    \end{subfigure}
    \begin{subfigure}{0.45\textwidth}
    \centering
    \includegraphics[width=\linewidth,height=0.45\textwidth]{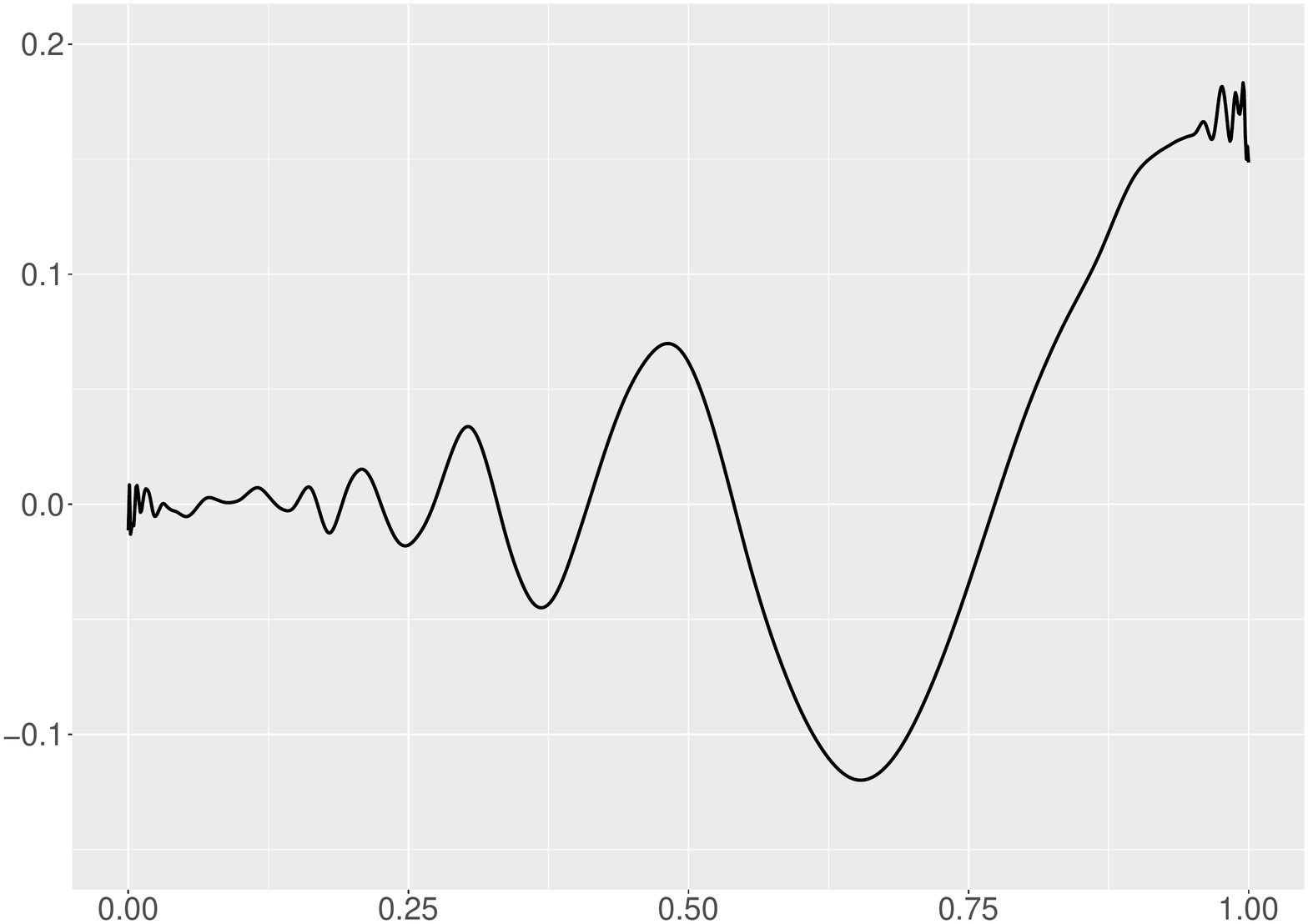}
    \caption{Reconstruction from Wavelet by BayesThresh approach}
    \end{subfigure}
    \begin{subfigure}{0.45\textwidth}
    \centering
    \includegraphics[width=\linewidth,height=0.45\textwidth]{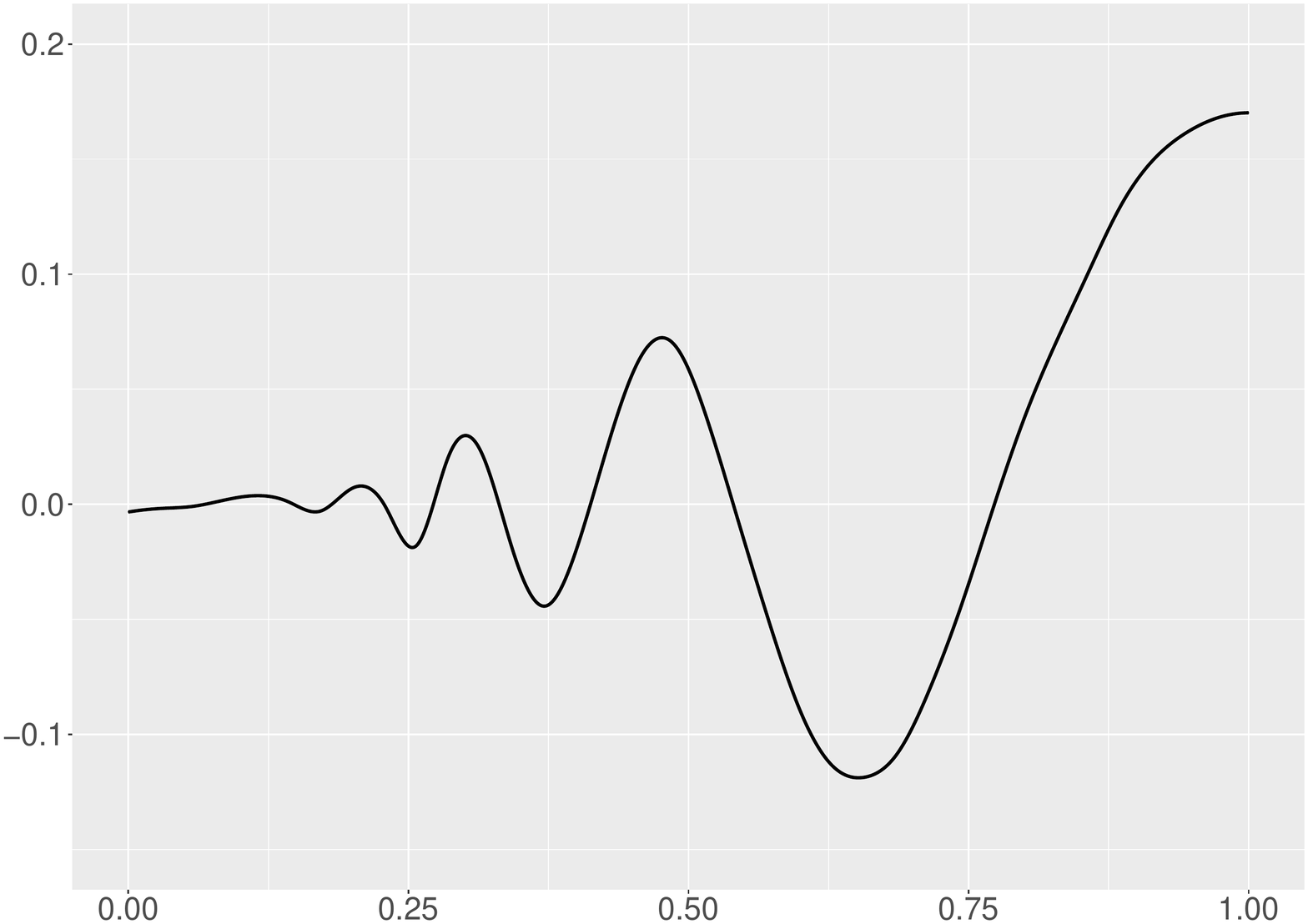}
    \caption{Reconstruction by P-spline}
    \end{subfigure}
    \begin{subfigure}{0.45\textwidth}
    \centering
    \includegraphics[width=\linewidth,height=0.45\textwidth]{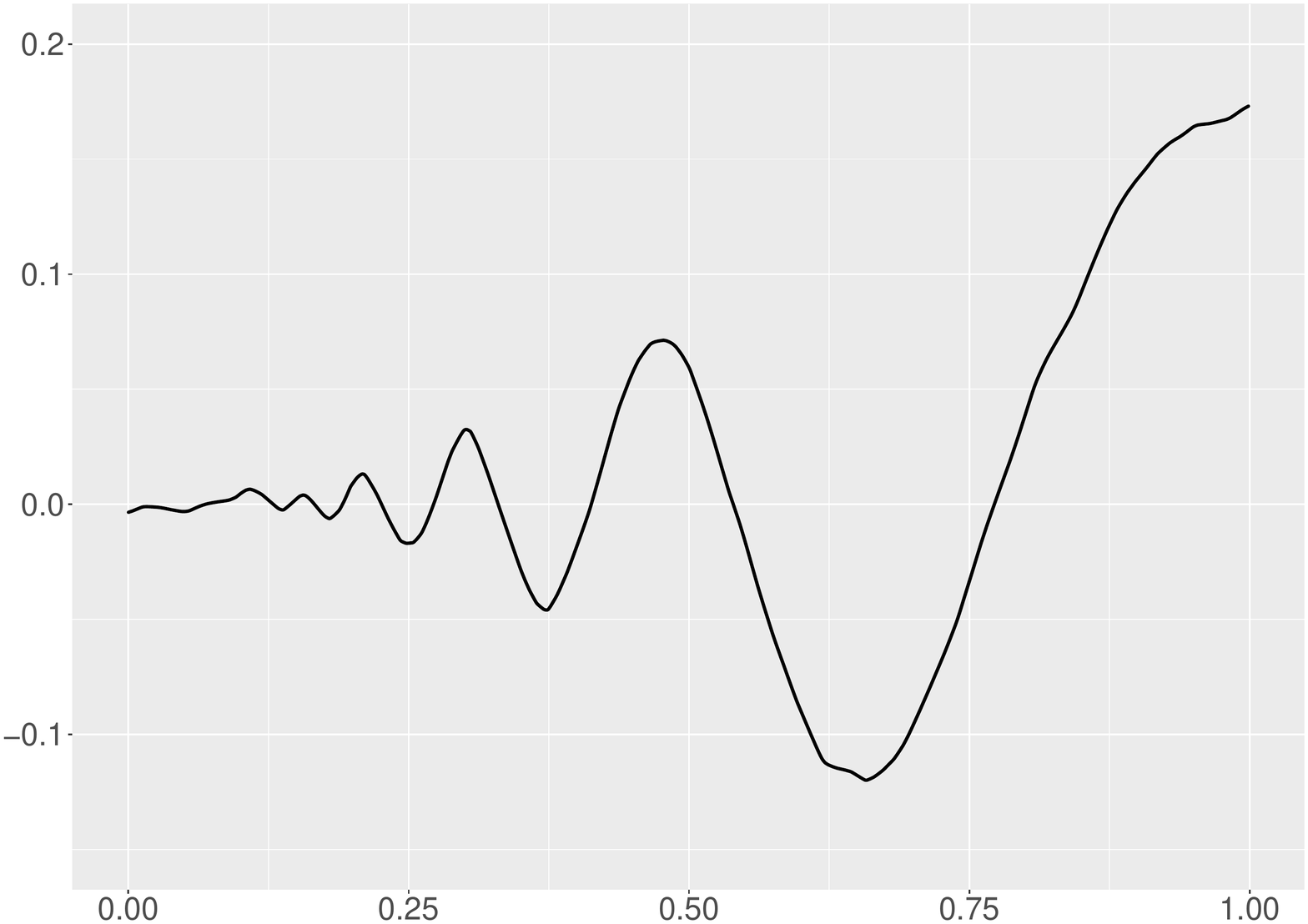}
    \caption{Reconstruction by adaptive V-spline, $\gamma=0$}
    \end{subfigure}
  \begin{subfigure}[t]{0.45\textwidth}
    \centering
    \includegraphics[width=\linewidth,height=0.45\textwidth]{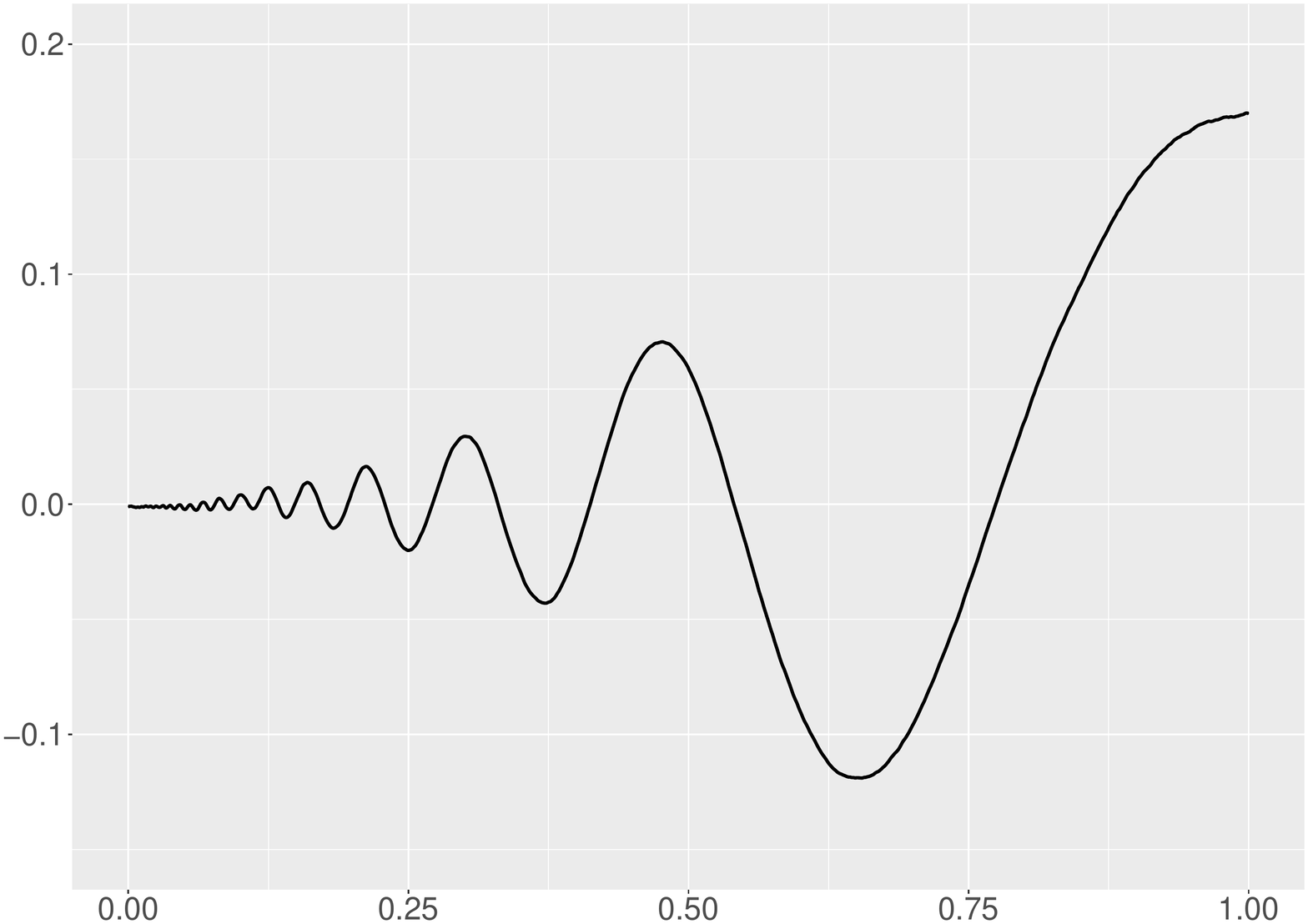}
    \caption{Reconstruction by non-adaptive V-spline}
    \end{subfigure}
    \begin{subfigure}[t]{0.45\textwidth}
    \centering
    \includegraphics[width=\linewidth,height=0.45\textwidth]{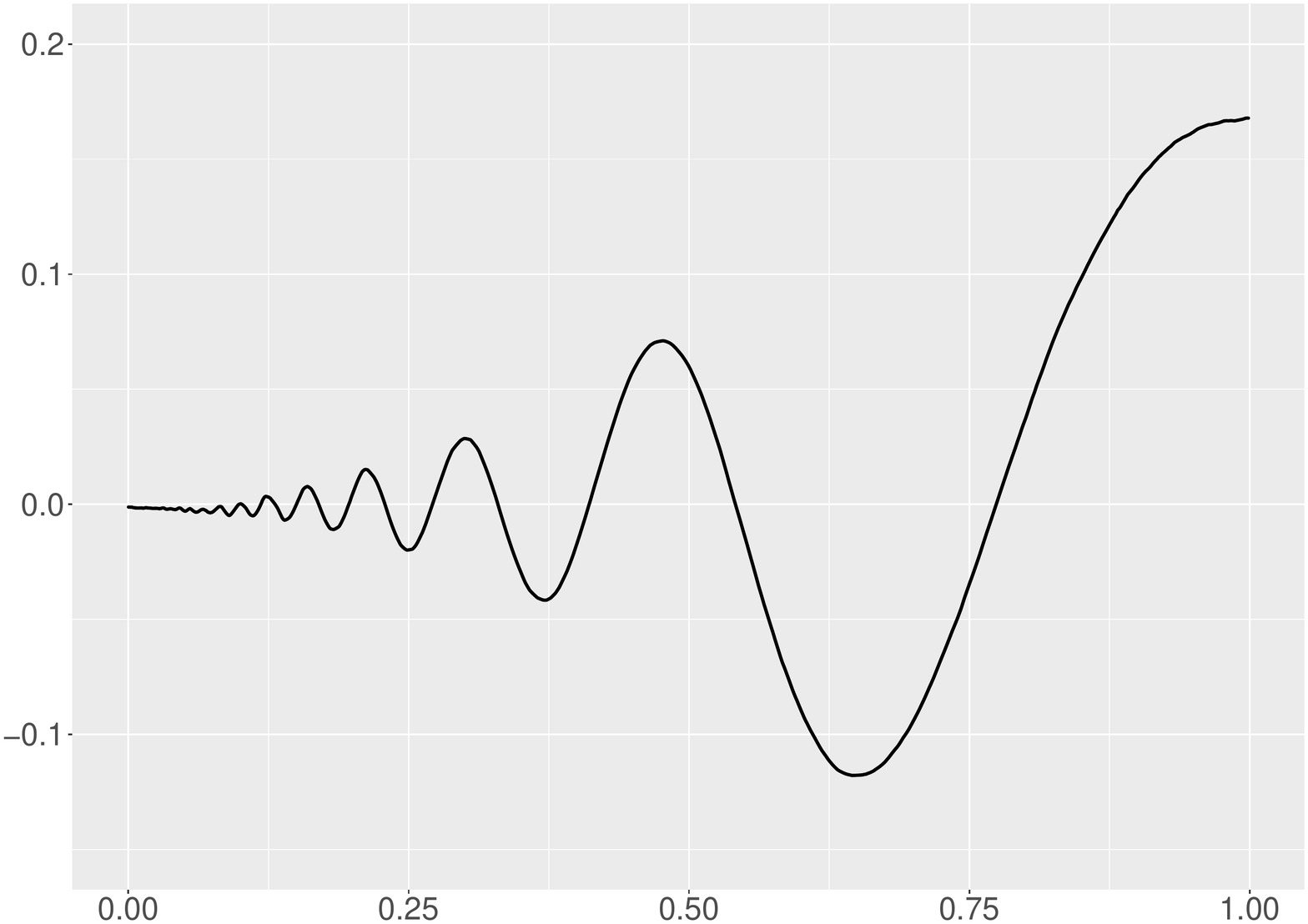}
    \caption{Reconstruction by adaptive V-spline\\\mbox{  } }
    \end{subfigure}
\caption{Numerical example: $\textit{Doppler}$. Comparison of different reconstruction methods with simulated data.}\label{num4}
 \end{figure}

Figures \ref{num1} to \ref{num4} give reconstructions of the trajectories based on the {\it Blocks}, {\it Bumps}, {\it Heaviside} and {\it Doppler} functions respectively, when SNR$=7$. Apart from the adaptive V-spline, we find all methods tend to oversmooth the reconstructions. Wavelet (\textit{sure}) suffers from an occasional ``high-frequency'' breakdown and wavelet (\textit{BayesThresh}) gives spurious fluctuations near the boundary knots. The \textit{P-spline} gives a smoother reconstruction than wavelet reconstructions, but does not perform as well as the adaptive V-spline with $\gamma=0$. The value of incorporating velocity information can be seen in the performance of the non-adaptive and adaptive V-splines. However, the non-adaptive V-spline can perform poorly when there are large absolute changes in velocity, e.g. in the \textit{Blocks} and \textit{Bumps} trajectories. Overall, the adaptive V-spline performs much better than the other methods and returns near-true reconstructions.

\begin{figure}
    \centering
    \begin{subfigure}{0.45\textwidth}
    \centering
    \includegraphics[width=\textwidth]{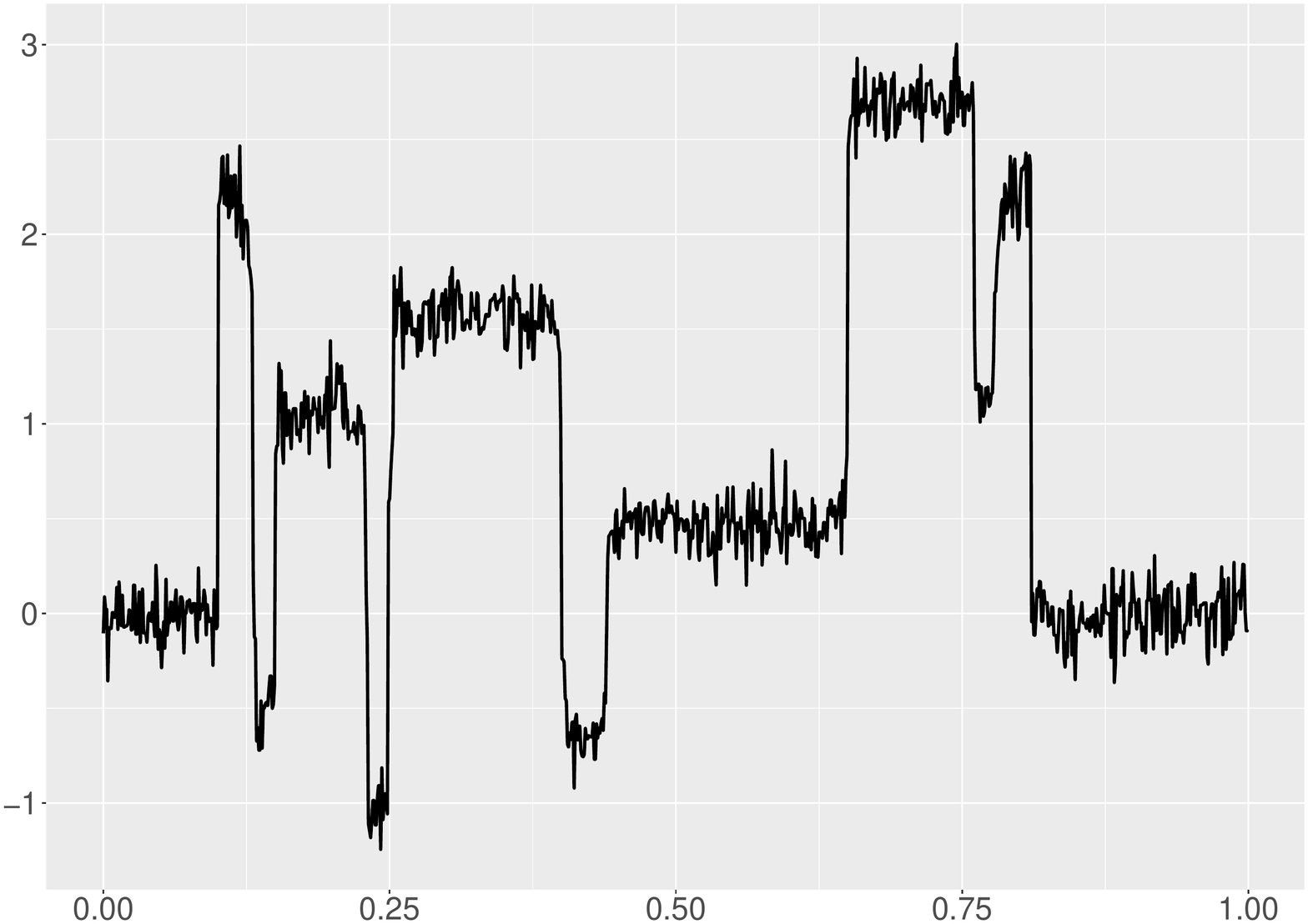}
    \caption{Estimated \textit{Blocks}  }
    \end{subfigure}%
    \begin{subfigure}{0.45\textwidth}
    \centering
    \includegraphics[width=\textwidth]{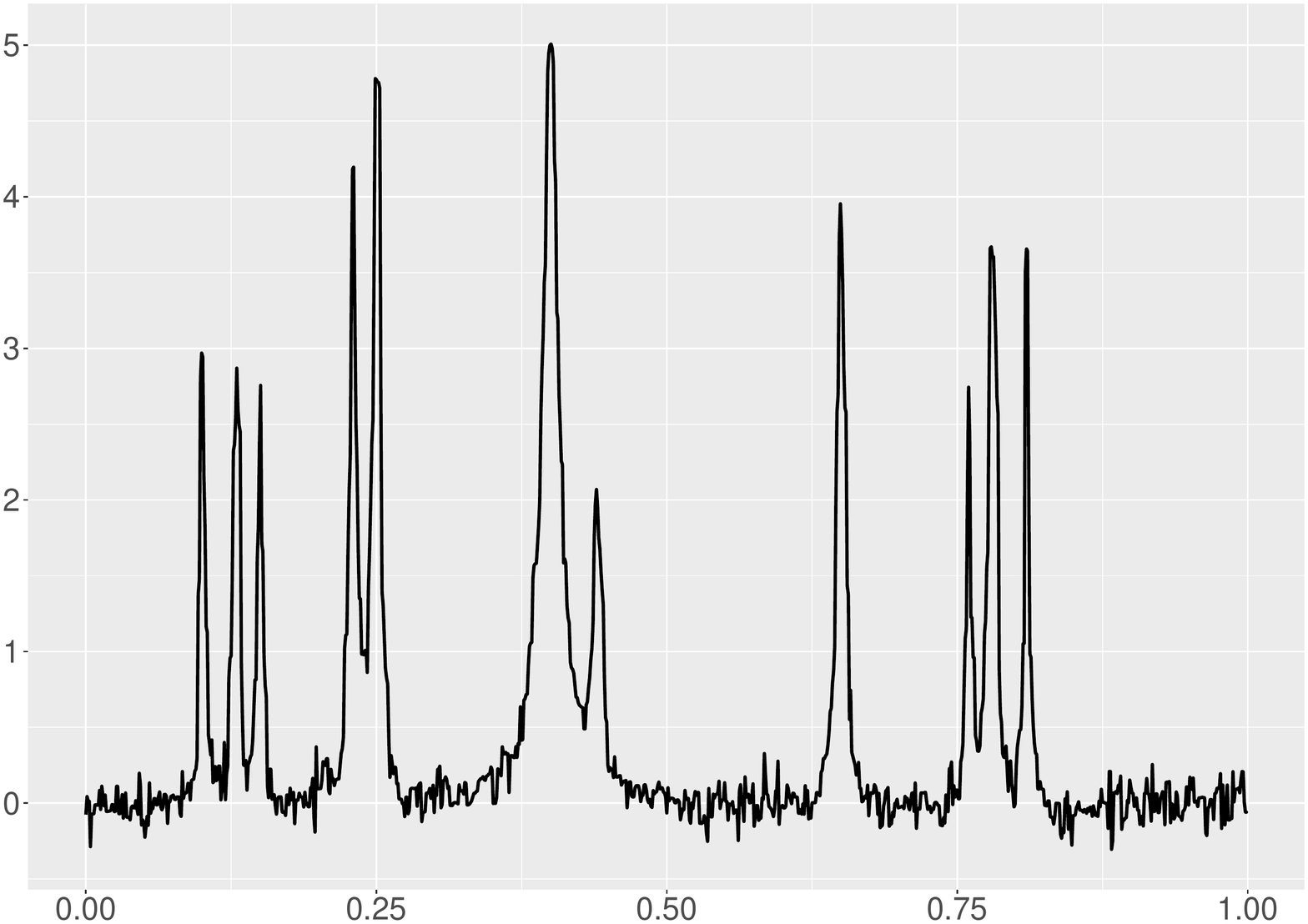}
    \caption{Estimated \textit{Bumps}  }
    \end{subfigure}
    \begin{subfigure}{0.45\textwidth}
    \centering
    \includegraphics[width=\textwidth]{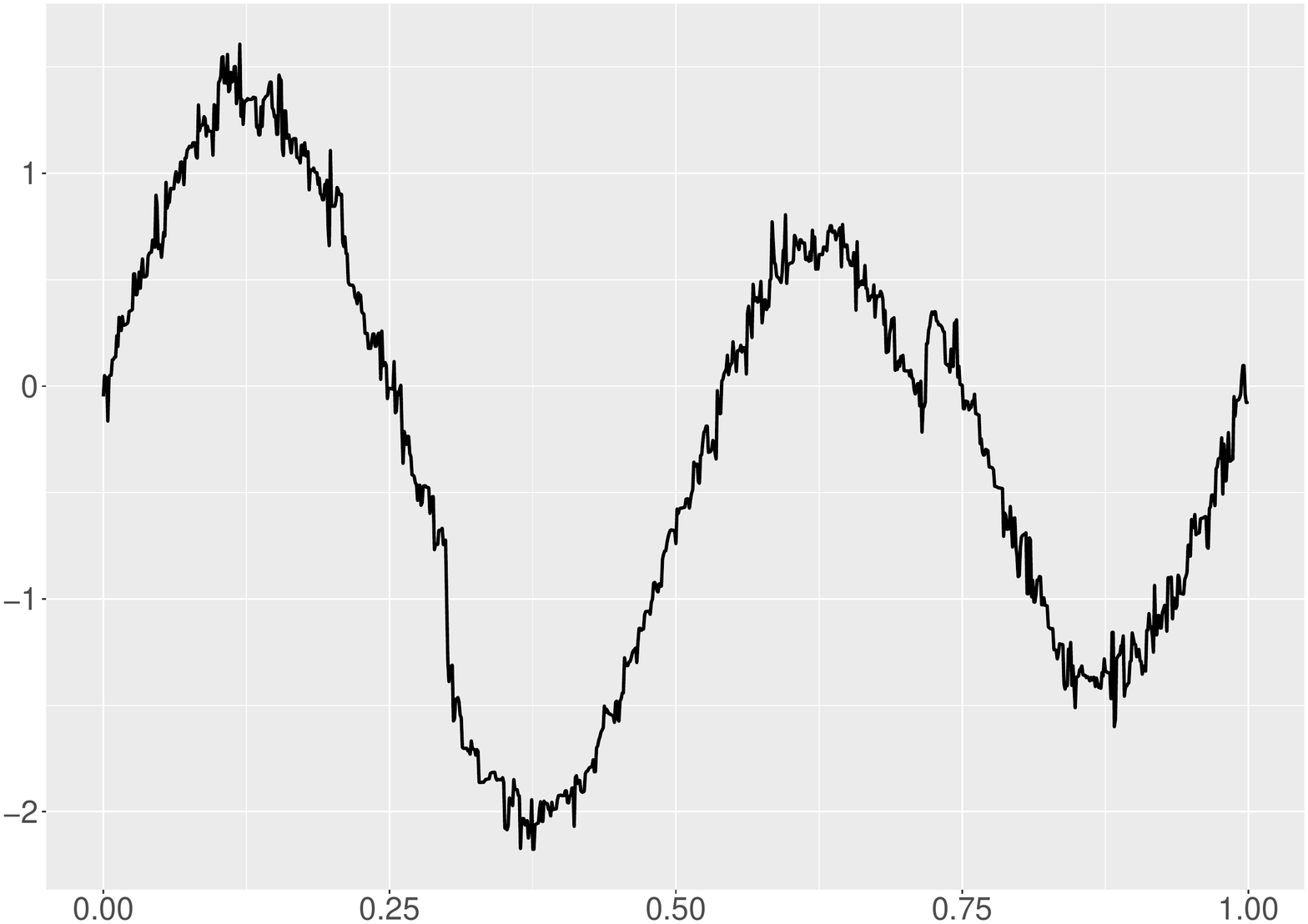}
    \caption{Estimated \textit{HeaviSine}  }
    \end{subfigure}
    \begin{subfigure}{0.45\textwidth}
    \centering
    \includegraphics[width=\textwidth]{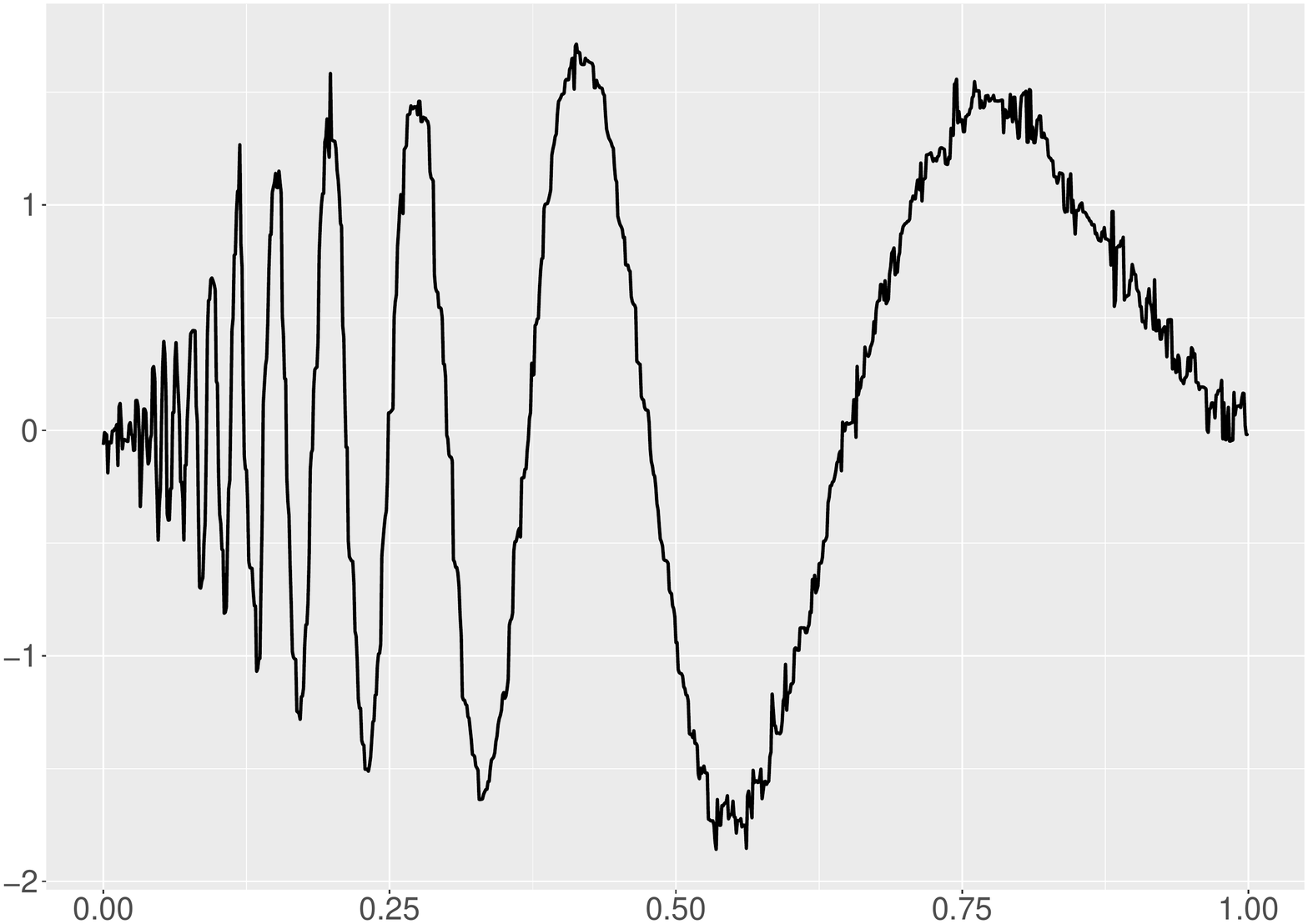}
    \caption{Estimated \textit{Doppler}  }
    \end{subfigure}
\caption{Trajectory velocities estimated by the adaptive V-spline.}\label{numvtractor}
 \end{figure}

Figure \ref{numvtractor} illustrates the trajectory velocities obtained by taking the first derivative of the adaptive V-spline trajectory. It is clear that the V-spline incorporates velocity information in order to improve trajectory reconstruction at the expense of smooth velocity reconstruction.

\subsubsection{Evaluation and residual analysis}

To examine the performance of the adaptive V-spline, we compute the true mean squared error for each of the reconstructions:
\begin{align}
%\mbox{MSE}&= \frac{1}{n} \sum_{i=1}^{n} \left( y_i-\hat{f}_{\eta,\gamma}(t_i) \right)^2,\\
\mbox{TMSE}&= \frac{1}{n} \sum_{i=1}^{n} \left( f(t_i)-\hat{f}(t_i) \right)^2.
\end{align}
The results are shown in Table \ref{tmse3200}. The adaptive V-spline returns the smallest true mean squared errors for the {\it Blocks} and {\it Bumps} trajectories. For the {\it Heaviside} and {\it Doppler} trajectories, only the non-adaptive V-spline shows similar performance. Further analysis also shows that the residuals from the V-splines are not significantly correlated at any lag, as expected.
  
%\begin{sidewaystable}
% \begin{table}
% 	\centering
% 	\caption{MSE. Mean squared errors of different methods. The numbers in bold indicate the least error among these methods under the same level. The difference is not significant.}\label{mse3200}
%	\setlength\tabcolsep{1.5pt}
%	\begin{tabular}{|c|c|C{1.9cm}|C{1.9cm}|C{1.9cm}|C{1.9cm}|C{1.9cm}|C{1.9cm}|}
%\hline	MSE $\left(10^{-4}\right)$   & SNR & V-spline & VS$_{\footnotesize\gamma=0}$ & VS$_{\scriptsize \mbox{APT=0}}$   & P-spline & W(sure)& W(Bayes)\\ \hline
%\multirow{2}{*}{\textit{Blocks}}     & 7   &  16.53& 15.99 & 16.69 & 16.14  & \textbf{15.39} & 16.68 \\ \cline{2-8}
%       & 3   &  89.79 & \textbf{87.64} & 89.94  & 88.27 & 98.35 & 90.24 \\ \hline
%\multirow{2}{*}{\textit{Bumps}}     & 7   & 4.40 & 4.19 & 4.55 & 4.33 & \textbf{4.18} & 4.59 \\ \cline{2-8}
%      & 3   & 23.93 & \textbf{23.19} & 24.10 & 23.55 & 26.23 & 23.74 \\ \hline
%\multirow{2}{*}{\textit{HeaviSine}}  & 7   & 4.16 & 4.01 &4.16 & 4.02 & \textbf{3.79} & 4.19 \\ \cline{2-8}
%     & 3   & 22.63 & \textbf{22.19} & 22.65 & 22.02 & 23.53 & 22.07 \\ \hline
%\multirow{2}{*}{\textit{Doppler}}    & 7   & 1.15 & \textbf{1.07} & 1.10 & 1.15  & \textbf{1.07} & 1.13  \\ \cline{2-8}
%      & 3   & 6.27 & \textbf{5.94} &6.28 & 6.05  & 6.85 & 6.29  \\ \hline
%	\end{tabular}
%\end{table}
 %\end{sidewaystable}

%\begin{sidewaystable} 
\begin{table}
	\centering
	\caption{True mean squared errors (TMSE) of different reconstruction methods. The numbers in bold indicate the least error in each scenario. VS$_{\scriptsize \lambda_0}$ refers to the non-adaptive V-spline; VS$_{\footnotesize\gamma=0}$ refers to the adaptive V-spline with $\gamma=0$; W refers to wavelet reconstruction.}\label{tmse3200}
	\setlength\tabcolsep{1.5pt}
	\begin{tabular}{|c|c|C{1.9cm}|C{1.9cm}|C{1.9cm}|C{1.9cm}|C{1.9cm}|C{1.9cm}|}
\hline	TMSE $\left(10^{-6}\right)$  & SNR & V-spline & VS$_{\footnotesize\gamma=0}$ & VS$_{\scriptsize \lambda_0}$   & P-spline & W(sure) &  W(Bayes)\\ \hline
		
\multirow{2}{*}{\textit{Blocks}}  & 7   & \textbf{1.75} & 54.25 &  28.68   & 54.76   & 201.02   & 182.12   \\ \cline{2-8}
	     & 3   & \textbf{16.44} & 152.5 & 30.76  & 171.59   & 1138.08  & 712.36  \\ \hline
\multirow{2}{*}{\textit{Bumps}}     & 7  & \textbf{1.64} & 23.44  & 21.10     & 24.21 & 71.71 & 69.26 \\ \cline{2-8}
        & 3  & \textbf{8.51} & 77.78  &37.12     & 77.52 & 330.77 & 238.79 \\ \hline
\multirow{2}{*}{\textit{HeaviSine}}  & 7 & \textbf{1.53}& 7.80  & 1.56     & 9.54   & 55.37  &44.88  \\ \cline{2-8}
      & 3 & \textbf{8.21}& 33.56  & 8.49 & 34.26 & 240.72& 110.49\\ \hline
\multirow{2}{*}{\textit{Doppler}}    & 7   & 1.51& 6.67  & \textbf{1.08}   &  8.26   & 14.87  & 12.01  \\ \cline{2-8}
    & 3   & \textbf{8.10} & 22.14  & 8.25   & 19.95    &81.48  &50.33   \\ \hline
	\end{tabular}	
\end{table}
%\end{sidewaystable}

%\subsubsection{Residual Analysis}

%The simulated data is generated by equations \eqref{tractorsplinegeneratefunctions} and the SNRs are set at 7 and 3 separately to compare the performances of different algorithms. All of the algorithms can reconstruct the true trajectory from noisy data and return acceptable MSE values, though V-spline returns the least TMSE in most of the circumstances. 

Table \ref{tablecompareSNR} shows the ability of the adaptive V-spline to retrieve the true SNR, calculated by $\sigma_{\hat{f}} / \sigma_{(\hat{f}-y)}$.
%The measurements are generated from $f$ and $g$ with predefined SNR. The V-spline reconstructs the true trajectory and retrieves the SNR value, both of which are close to the truth. 

\begin{table}
	\centering
    \caption{Retrieved SNR.}\label{tablecompareSNR}
	\begin{tabular}{|c|C{3cm}|C{3cm}|C{3cm}|}
\hline	 SNR   & predefined value & generated $f$ & V-spline $\hat{f}$ \\ \hline
\multirow{2}{*}{\textit{Blocks}}  & 7   & 6.9442    &  6.9485     \\ \cline{2-4}
		   & 3   &  2.9761   &  2.9817   \\ \hline
\multirow{2}{*}{\textit{Bumps}}    & 7  & 6.9442    &  6.9548  \\ \cline{2-4}
		   & 3  & 2.9761    &   2.9953 	   \\ \hline
\multirow{2}{*}{\textit{HeaviSine}}  & 7 & 6.9442    &   6.9207   \\ \cline{2-4}
		  & 3 & 2.9761    &   2.9706  \\ \hline
\multirow{2}{*}{\textit{Doppler}}     & 7   & 6.9442   &  6.8757   \\ \cline{2-4}
		  & 3   & 2.9761   &  2.9625   \\ \hline
	\end{tabular}
\end{table}

\subsection{Irregularly Sampled Time Series Data}

%A set of irregularly sampled time series data has different time differences between each pair of successive points. The distribution of $\Delta T_i$ is not uniform.

In this section, we show that the proposed V-spline has the ability to reconstruct the true trajectory even though the data is irregularly sampled. For each of the functions of the previous section, we set SNR to 7 and generate a mother simulation of length $n=1024$. We then obtain a regularly sampled daughter simulation of length 512 using the indices $1,3,\ldots,1023$, and an irregularly sampled daughter simulation with 512 indices chosen at random. %We then compare the reconstructions in terms of TMSE and in ability to retrieve the true SNR.

Table \ref{tablecompareSMEIreReg} shows that the TMSE tends to increase when the data are irregularly sampled. In the case of {\it Blocks} and {\it Bumps}, the TMSE increases by a factor 3 or 4. However the ability to retrieve the true SNR does not appear to be affected; see Table \ref{tablecompareSNRIreReg}.

%See figure \ref{gghistIrregularTime}. 

%\begin{figure}[!h]
%    \centering
%    \includegraphics[width=0.75\textwidth]{ggplot/gghistIrregularTime.pdf}
% \caption{Histogram of $\Delta T$ for irregularly sampled data}\label{gghistIrregularTime}
% \end{figure}

%The reconstructions of regularly and irregularly sampled data are very competitive. 

The reconstructions themselves are shown in Figure \ref{irregularFigure}. In the irregularly sampled cases, the reconstruction tends to smooth some features in the underlying trajectories. In {\it Bumps} and, to a lesser extent, in {\it Doppler}, the reconstruction occasionally adds in features. This occurs when there are no sampled data at times when the underlying velocity is changing rapidly.

\begin{figure}[!h]
    \centering
    \begin{subfigure}{\textwidth}
    \centering
    \includegraphics[width=0.45\textwidth]{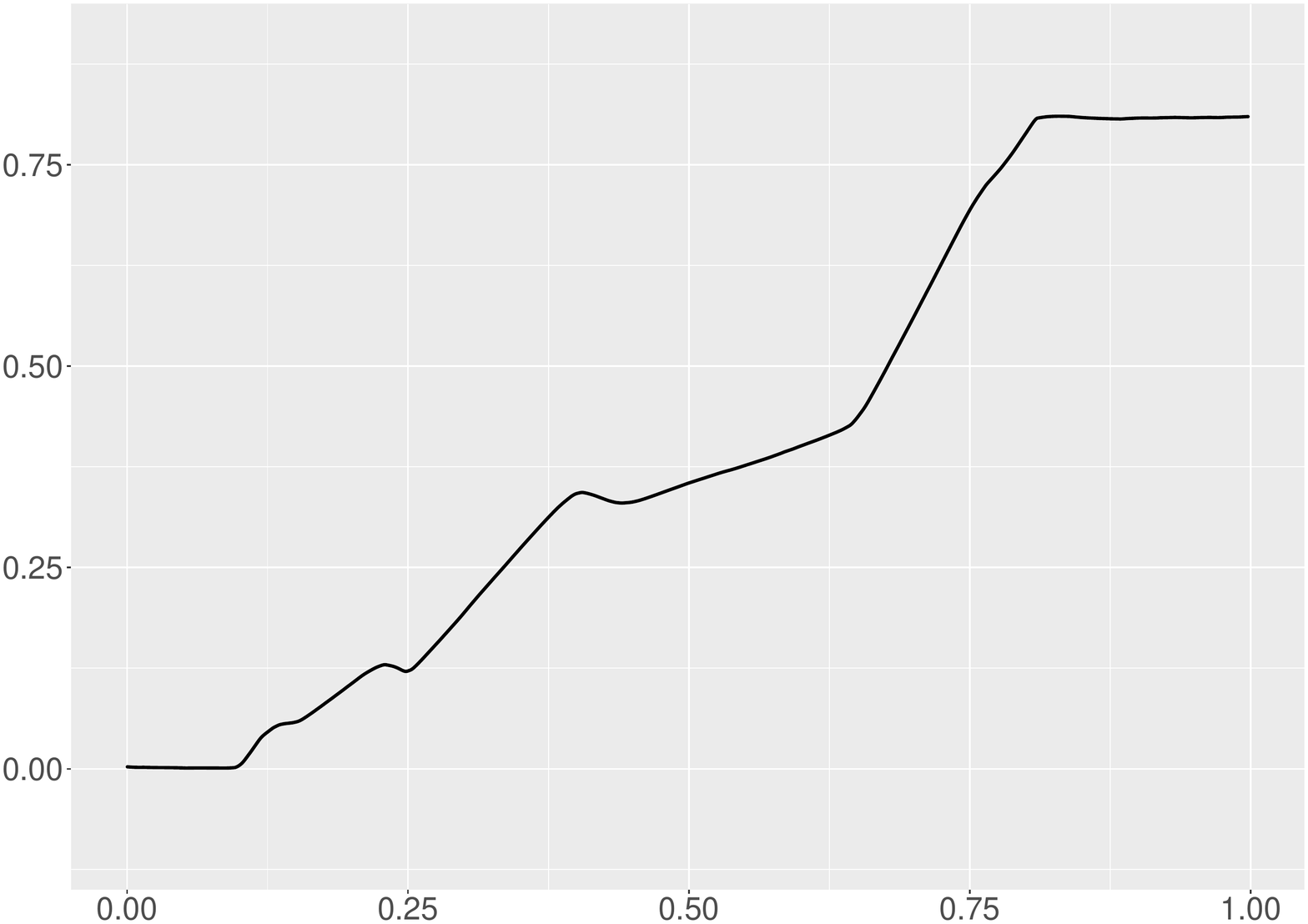}
    \includegraphics[width=0.45\textwidth]{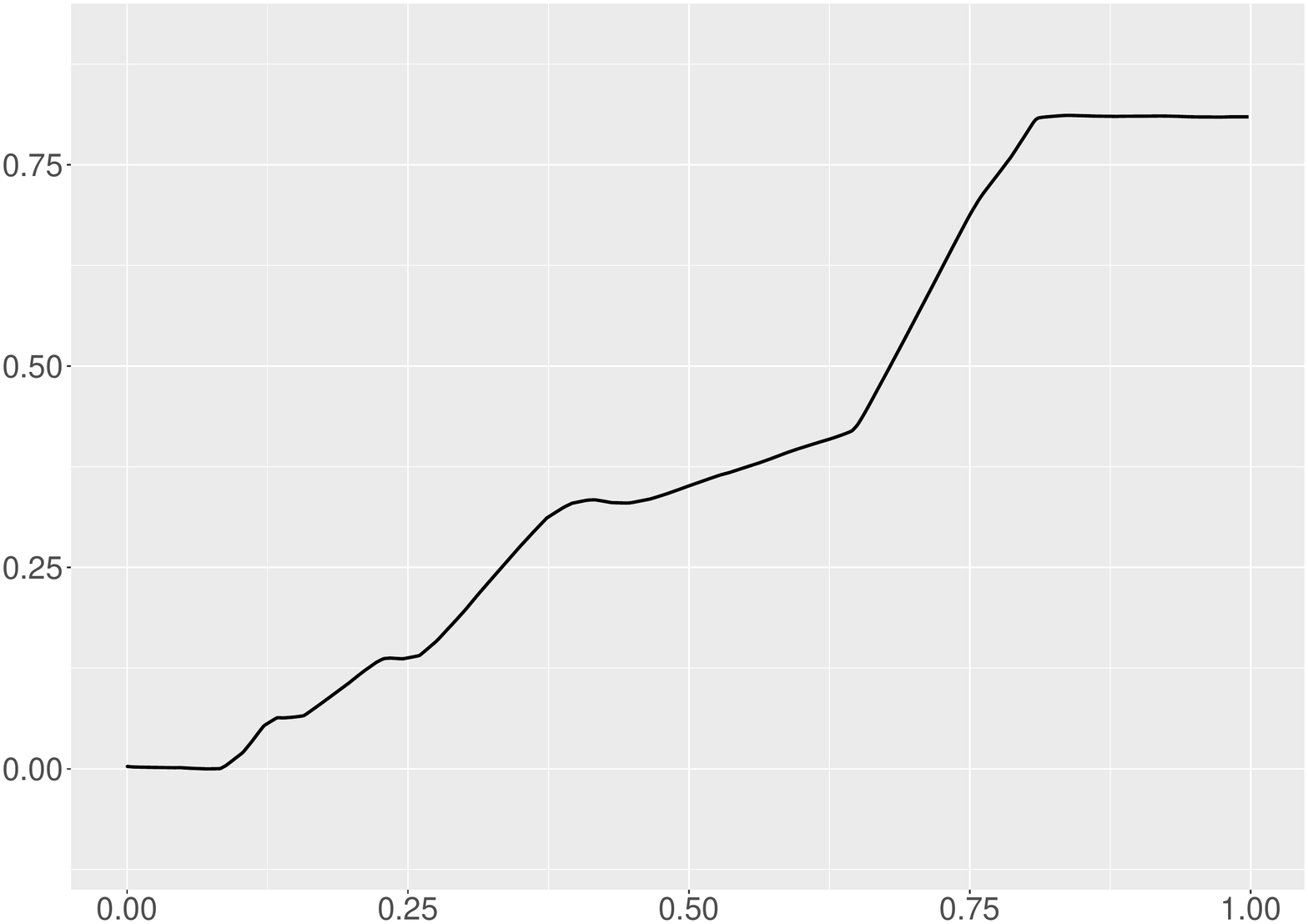}
    \caption{Reconstruction of \textit{Blocks} from regularly and irregularly sampled data}
    \end{subfigure}
    \begin{subfigure}{\textwidth}
    \centering
    \includegraphics[width=0.45\textwidth]{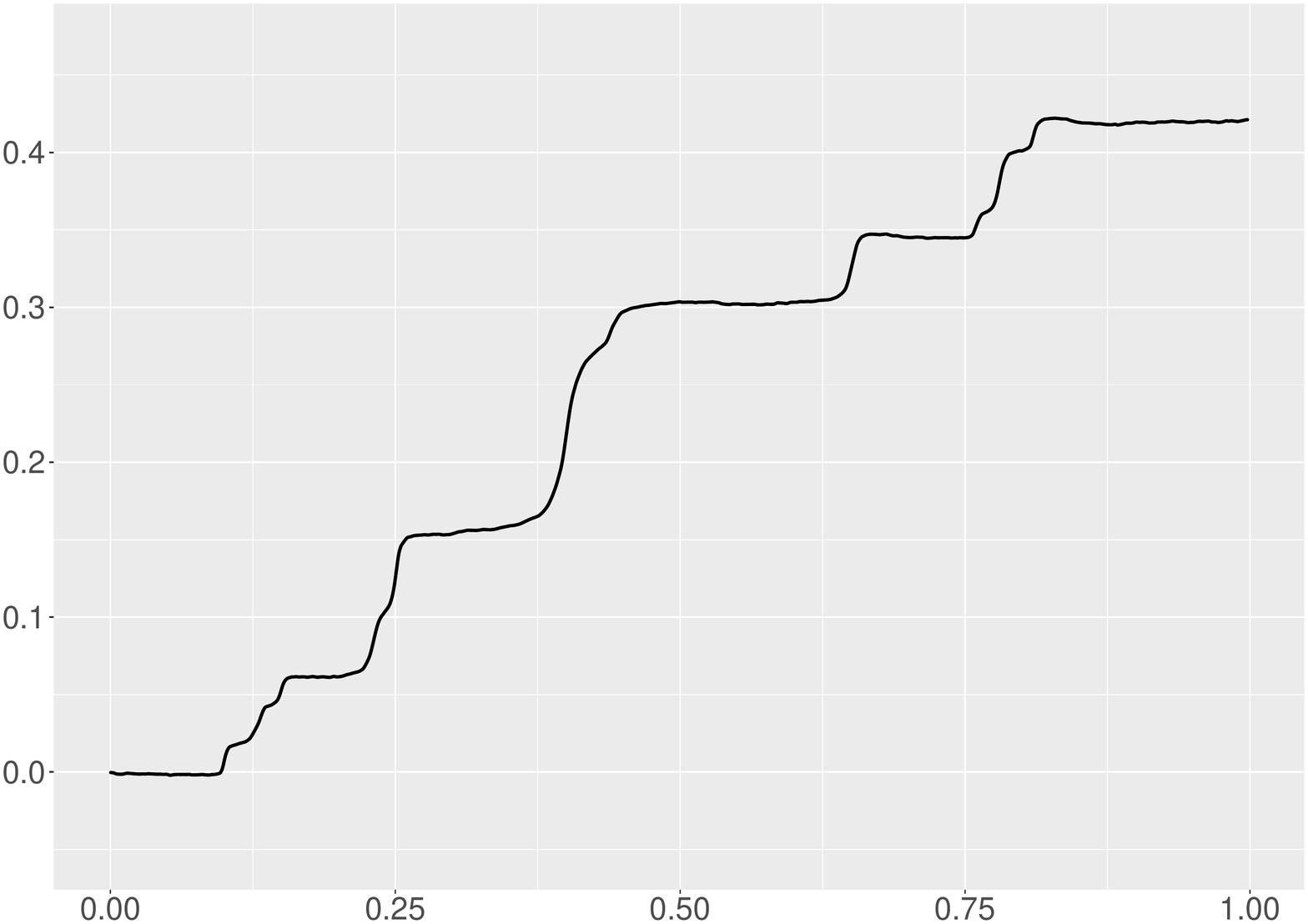}
    \includegraphics[width=0.45\textwidth]{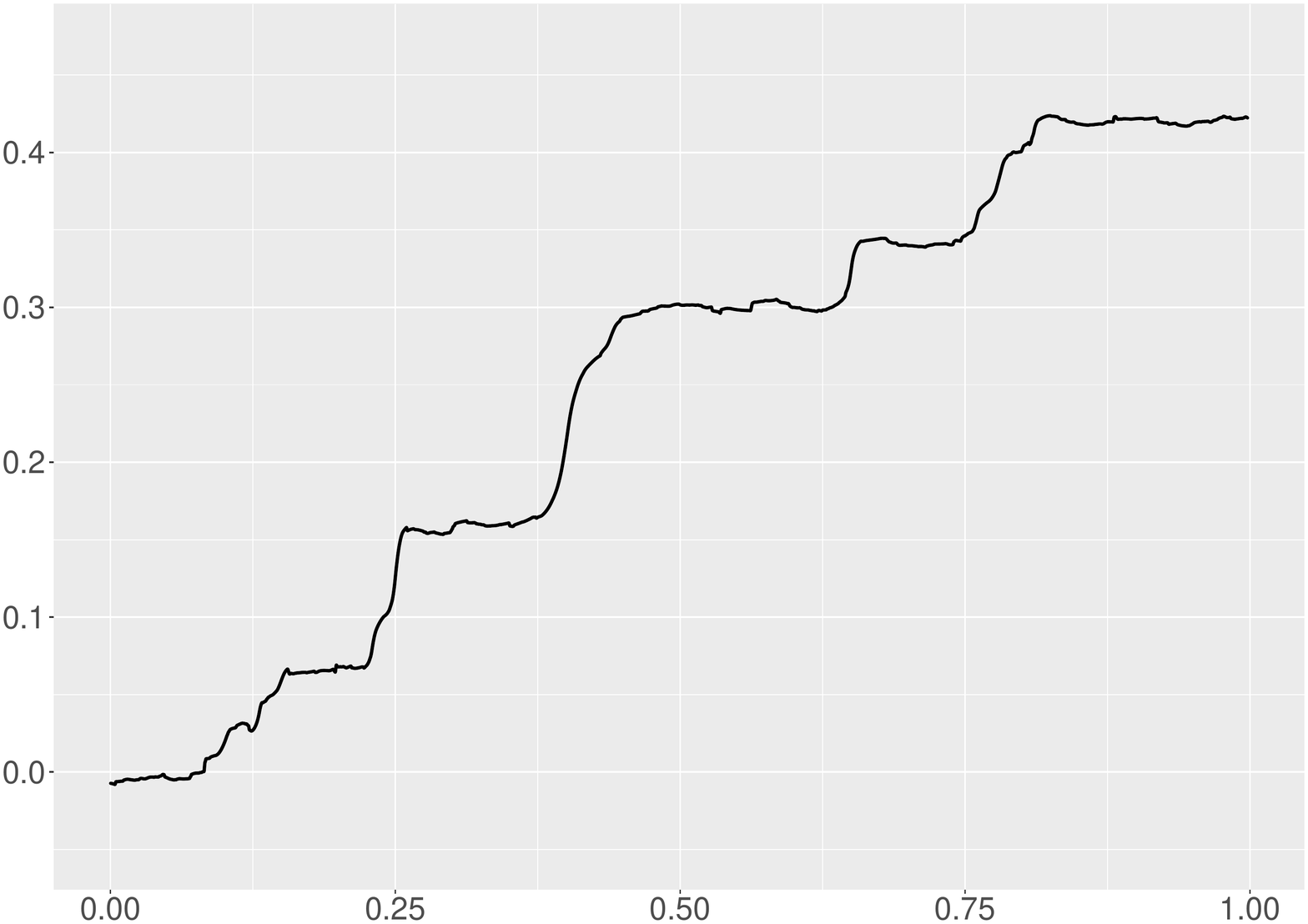}
    \caption{Reconstruction of \textit{Bumps} from regularly and irregularly sampled data}
    \end{subfigure}
    \begin{subfigure}{\textwidth}
    \centering
    \includegraphics[width=0.45\textwidth]{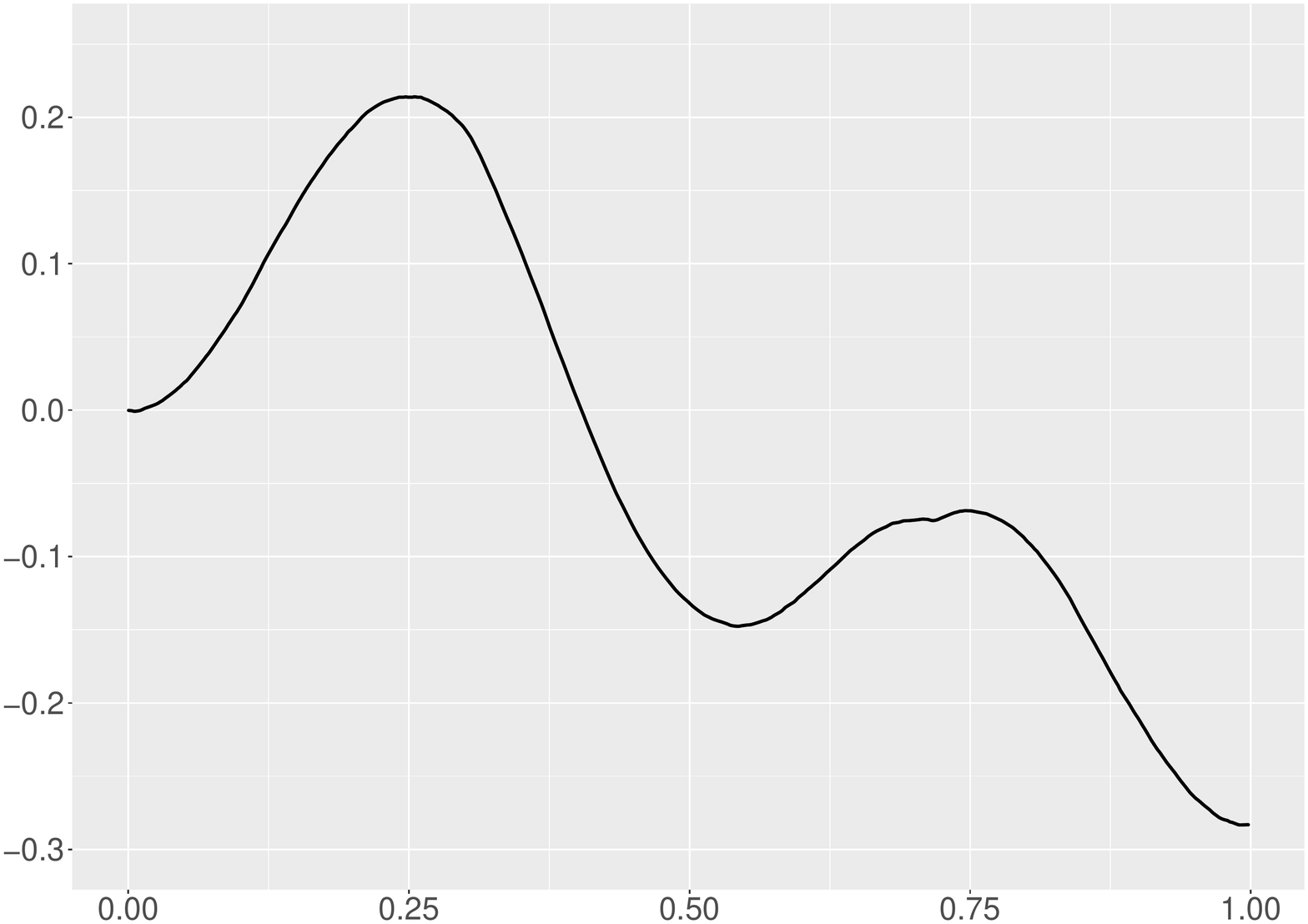}
    \includegraphics[width=0.45\textwidth]{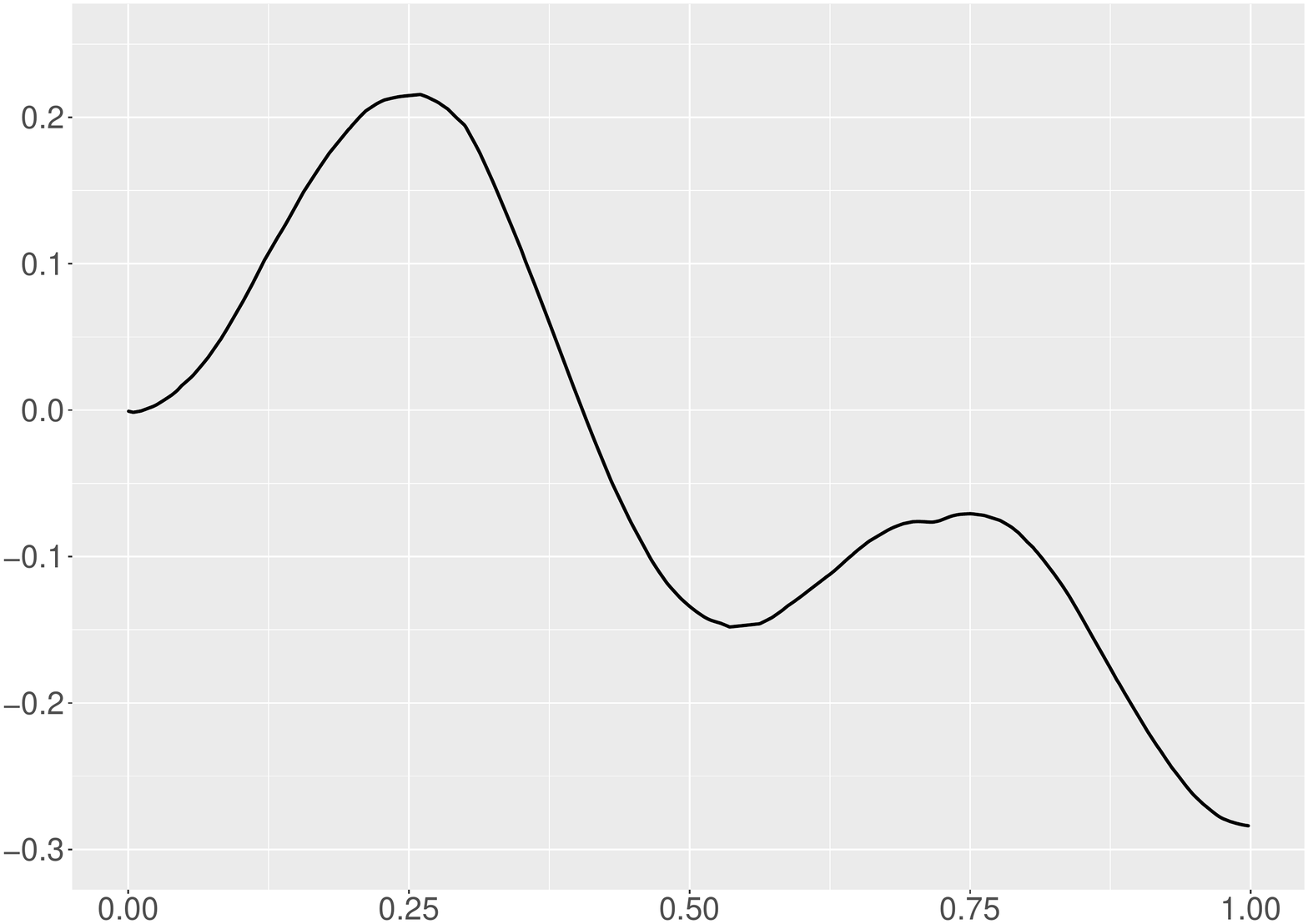}
    \caption{Reconstruction of \textit{HeaviSine} from regularly and irregularly sampled data}
    \end{subfigure}
    \begin{subfigure}{\textwidth}
    \centering
    \includegraphics[width=0.45\textwidth]{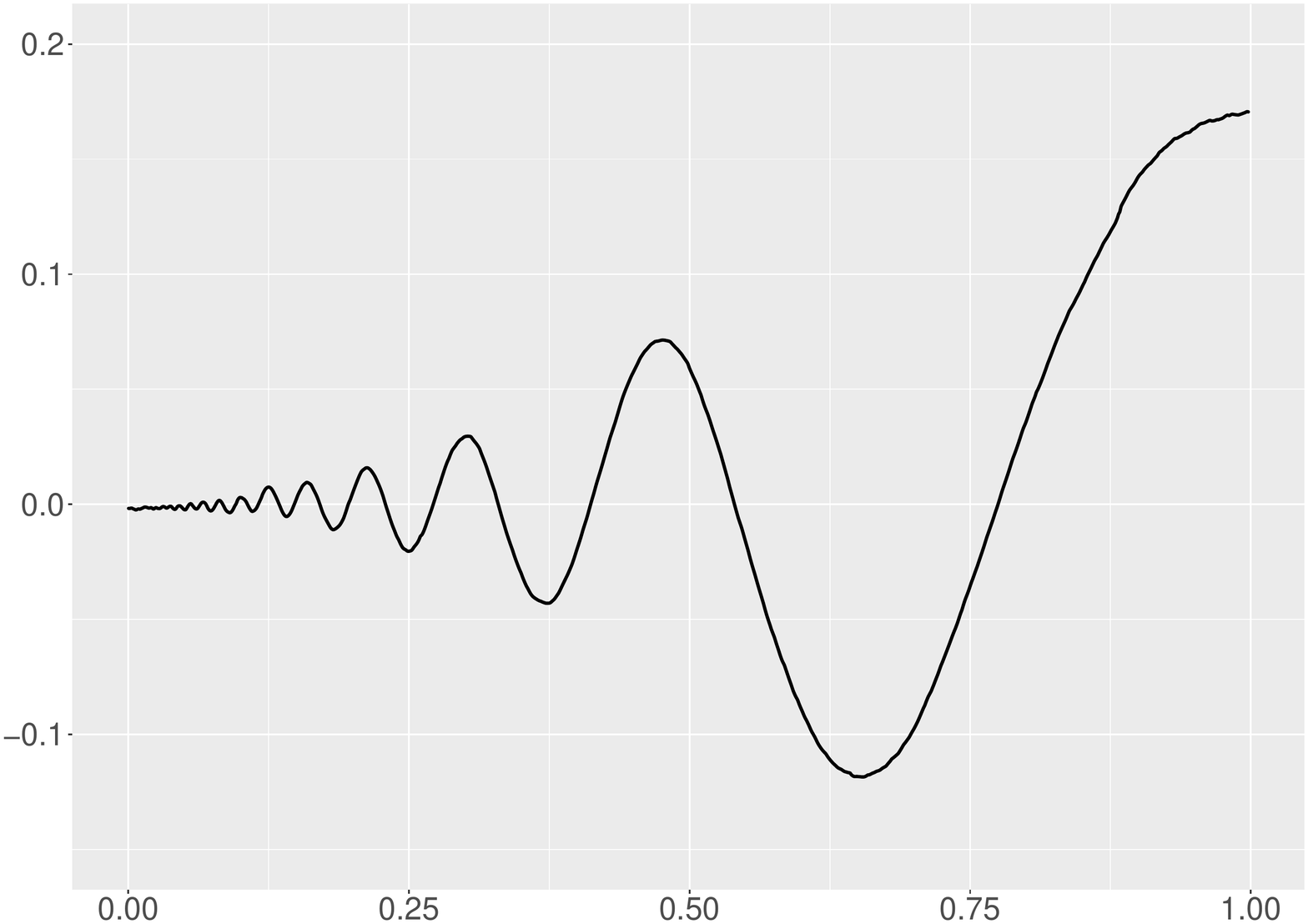}
    \includegraphics[width=0.45\textwidth]{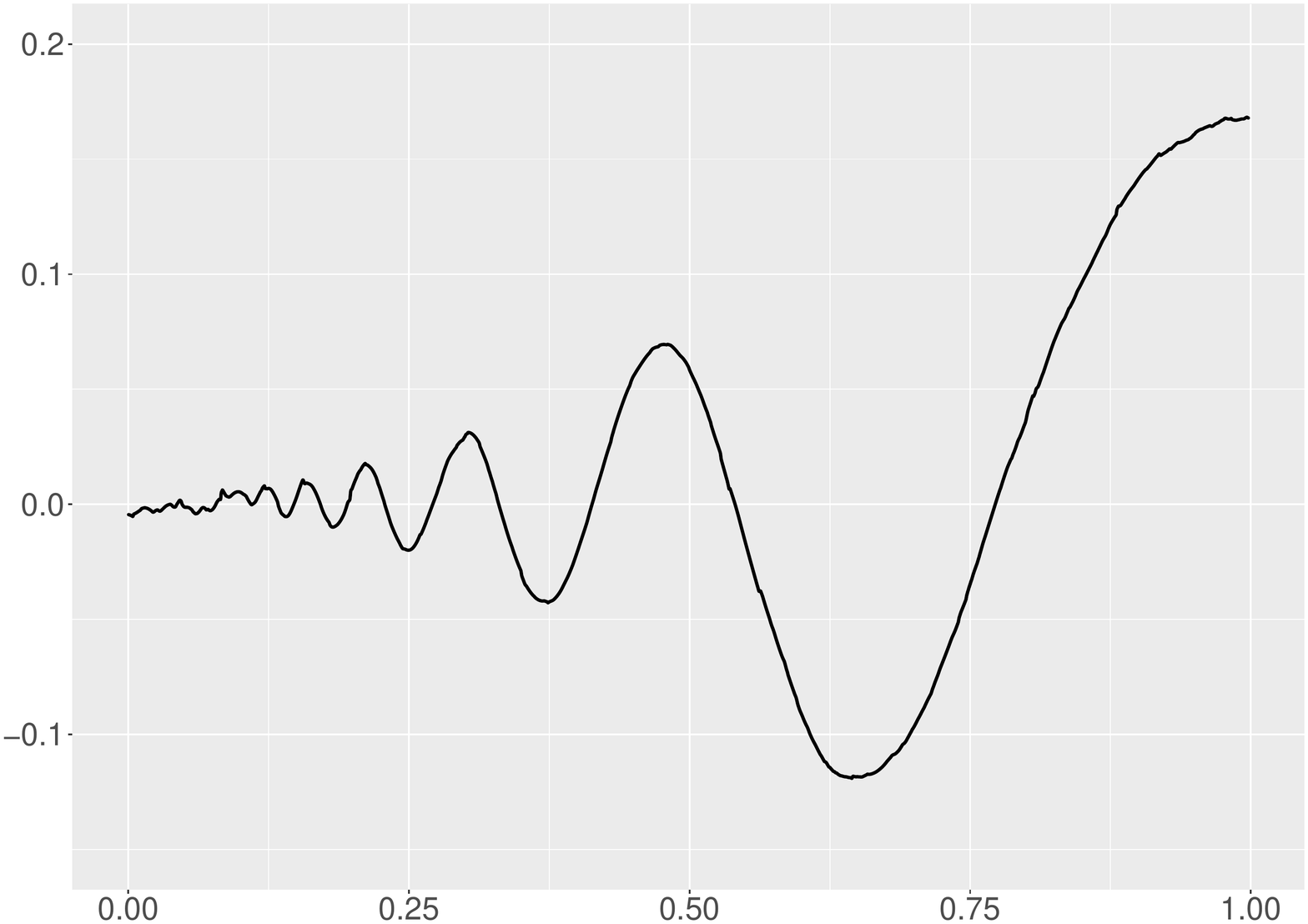}
    \caption{Reconstruction of \textit{Doppler} from regularly and irregularly sampled data}
    \end{subfigure}
 \caption{Comparison of adaptive V-spline trajectory reconstructions of regularly and irregularly sampled data.}\label{irregularFigure}
 \end{figure}

%\begin{table}
%	\centering
%    \caption{MSE and TMSE of reconstructions from regularly and irregularly sampled data  }\label{tablecompareSMEIreReg}
%	\begin{tabular}{|c|C{2cm}|C{2cm}|C{2cm}|C{2cm}|} \hline	
%	& \multicolumn{2}{|c|}{MSE $\times 10^{-4}$} & \multicolumn{2}{c|}{TMSE $\times 10^{-6}$} \\ \cline{2-5}
%	                 & Regular & Irregular & Regular & Irregular \\ \hline
%\textit{Blocks}        &    8.0260 &  8.3358  & 3.5197 & 10.8596  \\  \hline
%\textit{Bumps}       &    2.1374  & 2.0203  & 1.6662 & 6.2586 \\  \hline
%\textit{HeaviSine}  &   2.0232   & 2.1272  &  1.1275 & 1.1077  \\ \hline
%\textit{Doppler}     &  0.5251 & 0.5219 & 1.0101 & 1.7832   \\ \hline
%	\end{tabular}
%\end{table}

\begin{table}
	\centering
    \caption{TMSE of reconstructions from regularly and irregularly sampled data  }\label{tablecompareSMEIreReg}
	\begin{tabular}{|c|C{2cm}|C{2cm}|} \hline	
	 & \multicolumn{2}{|c|}{TMSE $\times 10^{-6}$} \\ \cline{2-3}
	                  & Regular & Irregular \\ \hline
\textit{Blocks}         & 3.5197 & 10.8596  \\  \hline
\textit{Bumps}        & 1.6662 & 6.2586 \\  \hline
\textit{HeaviSine}    &  1.1275 & 1.1077  \\ \hline
\textit{Doppler}      & 1.0101 & 1.7832   \\ \hline
	\end{tabular}
\end{table}

\begin{table}
	\centering
    \caption{Retrieved SNR of reconstructions from regularly and irregularly sampled data  }\label{tablecompareSNRIreReg}
	\begin{tabular}{|c|C{3cm}|C{3cm}|}
\hline	 SNR           & Regular & Irregular  \\ \hline
\textit{Blocks}        &    7.0479  & 6.8692    \\  \hline
\textit{Bumps}       &    7.0241  & 7.1937     \\  \hline
\textit{HeaviSine}  &   7.2367    & 6.8793   \\ \hline
\textit{Doppler}     &    6.8692   & 7.3645    \\ \hline
	\end{tabular}
\end{table}

%\clearpage 

\section{Inference of Tractor Trajectories}\label{splineapplication}

Real world vehicle trajectories exhibit many of the features seen in the test trajectories considered in the previous section, particularly the \textit{Blocks}, \textit{Bumps} and \textit{HeaviSine} trajectories. In this section, we extend the V-spline to more than one dimension and apply it to a real dataset, which was obtained from a GPS unit mounted on a tractor working in a horticultural setting. The data was irregularly recorded with highly variable time differences $\Delta T_i$. The original dataset contains $n=928$ records of longitude, latitude, speed, bearing and the status of the tractor's boom sprayer. These data were converted into northing and easting, and the velocities in those directions. The boom status, ``up'' or ``down'', denotes the operational state of the tractor, and may indicate different types of trajectories.

\subsection{The V-spline in $d$-dimensions}

To generalize the V-spline to $d$-dimensions, we consider the situation preceding eq. \eqref{tractorsplineObjective} but where now $y_i,v_i\in\mathbb{R}^d$. Then the function $f:[a,b]\rightarrow\mathbb{R}^d$ is a $d$-dimensional V-spline if it minimizes: 
\begin{equation}\label{tractorsplineObjective2D}
J[f]= \frac{1}{n} \sum_{i=1}^{n} \lVert y_i-f(t_i)\rVert_2^2 + \frac{\gamma}{n} \sum_{i=1}^{n} \lVert v_i-f'(t_i) \rVert_2^2 +\sum_{i=0}^{n} \lambda_i\int_{t_i}^{t_{i+1}} \lVert f''(t)\rVert_2^2 dt,
\end{equation}
where $\Vert\cdot\Vert_2$ is the Euclidean norm in $d$-dimensions. For each direction $\alpha=1,\ldots,d$, the fitted V-spline has the form $\hat{f}^{\alpha}(t)=\sum_{k=1}^{2n} N_k(t)\hat{\theta}^{\alpha}_k$, where
\begin{equation}\label{thetahat_d}
\hat{\theta}^{\alpha}=\left(B^\top B+\gamma C^\top C+n\Omega_{\lambda}\right)^{-1}\left(B^\top\mathbf{y}^{\alpha}+\gamma C^\top\mathbf{v}^{\alpha}\right).
\end{equation}
The parameters $\lambda$ and $\gamma$ are estimated by minimizing the cross-validation score:
\begin{equation}\label{tractorcv_d}
\mbox{CV}\left(\lambda,\gamma\right)=\frac{1}{n}\sum_{i=1}^{n} \left\Vert \frac{y_i-\hat{f}(t_i)+\gamma \frac{T_{ii}}{1-\gamma V_{ii}}(v_i-\hat{f}'(t_i))}{1-S_{ii}-\gamma\frac{T_{ii}}{1-\gamma V_{ii}}U_{ii}} \right\Vert_2^2
\end{equation}

\begin{figure}
\centering
 \begin{subfigure}{0.45\textwidth}
    \centering
    \includegraphics[width=\textwidth]{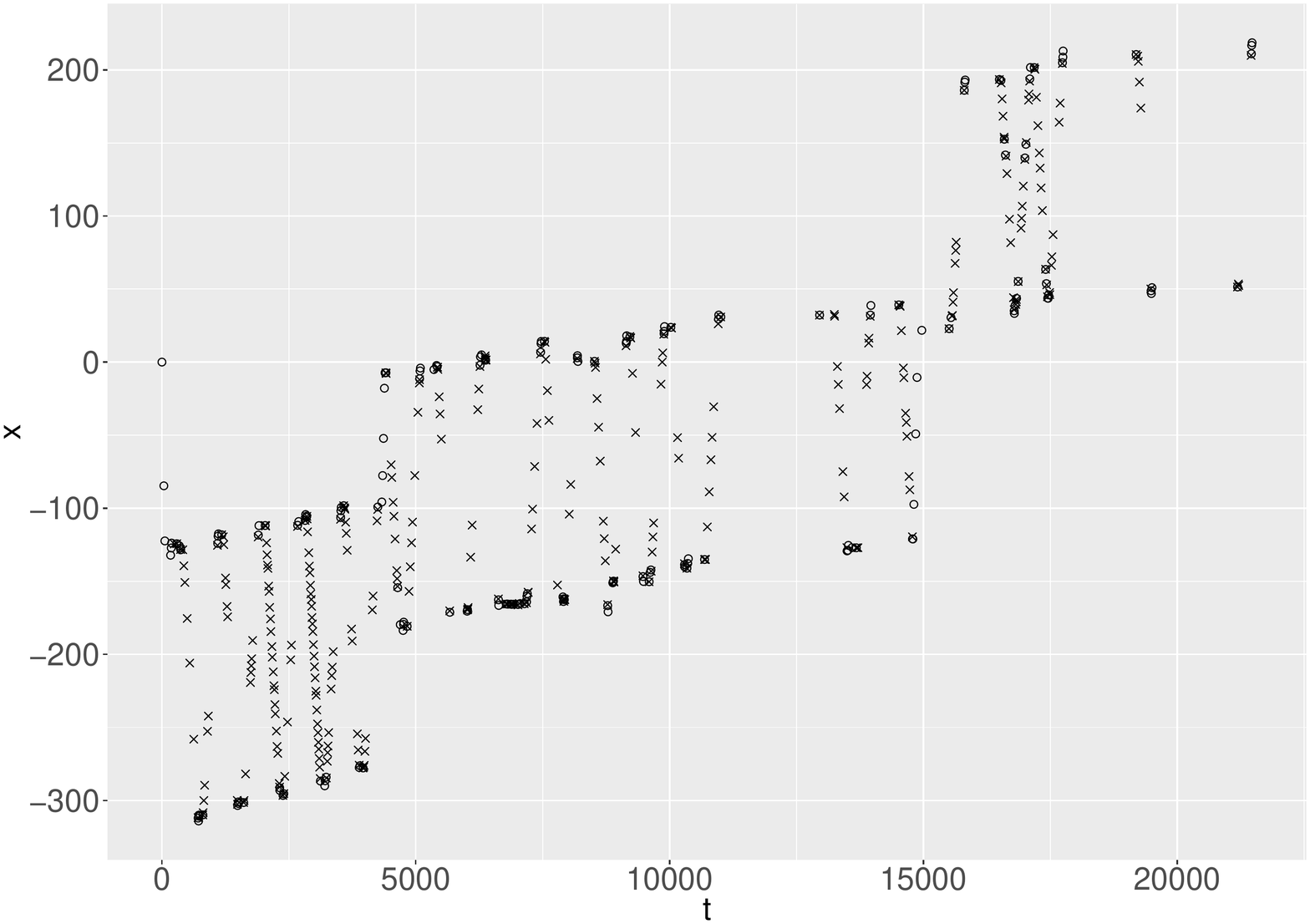}
    \caption{Easting $x$ vs time $t$}\label{gg512PointsX}
    \end{subfigure}
    \begin{subfigure}{0.45\textwidth}
    \centering
    \includegraphics[width=\textwidth]{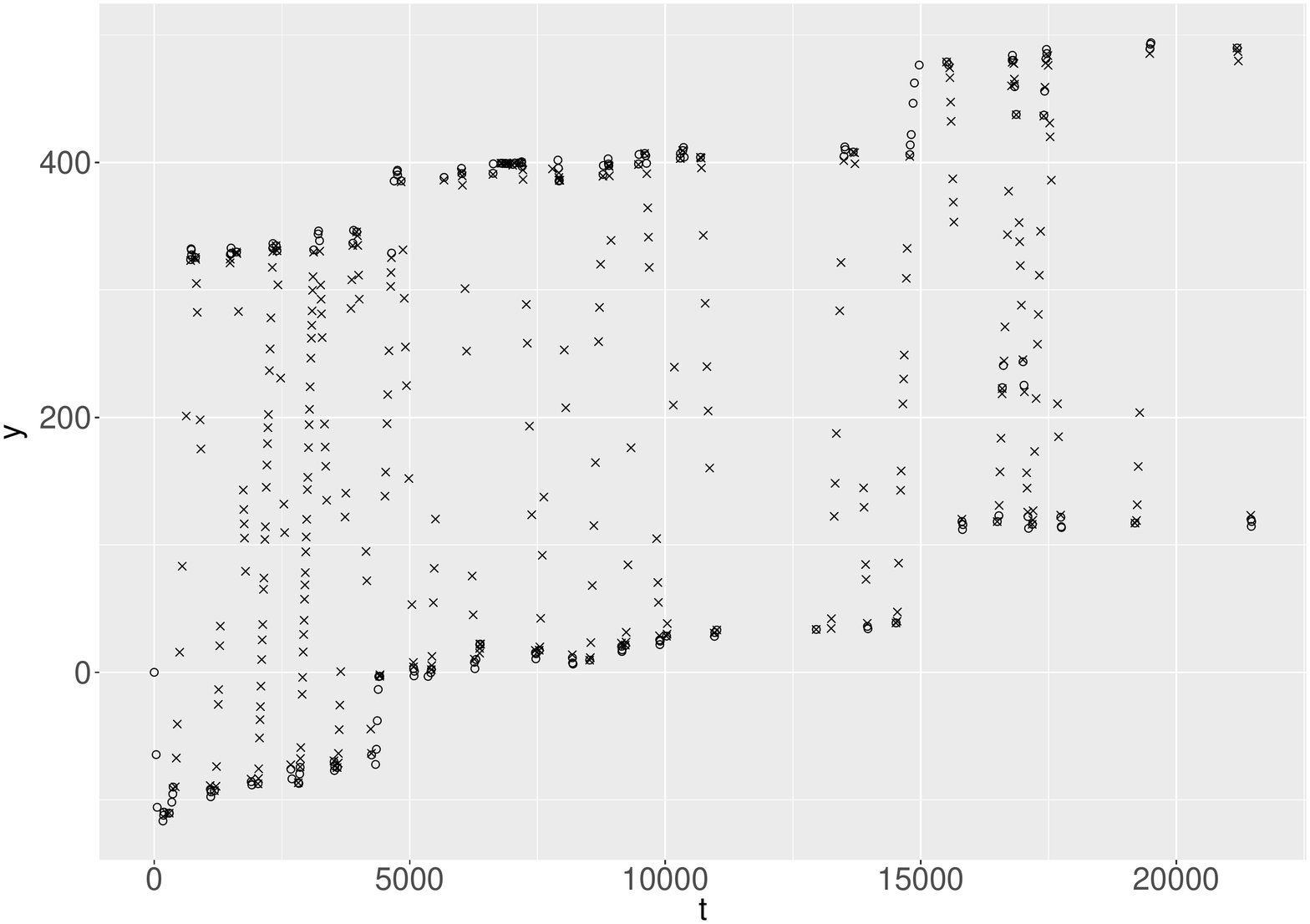}
    \caption{Northing $y$ vs time $t$}\label{gg512PointsY}
    \end{subfigure}
    \begin{subfigure}{0.45\textwidth}
    \centering
    \includegraphics[width=\textwidth]{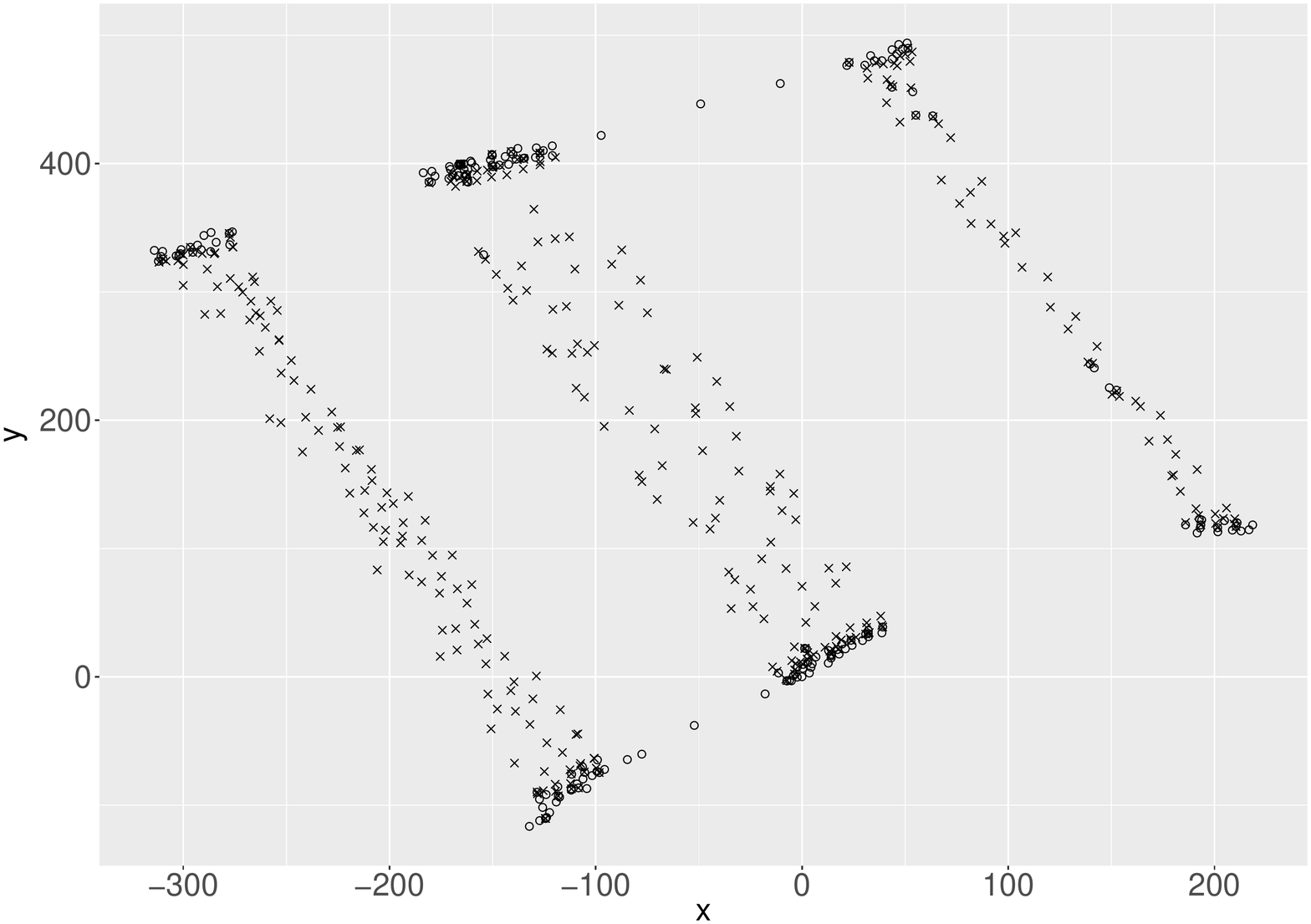}
    \caption{Northing $y$ vs easting $x$}\label{gg512Points}
    \end{subfigure}%
    \begin{subfigure}{0.45\textwidth}
    \centering
    \includegraphics[width=\textwidth]{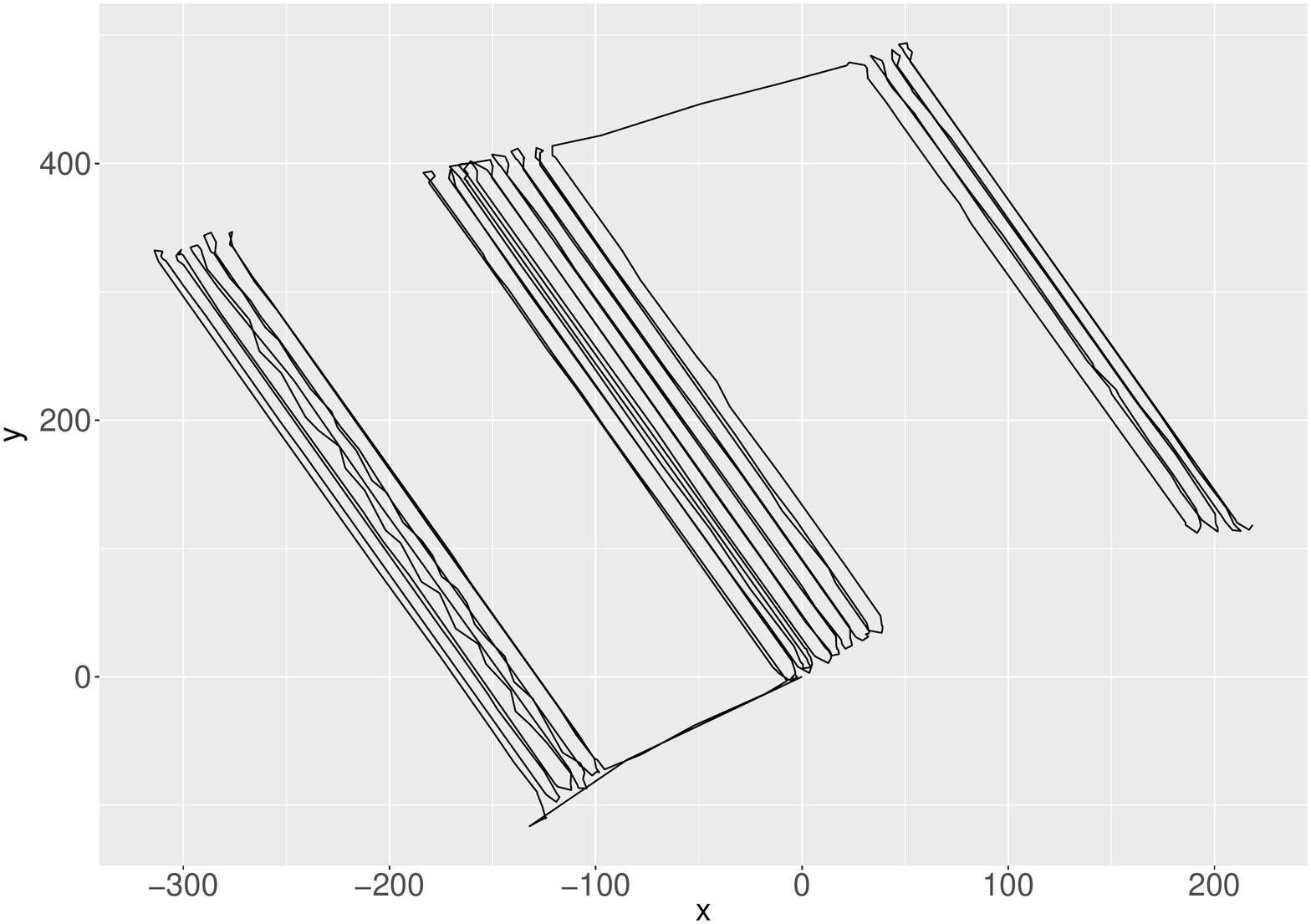}
    \caption{Line-based reconstruction}\label{gg512Path}
    \end{subfigure}
\caption{First 512 observations of position. Circles indicate the tractor boom is up; crosses indicate the boom is down. Figure \ref{gg512Path} is the line-based reconstruction in which consecutive points are connected with straight lines.}\label{original512}
 \end{figure}

%In order to fit the real data, we bring the parameter $\eta_d$ to our model. Then, we are now having three parameters $\eta_d$ and $\eta_u$ regarding boom status and $\gamma$ controlling velocity residuals. The criteria of a good fitting are that it can catch more information, recognize time gaps between two points where tractor stops and return a smaller MSE. 

In what follows, we allow the non-adaptive and adaptive V-splines to depend on the boom status. Specifically, letting $b_i=0$ denote boom ``up'', $b_i=1$ denote boom ``down'', and $\bar{v}_i=\Vert y_{i+1}-y_i\Vert_2/\Delta T_i$ be the average velocity on the interval $[t_i,t_{i+1})$, the penalty term for the non-adaptive V-spline is
\begin{equation}\label{penaltylamb0}
\lambda_i=b_i\lambda_d+(1-b_i)\lambda_u,
\end{equation}
and for the adaptive V-spline it is
\begin{equation}\label{penaltylamb}
\lambda_i=\big\{b_i\eta_d+(1-b_i)\eta_u\big\}\frac{\Delta T_i}{\bar{v}_i^2}.
\end{equation}

\subsection{Trajectory reconstruction}

It is instructive to first consider what happens when the northing and easting trajectories are reconstructed separately. To facilitate comparison with the wavelet approach we restrict attention to the first 512 observations. The data are shown in Figure \ref{original512}. It is evident that the trajectory of the tractor is typically characterized by motion along rows with boom down, and tight turns at the ends of rows with boom up. In this section, we use $x$ to denote easting and $y$ to denote northing.

Figure \ref{1dy} compares different methods for reconstructing the northing trajectory. The P-spline reconstruction fails due to instances where $y_{i+1}=y_i$. Wavelet ({\em sure}) overshoots at turning points, as does the non-adaptive V-spline, though more smoothly. The non-adaptive V-spline overshoots because it is incorporating outdated velocity information during a period where the velocity is changing quickly. On the other hand, wavelet ({\em BayesThresh}), the adaptive V-spline with $\gamma=0$ and the adaptive V-spline give acceptable results.

\begin{figure}
    \centering
    \begin{subfigure}{0.45\textwidth}
    \centering
    \includegraphics[width=\linewidth,height=0.5\textwidth]{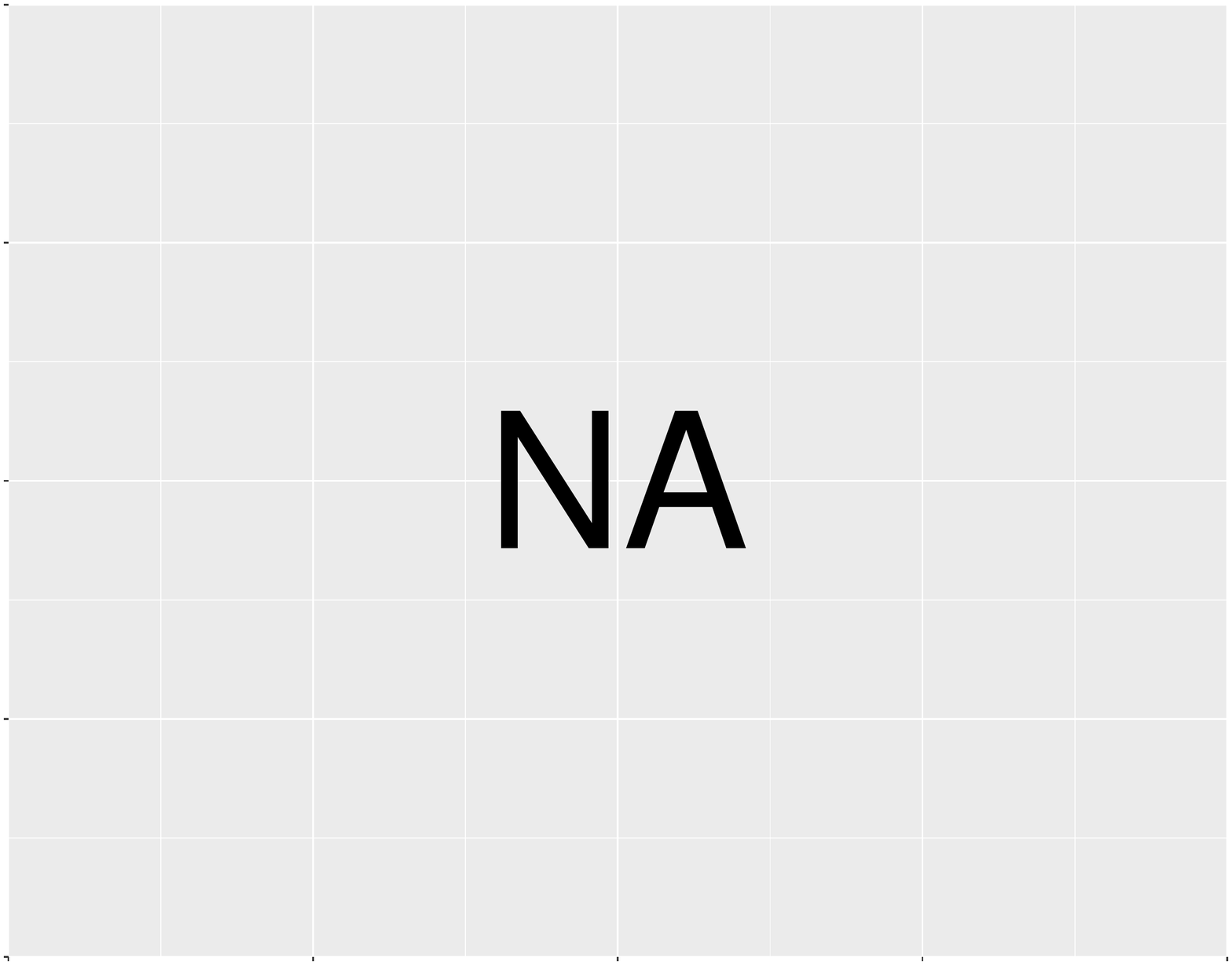}
    \caption{Not available for P-spline}\label{ggRealdataYPSpline}
    \end{subfigure}
    \begin{subfigure}{0.45\textwidth}
    \centering
    \includegraphics[width=\linewidth,height=0.5\textwidth]{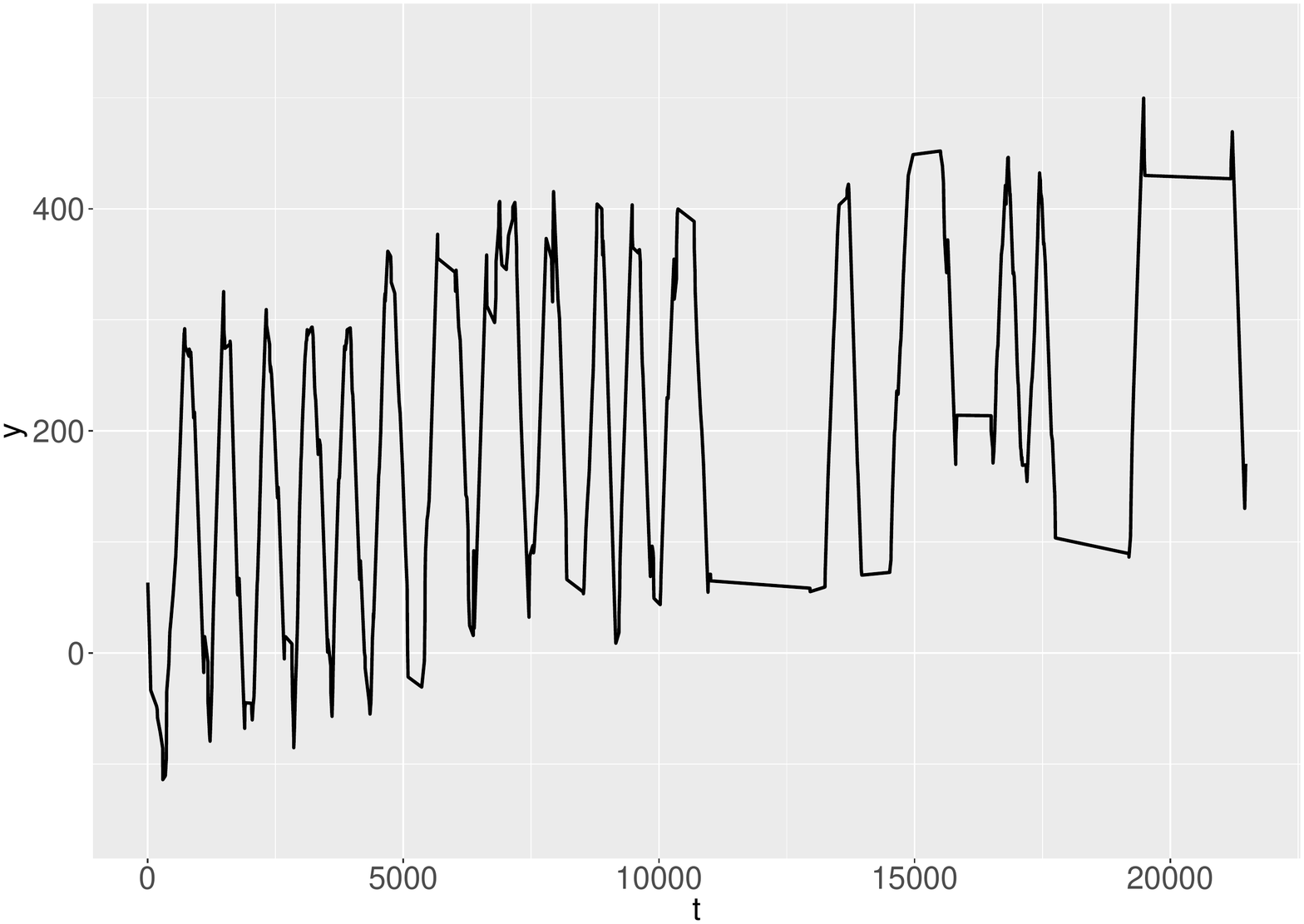}
    \caption{Reconstruction by wavelet ($\textit{sure}$)}\label{ggRealdataYSure}
    \end{subfigure}
    
    \begin{subfigure}{0.45\textwidth}
    \centering
    \includegraphics[width=\linewidth,height=0.5\textwidth]{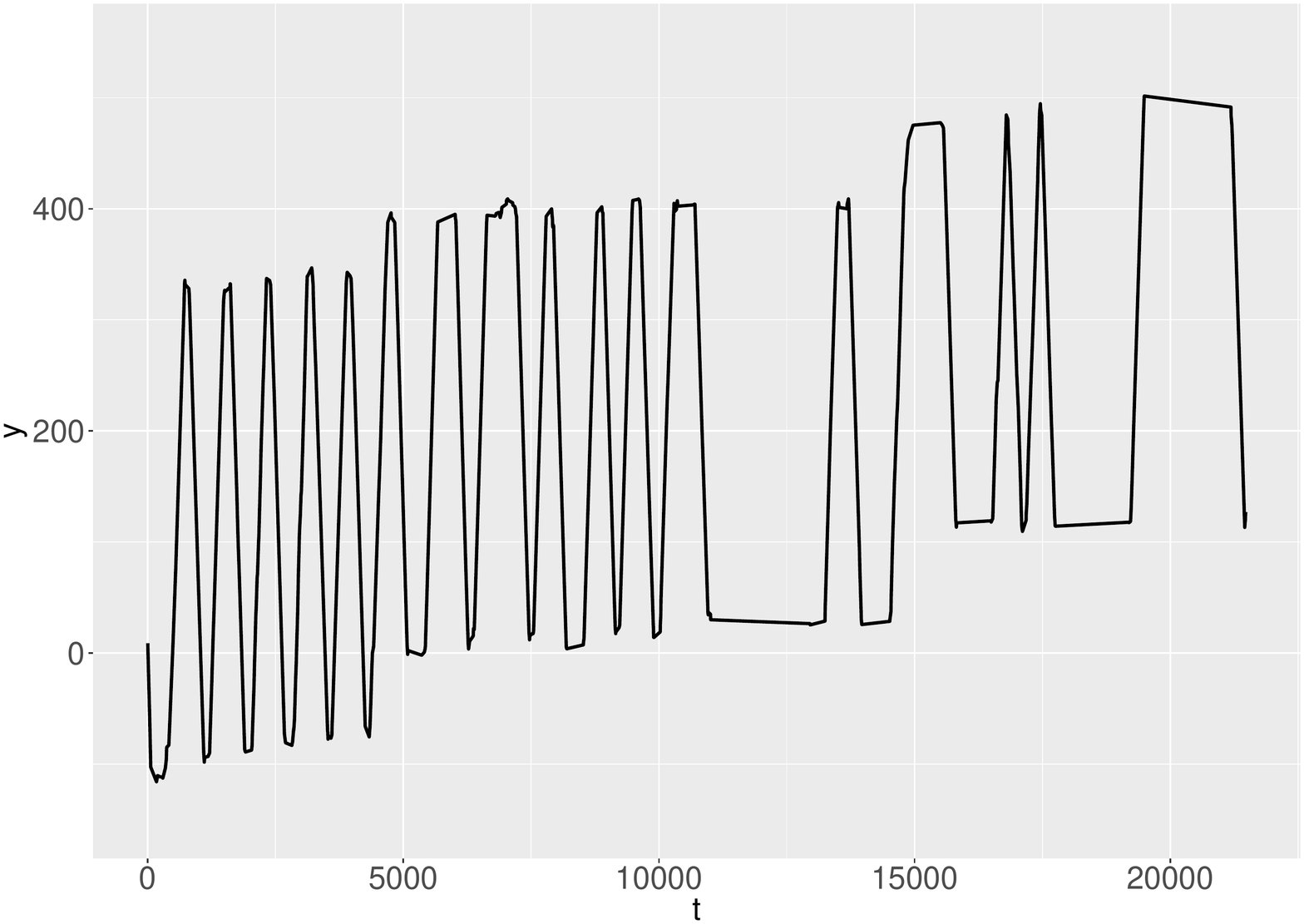}
    \caption{Reconstruction by wavelet (\textit{Bayes}) }\label{ggRealdataYBayes}
    \end{subfigure}
    \begin{subfigure}{0.45\textwidth}
    \centering
    \includegraphics[width=\linewidth,height=0.5\textwidth]{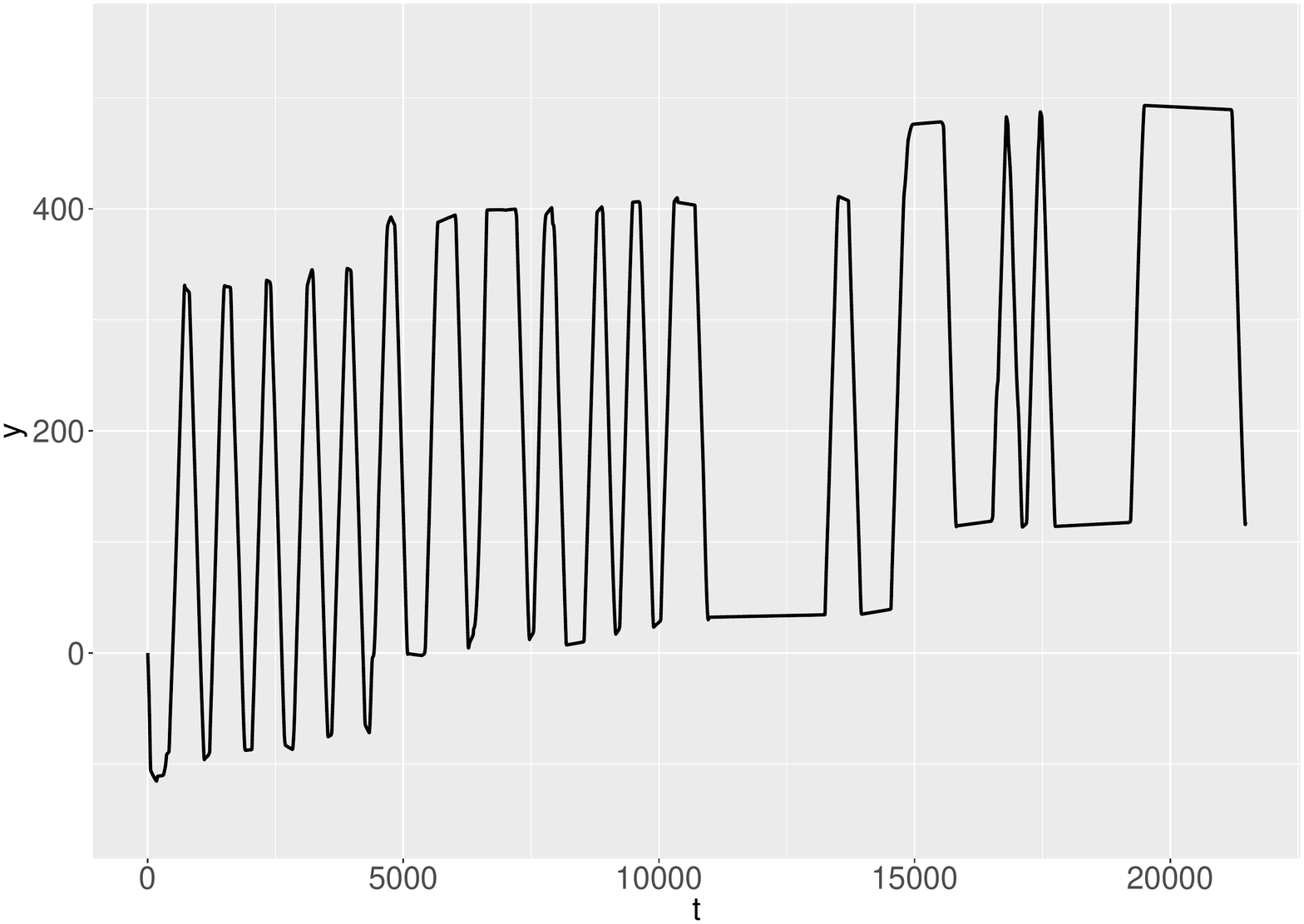}
    \caption{Reconstruction by adaptive V-spline, $\gamma=0$ }\label{ggRealdataYTractorGamma}
    \end{subfigure}
    
    \begin{subfigure}[t]{0.45\textwidth}
    \centering
    \includegraphics[width=\linewidth,height=0.5\textwidth]{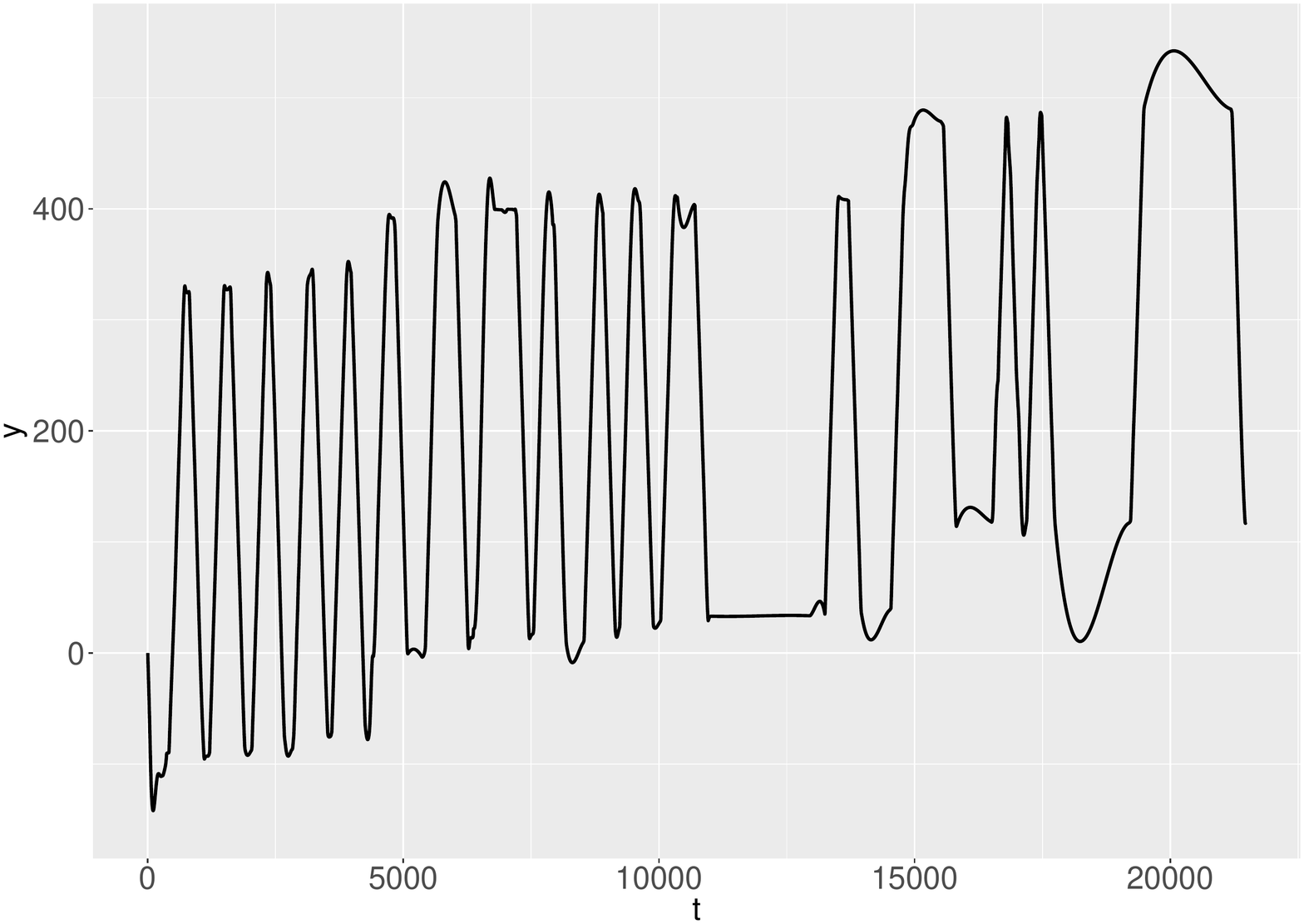}
    \caption{Reconstruction by non-adaptive V-spline}\label{ggRealdataYTractorAPT}
    \end{subfigure} 
    \begin{subfigure}[t]{0.45\textwidth}
    \centering
    \includegraphics[width=\linewidth,height=0.5\textwidth]{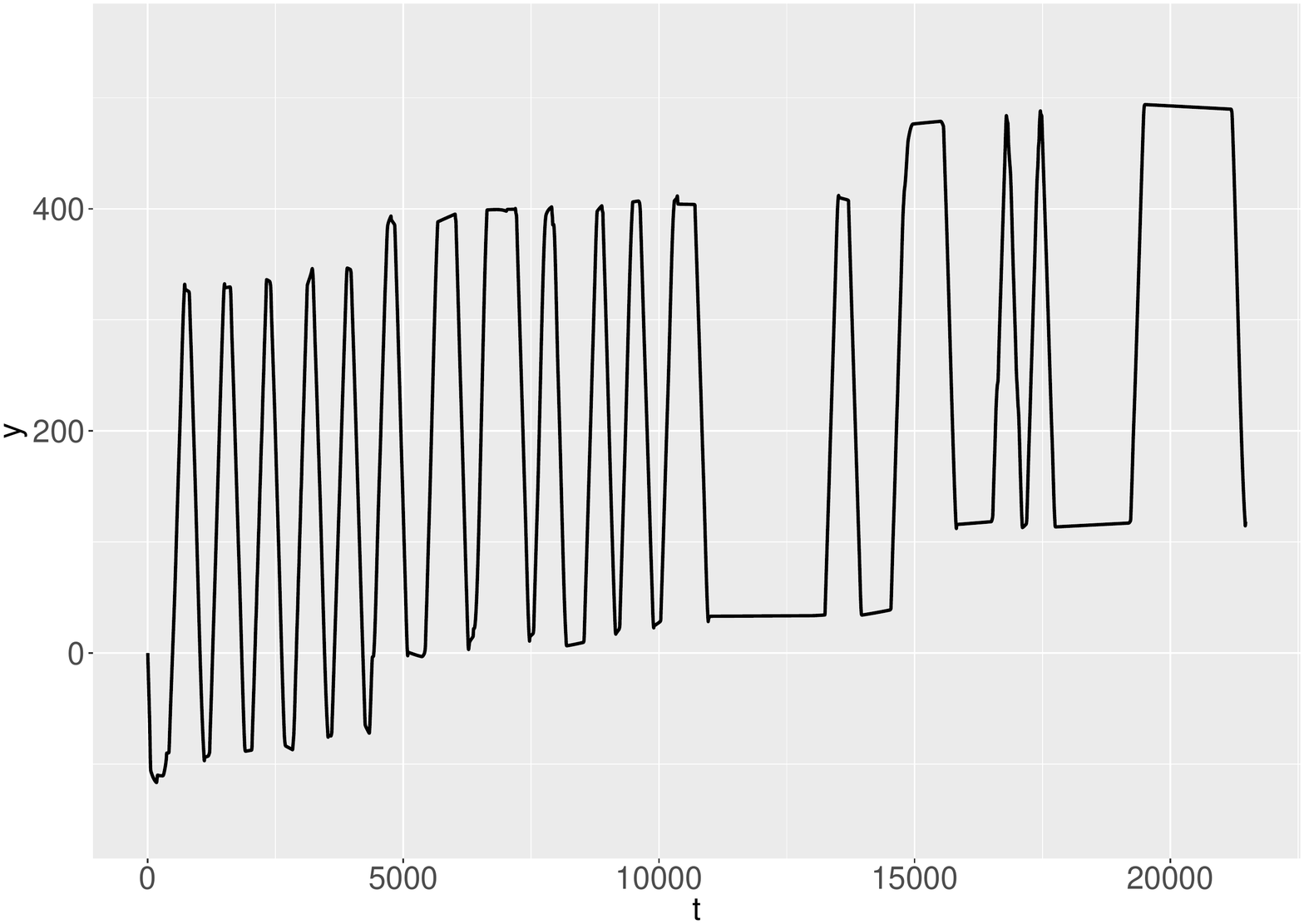}
    \caption{Reconstruction by adaptive V-spline\\ \mbox{  }}\label{ggRealdataYTractor}
    \end{subfigure}
    
\caption{Fitted northing trajectories. Figure \ref{ggRealdataYPSpline} The P-spline is fails due to instances where $y_{i+1}=y_i$, which lead to issues with invertibility. Figure \ref{ggRealdataYSure} Wavelet ($\textit{sure}$) exhibiting overshooting at turning points. Figure \ref{ggRealdataYBayes} Wavelet ($\textit{BayesThresh}$) showing improvement over wavelet ($\textit{sure}$). Figure \ref{ggRealdataYTractorGamma} The adaptive V-spline that does not incorporate velocity information. Figure \ref{ggRealdataYTractorAPT} Non-adaptive V-spline exhibiting smooth overshooting at turning points. Figure \ref{ggRealdataYTractor} The adaptive V-spline that includes velocity information.}\label{1dy}
 \end{figure}

\subsection{2-Dimensional Trajectory}

%In a 2-dimensional trajectory reconstruction, different from combined 1-dimensional reconstruction, we use the same parameters $\lambda_d$, $\lambda_u$ and $\gamma$ for both $x$ and $y$ axes. The overall best parameters return the least cross-validation score on all axes. Explicitly, it is calculated by the following formula 
%\begin{equation}
%\mbox{CV}=\mbox{CV}_x+\mbox{CV}_y.
%\end{equation}
%In the adjusted penalty term, $\Delta d$ is the Euclidean distance $\Delta_d(p_1,p_2)=\sqrt{(\Delta x)^2+(\Delta y)^2}$ between two positions on the 2D surface. Similar to 1-dimensional reconstruction, the velocity information keeps trajectory in the right direction and the penalty term makes sure that the crazy curve will disappear between long-time-gap points. 
Figure \ref{completecombind2dxy} shows the full 2D adaptive V-spline reconstruction of the tractor trajectory for the first 512 observations, as well as the derived northing and easting trajectories. Also depicted by red dots of varying sizes are the sizes of the penalty terms $\lambda_i$. As expected, the penalty terms typically become large at turning points and before long waiting periods. A histogram of the $\lambda_i$ and the implied penalty function $\lambda(t)$ are given in Figure \ref{2dpenalty}. Figure \ref{complete2DXY} is a full reconstruction from the complete dataset.% The overall reconstruction gives a smoothing path that goes through each measurement and avoids curvatures at turning points. 

\begin{figure}
  \centering
 \begin{subfigure}{\textwidth}
     \centering
     \includegraphics[width=0.9\linewidth]{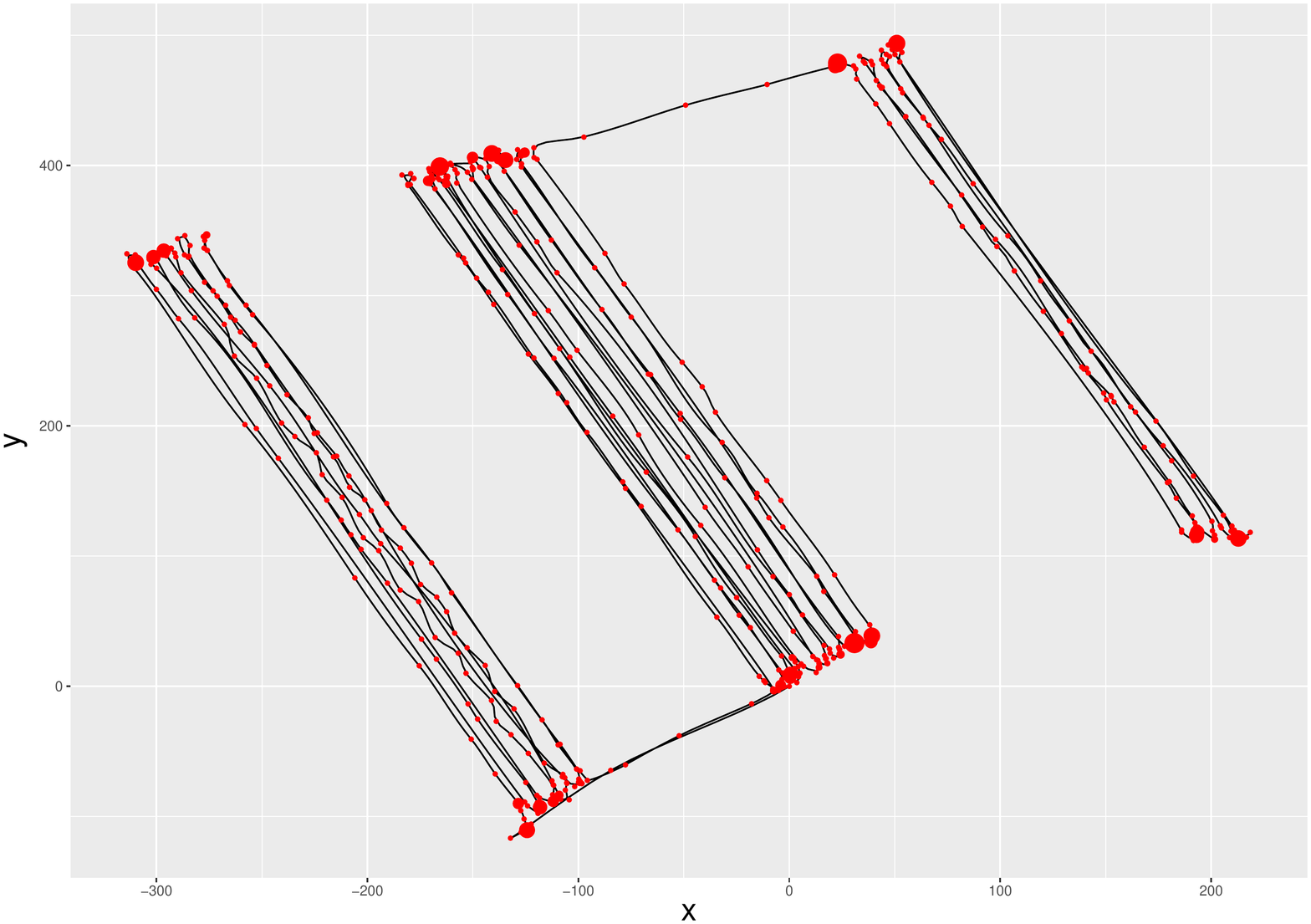}
     \caption{Reconstruction of northing $y$ vs easting $x$}
     \end{subfigure}
      \begin{subfigure}{\textwidth}
     \centering
     \includegraphics[width=0.45\linewidth]{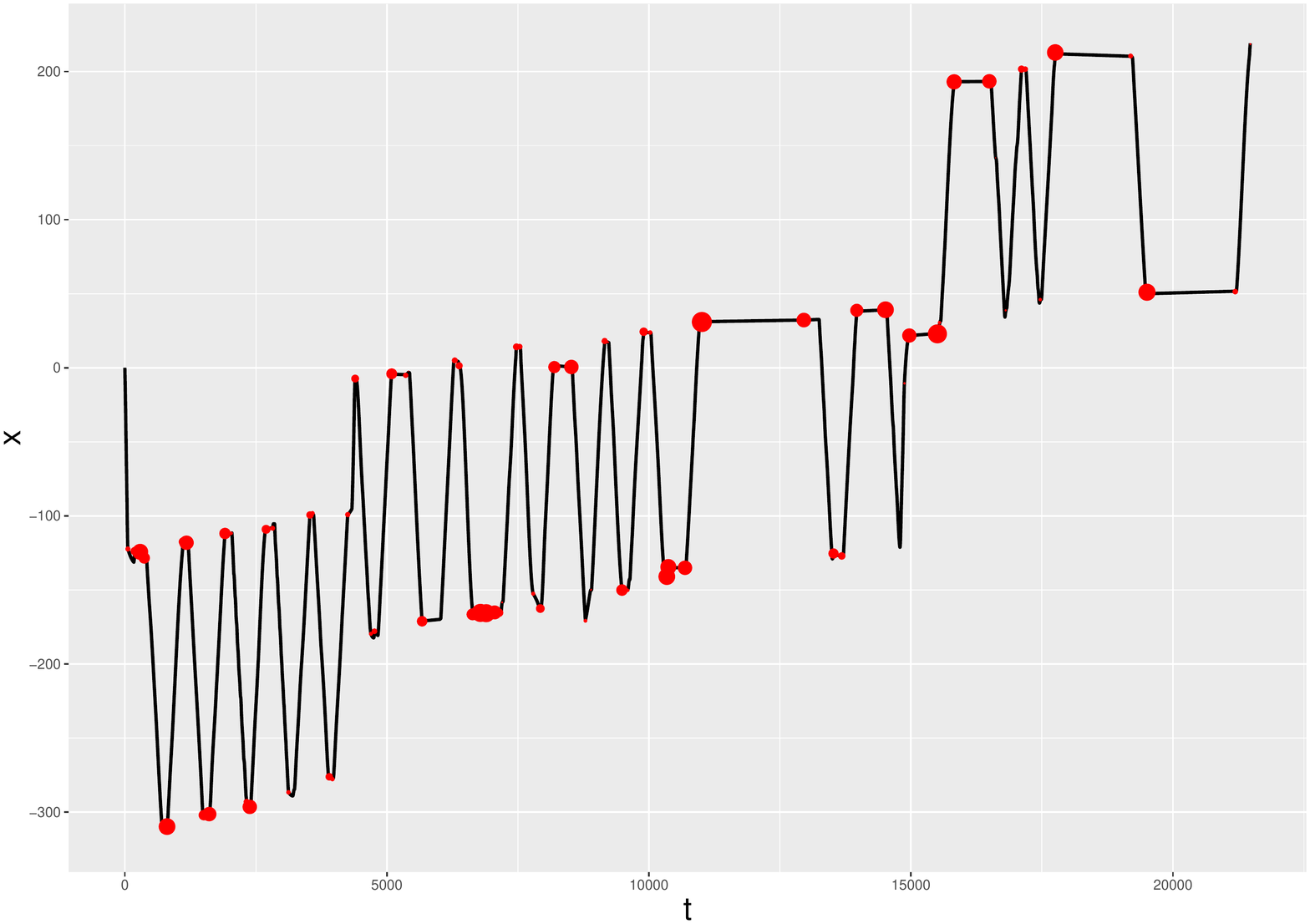}
     \includegraphics[width=0.45\linewidth]{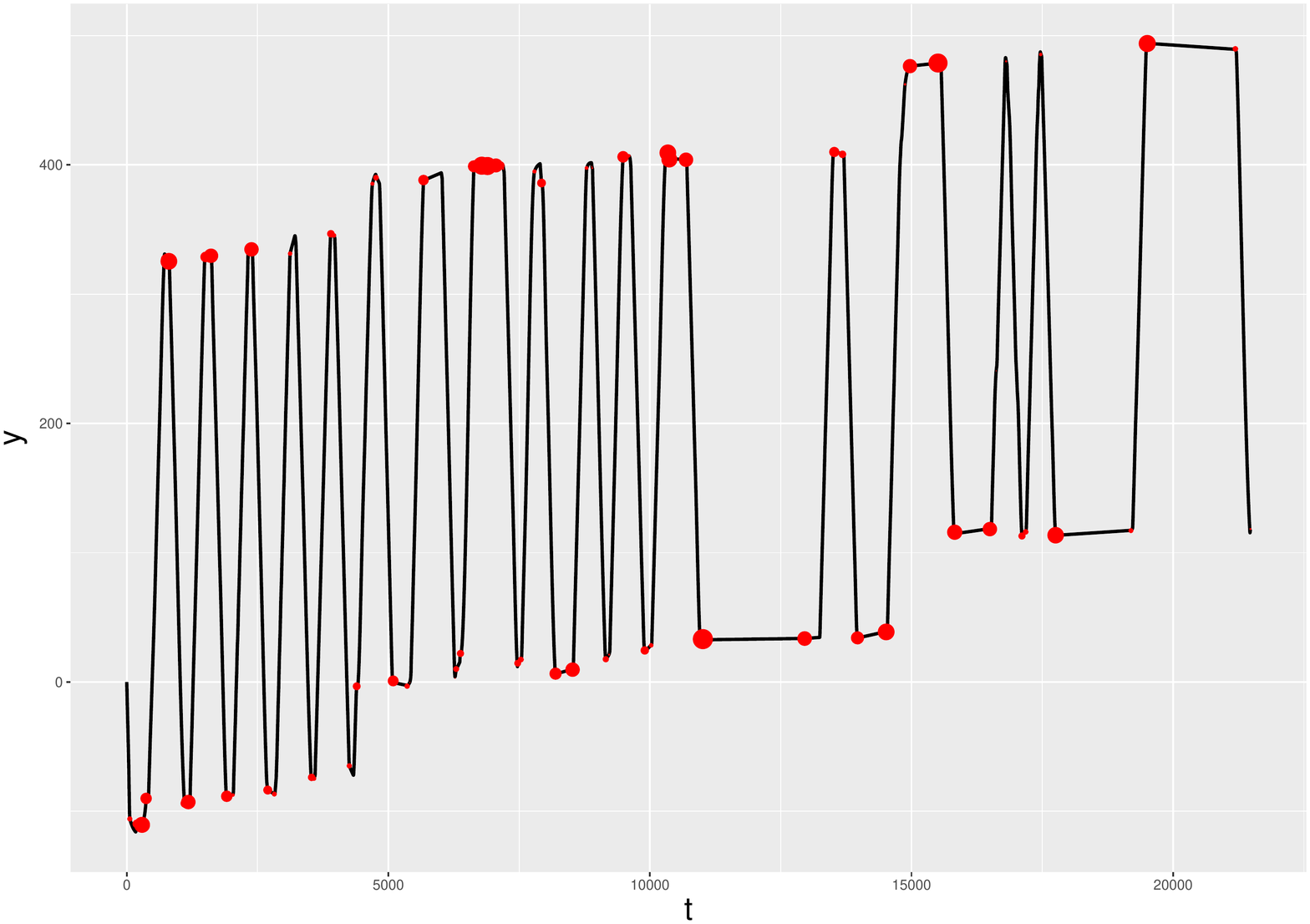}
     \caption{Derived easting $x$ and northing $y$ trajectories vs time $t$}
     \end{subfigure}
     \caption{Full 2-dimensional reconstruction; larger red dots indicate larger values of $\lambda_i$.}
     \label{completecombind2dxy}
\end{figure}

%The penalty function $\lambda(t)$ of a 2-dimensional reconstruction is shared by $x$ and $y$ axes and presented in figure \ref{2dpenalty}. The complete penalty term is 
%\begin{equation*}
%n\theta_x^\top\Omega_{\eta_d,\eta_u}\theta_x + n\theta_y^\top\Omega_{\eta_d,\eta_u}\theta_y.
%\end{equation*}
%Similarly, most of the large penalty values appear at long-time-gap knots and turning points. A histogram plot of penalty function shows that most of the values are small and only a couple of them are large. 
\begin{figure}
  \centering
   \includegraphics[width=0.45\textwidth]{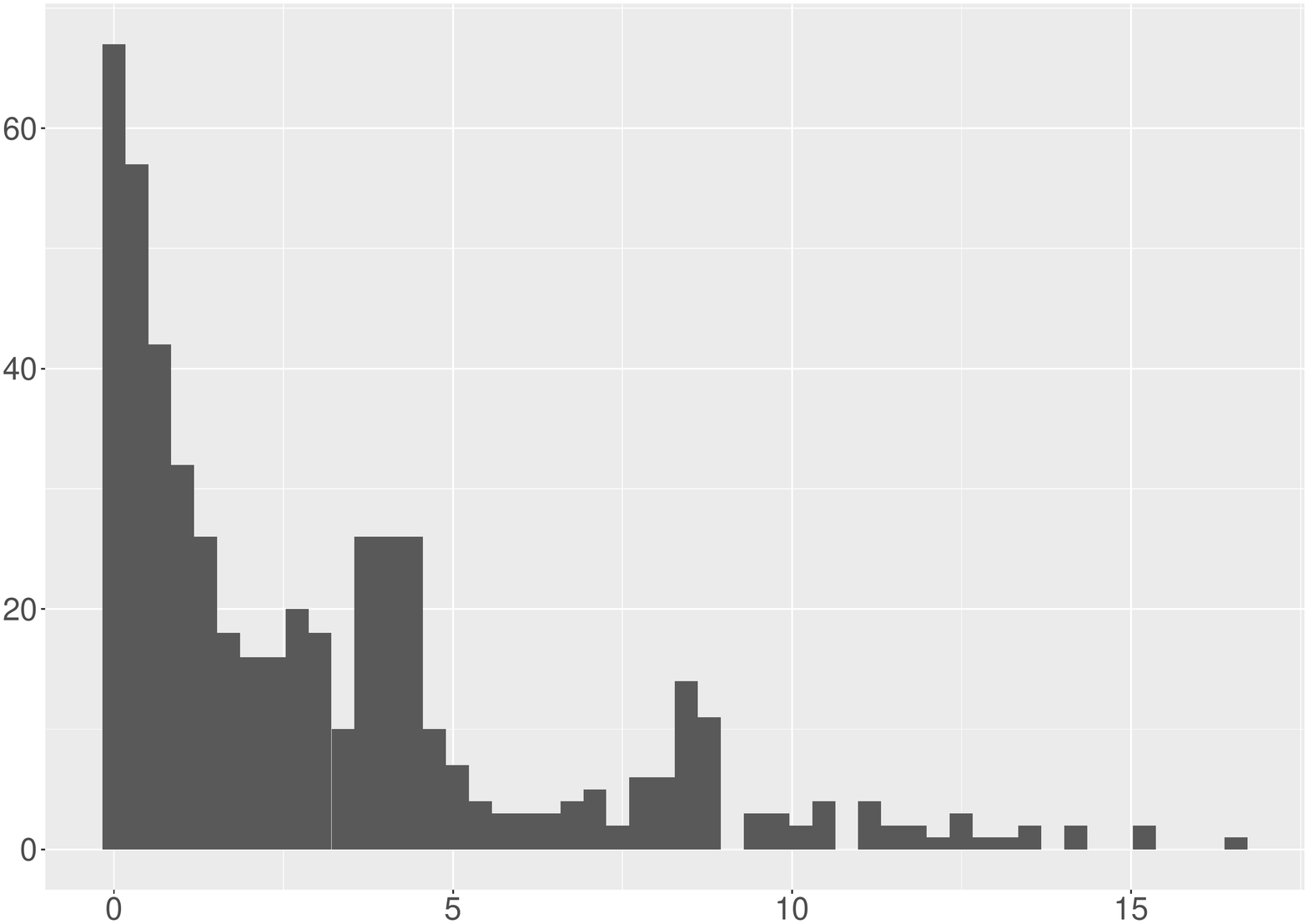} 
    \includegraphics[width=0.45\textwidth]{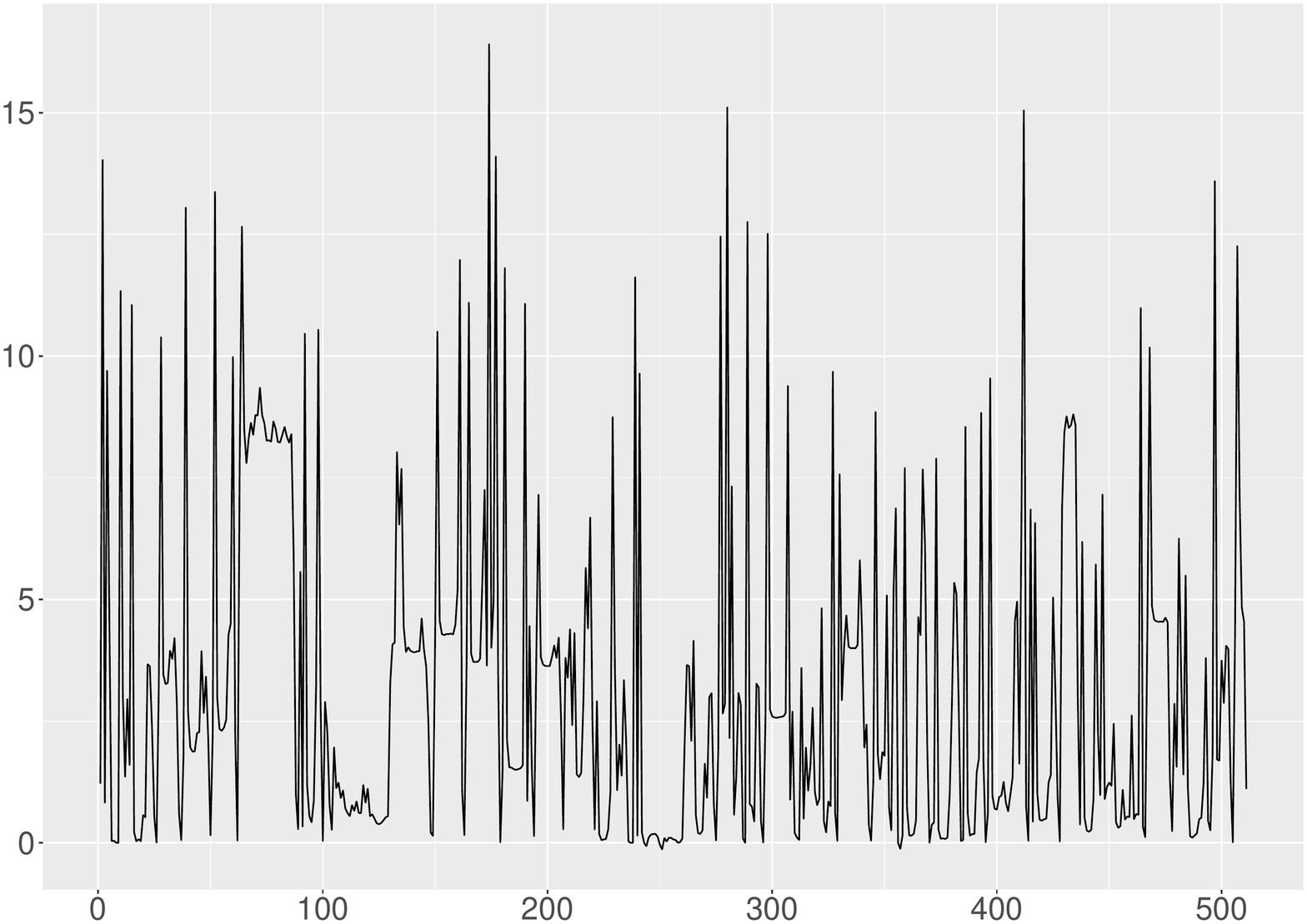}
  \caption{Histogram of the $\lambda_i$ and implied penalty function $\lambda(t)$}\label{2dpenalty}
\end{figure}

%Instead of reconstructing on $x$ and $y$ axes separately, it chooses the penalty value with respect to the balance on both of the two directions. 
\begin{figure}
\centering
\includegraphics[width=0.9\linewidth]{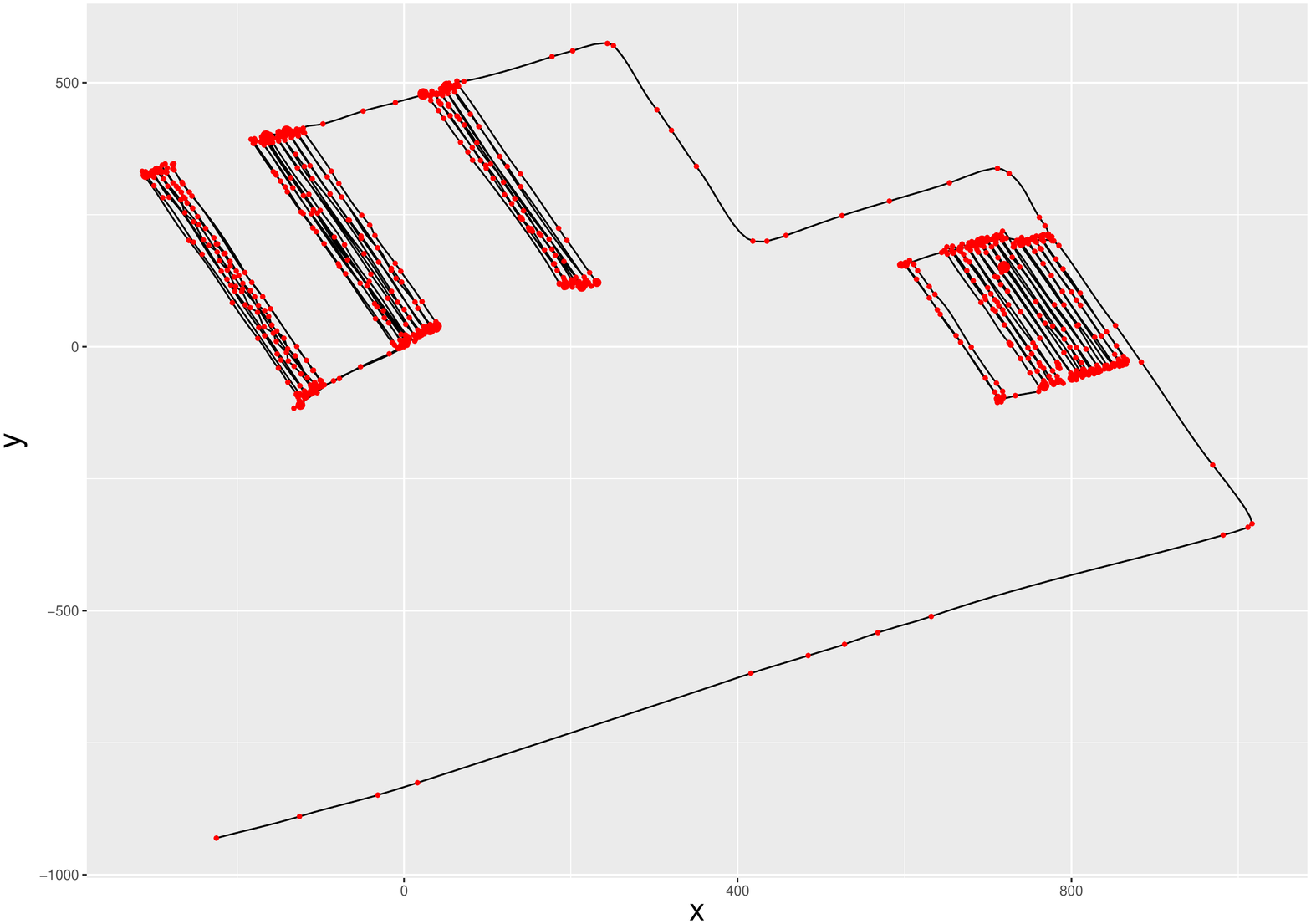}
\caption{Full 2-dimensional reconstruction for the complete dataset; larger red dots indicate larger values of $\lambda_i$.}\label{complete2DXY}
\end{figure}

\section{Discussion}

In this paper, a smoothing spline called the V-spline is proposed that minimizes an objective function which incorporates both position and velocity information. Given $n$ knots, the V-spline has $2n$ degrees of freedom corresponding to $n-1$ cubic polynomials with their value and first derivative matched at the $n-2$ interior knots. The degrees of freedom are then fixed by $n$ position observations and $n$ velocity observations. Note that in the limit $\gamma\rightarrow 0$, the V-spline reduces to having $n$ degrees of freedom. (See Appendix \ref{proofofCorollary} for details.) An adaptive version of the V-spline is also introduced that seeks to control the impact of irregularly sampled observations and noisy velocity measurements. In simulation studies, the V-spline shows improved true mean squared error in reconstructions.

%The adaptive penalty function adapts to complicated curvatures. In a $d$-dimensional space, V-spline can be projected into sub-spaces with respect to $t$ and combined each solution together as a final. This method performs better when we know the position and velocity information than other methods. 

%Additionally, the reconstruction of a V-spline contains $4\times (n-1)$ parameters if we have $n$ knots. By adding $2\times (n-2)$ constraints, the original function, and its first derivative are continuous at each interior knots, the degrees of freedom will be $4\times (n-1)-2\times (n-2)=2n$. Because there are $n$ position and $n$ velocity points, we do not need to specify more parameters or add more constraints to the model. 

%In our experimental studies, the MSE of the V-spline were neither significantly better or worse than other methods. However, its true MSE was significantly less. 

In most work on polynomial smoothing splines, observation errors are assumed to be independent. However, this can be a questionable assumption in practice and it is known that correlation greatly affects the selection of smoothing parameters \citep{wang1998smoothing}. Parameter selection methods such as generalized maximum likelihood (GML), generalized cross-validation (GCV) typically underestimate smoothing parameters when data are correlated. To accommodate an autocorrelated error sequence, \cite{diggle1989spline} extend GCV, and \cite{wang1998smoothing} extends GML and unbiased risk (UBR) to simultaneously estimate the smoothing and correlation parameters. \cite{kohn1992nonparametric} also propose an algorithm to evaluate cross-validation functions when the autocorrelated errors are modelled by an autoregressive moving average. Here we show that there is a generalized cross-validation scheme for the V-spline that is appropriate for correlated observation errors. 

%Consider the model $y=f(t)+\varepsilon_1$ and $v=f'(t)+\varepsilon_2$, where $\varepsilon_1\sim N\left(0,\sigma^2W_1^{-1}\right)$, $\varepsilon_2\sim N\left(0,\frac{\sigma^2}{\gamma}W_2^{-1}\right)$ with variance parameters $\sigma^2$ and $\gamma$, and two covariance matrices $W_1^{-1}$ and $W_2^{-1}$. 
When observation errors are correlated (and considering splines in one-dimension for simplicity), the V-spline minimizes
\begin{equation}
\frac{1}{n}\left(\mathbf{y}-\mathbf{f}\right)^\top W_1\left(\mathbf{y}-\mathbf{f}\right)+\frac{\gamma}{n}\left(\mathbf{v}-\mathbf{f}'\right)^\top W_2\left(\mathbf{v}-\mathbf{f}'\right)+\sum_{i=1}^{n-1}\lambda_i\int_{t_i}^{t_{i+1}}(f''(t))^2dt,
\end{equation}
where $W_1$ and $W_2$ are precision matrices, assumed known. The solution is $\hat{f}(t)=\sum_{k=1}^{2n}N_k(t)\hat{\theta}_k$, where
\begin{equation}
\hat{\theta}=\left(B^\top W_1 B+ \gamma C^\top W_2C+n\Omega_\lambda\right)^{-1}\left(B^\top W_1 \mathbf{y}+\gamma C^\top W_2\mathbf{v}\right).
\end{equation}

% the GCV of a V-spline can be written similarly as equation \eqref{tractorcv} but using the average trace of the matrices.

%in the following form \footnotesize
%\begin{equation*}
%\mbox{GCV}(\lambda,\gamma)=\frac{\left(\mathbf{\hat{f}}-\mathbf{y}\right)^\top \left(\mathbf{\hat{f}}-\mathbf{y}\right) + \frac{2\tr\left(\gamma T\right)}{\tr\left(I-\gamma V\right)}\left(\mathbf{\hat{f}}-\mathbf{y}\right)^\top \left(\mathbf{\hat{f}}'-\mathbf{v}\right) + \left( \frac{\tr(\gamma T)}{\tr(I-\gamma V)} \right)^2 \left(\mathbf{\hat{f}}'-\mathbf{v}\right)^\top \left(\mathbf{\hat{f}}'-\mathbf{v}\right)}{\left( \tr(I-S-\frac{\tr(\gamma T)}{\tr(I-\gamma V)}U) \right)^2}.
%\end{equation*} \normalsize
The usual leave-one-out cross-validation algorithm requires knowledge of the diagonal elements of the smoother matrices. GCV achieves computational savings and robustness by approximating $S_{ii}\approx\frac{1}{n}\tr(S)$, etc. \citep{syed2011review}. The natural extension of the GCV to the V-spline with correlated errors is \scriptsize
\begin{equation}
\mbox{GCV}\left(\lambda,\gamma\right)=\frac{\left(\mathbf{\hat{f}}-\mathbf{y}\right)^\top W_1\left(\mathbf{\hat{f}}-\mathbf{y}\right)+\frac{2\tr\left(\gamma T\right)}{\tr\left(I-\gamma V\right)}\left(\mathbf{\hat{f}}-\mathbf{y}\right)^\top W_1^{1/2}(W_2^{1/2})^{\top}\left(\mathbf{\hat{f}}'-\mathbf{v}\right) + \left( \frac{\tr\left(\gamma T\right)}{\tr\left(I-\gamma V\right)} \right)^2 \left(\mathbf{\hat{f}}'-\mathbf{v}\right)^\top W_2\left(\mathbf{\hat{f}}'-\mathbf{v}\right)}{\left( \tr\left(I-S-\frac{\tr\left(\gamma T\right)}{\tr\left(I-\gamma V\right)}U\right) \right)^2}.
\end{equation}
\normalsize

One current drawback of the V-spline is that it is slower than other spline-based methods, however we have not yet optimized CV parameter estimation, particularly around solving for $\hat{\theta}$ in \eqref{thetahat}. Future directions for the V-spline include application to ship tracking \citep{hintzen2010improved} and development of a fast on-line filtering mode.

\section*{Acknowledgements}

This work was funded by grant UOOX1208 from the Ministry of Business, Innovation \& Employment (NZ). GPS data was provided by TracMap (NZ).

%\numberwithin{equation}{section}
%\numberwithin{lemma}{section}

\appendix

\section{Penalty Matrix in \eqref{tractormse}}\label{PenaltyTermDetails}

We have $\Omega_{\lambda}=\sum_{i=1}^{n-1}\lambda_i\Omega^{(i)}$, where $[\Omega^{(i)}]_{jk}=\int_{t_i}^{t_{i+1}} N''_j(t)N''_k(t)dt$. Thus $\Omega^{(i)}$ is a $2n \times 2n$ bandwidth four symmetric matrix and its non-zero, upper triangular elements are 
\begin{align}
\Omega_{2i-1,2i-1}^{(i)} & =\int_{t_{i}}^{t_{i+1}} \frac{d^2 h_{00}^{(i)}(t)}{dt^2}  \frac{d^2 h_{00}^{(i)}(t)}{dt^2} dt=\frac{12}{\Delta T_i^3}\\
\Omega_{2i-1,2i}^{(i)} &=\int_{t_{i}}^{t_{i+1}} \frac{d^2 h_{00}^{(i)}(t)}{dt^2}  \frac{d^2 h_{10}^{(i)}(t)}{dt^2} dt=\frac{6}{\Delta T_i^2}\\
\Omega_{2i-1,2i+1}^{(i)} &=\int_{t_{i}}^{t_{i+1}} \frac{d^2 h_{00}^{(i)}(t)}{dt^2}  \frac{d^2 h_{01}^{(i)}(t)}{dt^2} dt=\frac{-12}{\Delta T_i^3}\\
\Omega_{2i-1,2i+2}^{(i)} &=\int_{t_{i}}^{t_{i+1}} \frac{d^2 h_{00}^{(i)}(t)}{dt^2}  \frac{d^2 h_{11}^{(i)}(t)}{dt^2} dt=\frac{6}{\Delta T_i^2}\\
\Omega_{2i,2i}^{(i)} &=\int_{t_{i}}^{t_{i+1}} \frac{d^2 h_{10}^{(i)}(t)}{dt^2}  \frac{d^2 h_{10}^{(i)}(t)}{dt^2} dt=\frac{4}{\Delta T_i} \\
\Omega_{2i,2i+1}^{(i)} &=\int_{t_{i}}^{t_{i+1}} \frac{d^2 h_{10}^{(i)}(t)}{dt^2}  \frac{d^2 h_{01}^{(i)}(t)}{dt^2} dt=\frac{-6}{\Delta T_i^2}\\
\Omega_{2i,2i+2}^{(i)} &=\int_{t_{i}}^{t_{i+1}} \frac{d^2 h_{10}^{(i)}(t)}{dt^2}  \frac{d^2 h_{11}^{(i)}(t)}{dt^2} dt=\frac{2}{\Delta T_i}\\
\Omega_{2i+1,2i+1}^{(i)} &=\int_{t_{i}}^{t_{i+1}} \frac{d^2 h_{01}^{(i)}(t)}{dt^2}  \frac{d^2 h_{01}^{(i)}(t)}{dt^2} dt=\frac{12}{\Delta T_i^3}\\
\Omega_{2i+1,2i+2}^{(i)} &=\int_{t_{i}}^{t_{i+1}} \frac{d^2 h_{01}^{(i)}(t)}{dt^2}  \frac{d^2 h_{11}^{(i)}(t)}{dt^2} dt=\frac{-6}{\Delta T_i^2}\\
\Omega_{2i+2,2i+2}^{(i)} &=\int_{t_{i}}^{t_{i+1}} \frac{d^2 h_{11}^{(i)}(t)}{dt^2}  \frac{d^2 h_{11}^{(i)}(t)}{dt^2} dt=\frac{4}{\Delta T_i}
\end{align}
where $\Delta T_i=t_{i+1}-t_i$ and $i=1,2,\ldots,n-1$.

\section{Proof of Theorem \ref{TractorSplineTheorem}}\label{AppendixTractorSplineProof}

Our proof is an extension of the smoothing spline proof in  \cite{green1993nonparametric}.
\begin{proof}
If $g:[a,b]\rightarrow \mathbb{R}$ is a proposed minimizer, construct a cubic spline $f(t)$ that agrees with $g(t)$ and its first derivatives at $t_1,\ldots,t_n$, and is linear on $[a, t_1]$ and $[t_n, b]$. Let $h(t) = g(t)-f(t)$. Then, for $i = 1,\dots ,n-1$, 
\begin{align*}
\int_{t_i}^{t_{i+1}}f''(t)h''(t)dt &=f''(t)h'(t) \bigg\rvert_{t_i}^{t_{i+1}}-\int_{t_i}^{t_{i+1}}f'''(t)h'(t)dt \\
&= 0-f'''\left(t_i^+\right)\int_{t_i}^{t_{i+1}}h'(t)dt \\
&= -f'''\left(t_i^+\right)\left( h(t_{i+1}) -h(t_i) \right)=0.\\
\end{align*}
Additionally, $\int_{a}^{t_1}f''(t)h''(t)dt=\int_{t_n}^{b}f''(t)h''(t)dt=0$, since $f(t)$ is assumed linear outside the knots. Thus, for $i=0,\ldots,n$, 
\begin{align*}
\int_{t_i}^{t_{i+1}}\lvert g''(t) \rvert^2dt &= \int_{t_i}^{t_{i+1}}\lvert f''(t)+h''(t)\rvert^2 dt\\
&= \int_{t_i}^{t_{i+1}}\lvert f''(t)\rvert^2dt+2\int_{t_i}^{t_{i+1}}f''(t)h''(t)dt+\int_{t_i}^{t_{i+1}}\lvert h''(t)\rvert^2dt\\
&=\int_{t_i}^{t_{i+1}}\lvert f''(t)\rvert^2dt+\int_{t_i}^{t_{i+1}}\lvert h''(t)\rvert^2dt\\
&\geq \int_{t_i}^{t_{i+1}}\lvert f''(t)\rvert^2dt.
\end{align*}
The result $J[f]\leq J[g]$ follows since $\lambda_i>0$. 

Furthermore, equality of the curvature penalty term only holds if $g(t) = f(t)$. On $[t_1, t_n]$, we require $h''(t) = 0$ but since $h(t_i) = h'(t_i) = 0$ for $i = 1, \ldots , n$, this means $h(t) = 0$. Meanwhile on $[a, t_1]$ and $[t_n, b]$, $f''(t) = 0$ so that equality requires $g''(t)=0$. Since $f(t)$ agrees with $g(t)$ and its first derivatives at $t_1$ and $t_n$, equality is forced on both intervals.
\end{proof}

\section{Proof of Theorem \ref{TractorsplineCorollary}}\label{proofofCorollary}

\begin{proof}

Let $Q_{\lambda}= C\Omega_{\lambda}= C\sum_{i=1}^{n-1}\lambda_i\Omega^{(i)}$, which is simply the even rows of $\Omega_{\lambda}$. Multiplying both sides of \eqref{thetahat} by $CG$, where $G=B^\top B+\gamma C^\top C+n\Omega_{\lambda}$, and using the fact that $C B^\top=0$ and $C C^\top=I$, we find
\[
Q_{\lambda}\hat{\theta} =\frac{\gamma}{n}(\mathbf{v}-\mathbf{f}').
\]

We note in passing that when $\gamma=0$, $Q_{\lambda}\hat{\theta}=0$. Since $Q_\lambda$ is an $n\times 2n$ matrix of full rank (this can be established from results below), the degrees of freedom of the V-spline decrease from $2n$ to $n$ when the velocity information is removed, as one would expect.

%By setting $\gamma\to 0$, the velocity information $v$ is essentially removed. The degrees of freedom of parameters decreases from $2n$ to $n$. Hence, there exists an $n\times 2n$ matrix $Q_\lambda$ restricting $n$ degrees of freedom of $\hat{\theta}$ and satisfying $Q_\lambda\hat{\theta}=0$.
%
%The matrices $B$ and $C$ have the following property:  
%\begin{align*}
%& BB^\top=CC^\top=I_n, \\
%& C^\top CB^\top=B^\top BC^\top=0.
%\end{align*}
%Denoting $G=B^\top B+ \gamma C^\top C +n\Omega_\lambda$ and giving $\hat{\theta} =(B^\top B+ \gamma C^\top C +n\Omega_\lambda)^{-1}(B^\top y+\gamma C^\top v)$, we will have $G\hat{\theta}=B^\top y+\gamma C^\top v$ and 
%\begin{align*}
%& BG\hat{\theta}=y+\gamma BC^\top v\\
%& CG\hat{\theta}=CB^\top y+\gamma v.
%\end{align*}
%Further, $C^\top CG\hat{\theta}=C^\top (CB^\top y+\gamma v) = \gamma C^\top v$. If by setting $\gamma=0$, one will get $Q_\lambda = C^\top CG$, which consists of the even rows of $\Omega_\lambda$.

%For $1<i<n$, 
%\begin{align*}
%&\lambda_i\frac{6}{\Delta_i^2}\theta_{2i-3}+\lambda_i\frac{2}{\Delta_i}\theta_{2i-2}+\left( \lambda_i\frac{-6}{\Delta_i^2}+ \lambda_{i+1}\frac{6}{\Delta_{i+1}^2} \right) \theta_{2i-1} \\
%+ & \left( \lambda_i\frac{4}{\Delta_i} + \lambda_{i+1}\frac{4}{\Delta_{i+1}}\right)\theta_{2i} 
%+ \lambda_{i+1}\frac{-16}{\Delta_{i+1}^2}\theta_{2i+1} +  \lambda_{i+1}\frac{2}{\Delta_{i+1}}\theta_{2i+2}
%\end{align*}

The only basis functions with support on $[t_i,t_{i+1})$ are the $N_{k}(t)$ with $k\in\{2i-1,2i,2i+1,2i+2\}$. Integrating by parts and using properties of the basis functions at the knots, we obtain the non-zero elements of the even rows of $\Omega^{(i)}$:
\begin{equation*}
\Omega^{(i)}_{2i,k}= - N''_{k}\left(t_i^+\right),  \qquad
\Omega^{(i)}_{2i+2,k}=N''_{k}\left(t_{i+1}^-\right),
\end{equation*}
where $k\in\{2i-1,2i,2i+1,2i+2\}$. The actual numerical values are given in Appendix \ref{PenaltyTermDetails} and confirm that $Q_\lambda$ has full rank. It follows that
\begin{equation*}
(Q_\lambda \hat{\theta})_{i} = \left(\lambda_{i-1} \Omega^{(i-1)}_{2i,\cdot}+\lambda_i \Omega^{(i)}_{2i,\cdot}\right) \hat{\theta}=\lambda_{i-1}f''(t_i^-)-\lambda_{i}f''(t_i^+),
\end{equation*}
where it is understood that $f''(t_1^-)=f''(t_n^+)=0$. Thus the backward implication is clear.

Continuity of $f''(t)$ at $t=t_1$ implies $\gamma(v_1 -f'_1)=0$. If we can show $\gamma=0$, then continuity of $f''(t)$ at the remaining observation times will establish the forward implication. Suppose instead $v_1=f'_1$. Then from \eqref{thetahat}, we have $v_1 = (0\,1\,0\cdots 0)G^{-1}(B^\top \mathbf{y}+\gamma C^\top \mathbf{v})$. For almost all $\mathbf{y}$ and $\mathbf{v}$, this implies $(0\,1\,0\cdots 0)^\top$ is an eigenvector of $G$ with eigenvalue $\gamma$. But, using the numerical values from Appendix \ref{PenaltyTermDetails}, the second column of $G$ is $(\frac{6\lambda_1}{\Delta T_1^2}\ \gamma\!+\!\frac{4\lambda_1}{\Delta T_1}\ \frac{-6\lambda_1}{\Delta T_1^2}\ \frac{2\lambda_1}{\Delta T_1}\ 0 \cdots 0)^\top\neq \gamma (0\,1\,0\cdots 0)^\top$, since $\lambda_1>0$.

%By integrating by parts and using properties of the basis functions at the knots, one can get the even rows of $\Omega^{(i)}$, which are 
%\begin{align*}
%\Omega^{(i)}_{2i,2i-1}&=N''_{2i-1}\left(t_i^-\right) - N''_{2i-1}\left(t_i^+\right)  \\
%\Omega^{(i)}_{2i,2i}   &=N''_{2i}\left(t_i^-\right) - N''_{2i}\left(t_i^+\right)  \\
%\Omega^{(i)}_{2i,2i+1}&=N''_{2i+1}\left(t_i^-\right) - N''_{2i+1}\left(t_i^+\right) \\
%\Omega^{(i)}_{2i,2i+2}&=N''_{2i+2}\left(t_i^-\right) - N''_{2i+2}\left(t_i^+\right) \\
%\Omega^{(i)}_{2i+2,2i-1}&=N''_{2i-1}\left(t_{i+1}^-\right) - N''_{2i-1}\left( t_{i+1}^+ \right) \\
%\Omega^{(i)}_{2i+2,2i}&=N''_{2i}\left(t_{i+1}^-\right) - N''_{2i}\left(t_{i+1}^+\right) \\
%\Omega^{(i)}_{2i+2,2i+1}&=N''_{2i+1}\left(t_{i+1}^-\right) - N''_{2i+1}\left(t_{i+1}^+\right) \\
%\Omega^{(i)}_{2i+2,2i+2}&=N''_{2i+2}\left(t_{i+1}^-\right) - N''_{2i+2}\left(t_i^+\right) 
%\end{align*}
%Thus 
%\begin{equation*}
%Q_\lambda = nC^\top C\Omega_\lambda = nC^\top C\sum_i\lambda_i \Omega^{(i)}.
%\end{equation*}
%Consequently, if and only if $\lambda$ is constant, $Q_\lambda\hat{\theta}=-\lambda\left( f''\left(t_i^+\right)-f''\left(t_i^-\right) \right)=0$, for $i=1,\ldots,n$, otherwise $Q_\lambda\theta=0$ is true but does not represent $f''\left(t_i^+\right)-f''\left(t_i^-\right)$. 
%
%As a result, $f''(t)$ is continuous at the knots $t_i$ if $\lambda(t)$ is constant and $\gamma=0$. 

\end{proof}

\section{Proof of Theorem \ref{tractorsplinecvscore}}\label{AppCVscore}

\begin{proof}

We start with the following lemma:

\begin{lemma} \label{cvlemma}
For $\lambda(t)$, $\gamma$ and for fixed $i$, let $\mathbf{f}^{(-i)}$ be the vector with components $f_j^{(-i)}=\hat{f}^{(-i)}\left(t_j,\lambda,\gamma\right)$,  $\mathbf{f}'^{(-i)}$ by the vector with components $f_j'^{(-i)}=\hat{f}'^{(-i)}\left(t_j,\lambda,\gamma\right)$, and define vectors $\mathbf{y}^*$ and $\mathbf{v}^*$ by 
\begin{align}
\begin{cases}
y_j^*=y_j &j \neq i\\
y_i^*=\hat{f}^{(-i)}(t_i) &\mbox{otherwise}
\end{cases}\\
\begin{cases}
v_j^*=v_j &j \neq i\\
v_i^*=\hat{f}'^{(-i)}(t_i) &\mbox{otherwise}
\end{cases}
\end{align}
Then
\begin{align}
\mathbf{\hat{f}}^{(-i)}&=S\mathbf{y}^*+\gamma T\mathbf{v}^*\\
\mathbf{\hat{f}}'^{(-i)}&=U\mathbf{y}^*+\gamma V\mathbf{v}^*
\end{align}
\end{lemma}

\begin{proof}
For any smooth curve $f$ with $\mathbf{y}^*$ and $\mathbf{v}^*$, we have 
\begin{align*}
&\frac{1}{n} \sum_{j=1}^n\left(y_j^*-f(t_j)\right)^2+\frac{\gamma}{n} \sum_{j=1}^n\left(v_j^*-f'(t_j)\right)^2+\sum_{j=1}^{n} \lambda_j\int_{t_j}^{t_{j+1}} (f''(t))^2dt \\
\geq &\frac{1}{n} \sum_{j\neq i}\left(y_j^*-f(t_j)\right)^2+\frac{\gamma}{n} \sum_{j\neq i}\left(v_j^*-f'(t_j)\right)^2+\sum_{j=1}^{n} \lambda_j\int_{t_j}^{t_{j+1}} (f''(t))^2dt\\
\geq &\frac{1}{n}\sum_{j\neq i}\left(y_j^*-\hat{f}^{(-i)}(t_j)\right)^2+\frac{\gamma}{n} \sum_{j\neq i}\left(v_j^*-\hat{f}'^{(-i)}(t_j)\right)^2+\sum_{j=1}^{n} \lambda_j\int_{t_j}^{t_{j+1}}  \left(\hat{f}^{''(-i)}(t)\right)^2dt\\
= &\frac{1}{n}\sum_{j=1}^{n}\left(y_j^*-\hat{f}^{(-i)}(t_j)\right)^2+\frac{\gamma}{n} \sum_{j=1}^{n}\left(v_j^*-\hat{f}'^{(-i)}(t_j)\right)^2+\sum_{j=1}^{n} \lambda_j\int_{t_j}^{t_{j+1}}  \left(\hat{f}^{''(-i)}(t)\right)^2dt\\
\end{align*}
%For any spline $f$ and a given $\lambda(t)$, 
%\begin{equation}
%\begin{split}
%&\frac{1}{n}\sum_{j=1}^{n}\left(y_j^*-f(t_j)\right)^2+\frac{\gamma}{n} \sum_{j=1}^{n}\left(v_j^*-f'(t_j)\right)^2+\int\lambda(t) f''^2 \\
%\geq&\frac{1}{n}\sum_{j\neq i}^{n}\left(y_j^*-f(t_j)\right)^2+\frac{\gamma}{n} \sum_{j\neq i}^{n}\left(v_j^*-f'(t_j)\right)^2+\int\lambda(t) f''^2\\
%\geq&\frac{1}{n}\sum_{j\neq i}^{n}\left(y_j^*-\hat{f}^{(-i)}(t_j)\right)^2+\frac{\gamma}{n} \sum_{j\neq i}^{n}\left(v_j^*-\hat{f}'^{(-i)}(t_j)\right)^2+\int\lambda(t) \hat{f}^{(-i)''2}\\
%=&\frac{1}{n}\sum_{j=1}^{n}\left(y_j^*-\hat{f}^{(-i)}(t_j)\right)^2+\frac{\gamma}{n} \sum_{j=1}^{n}\left(v_j^*-\hat{f}'^{(-i)}(t_j)\right)^2+\int\lambda(t) \hat{f}^{(-i)''2}
%\end{split}
%\end{equation}
by the definition of $\mathbf{\hat{f}}^{(-i)}$, $\mathbf{\hat{f}}'^{(-i)}$ and the fact that $y_i^*=\hat{f}^{(-i)}(t_i)$, $v_i^*=\hat{f}'^{(-i)}(t_i)$. It follows that $\hat{f}^{(-i)}$ is the minimizer of the objective function \eqref{tractorsplineObjective}, so that
\begin{align*}
\mathbf{\hat{f}}^{(-i)}&=S\mathbf{y}^*+\gamma T\mathbf{v}^*\\
\mathbf{\hat{f}}'^{(-i)}&=U\mathbf{y}^*+\gamma V\mathbf{v}^*
\end{align*}
as required.
\end{proof}

As a consequence of lemma \ref{cvlemma}, we obtain expressions for the deleted residuals $y_i-\hat{f}^{(-i)}(t_i)$ and $v_i-\hat{f}'^{(-i)}(t_i)$ in terms of $y_i-\hat{f}(t_i)$ and $v_i-\hat{f}'(t_i)$ respectively:

\begin{equation}\label{th3proofeq1}
\begin{split}
\hat{f}^{(-i)}(t_i)-y_i=& \sum_{j=1}^{n}S_{ij}y_j^*+ \gamma \sum_{j=1}^{n}T_{ij}v_j^*-y_i^*\\
=&\sum_{j\neq i}^{n}S_{ij}y_j+ \gamma \sum_{j\neq i}^{n}T_{ij}v_j+S_{ii}\hat{f}^{(-i)}(t_i)+\gamma T_{ii}\hat{f}'^{(-i)}(t_i)-y_i\\
=&\sum_{j=1}^{n}S_{ij}y_j+ \gamma \sum_{j=1}^{n}T_{ij}v_j+S_{ii}\left(\hat{f}^{(-i)}(t_i)-y_i\right)+\gamma T_{ii}\left(\hat{f}'^{(-i)}(t_i)-v_i\right)-y_i\\
=&\left(\hat{f}(t_i)-y_i\right)+S_{ii}\left(\hat{f}^{(-i)}(t_i)-y_i\right)+\gamma T_{ii}\left(\hat{f}'^{(-i)}(t_i)-v_i\right)
\end{split}
\end{equation}
and
\begin{equation}
\begin{split}
\hat{f}'^{(-i)}(t_i)-v_i=& \sum_{j=1}^{n}U_{ij}y_j^*+ \gamma \sum_{j=1}^{n}V_{ij}v_j^*-v_i^*\\
=&\sum_{j\neq i}^{n}U_{ij}y_j+ \gamma \sum_{j\neq i}^{n}V_{ij}v_j+U_{ii}\hat{f}^{(-i)}(t_i)+\gamma V_{ii}\hat{f}'^{(-i)}(t_i)-v_i\\
=&\sum_{j=1}^{n}U_{ij}y_j+ \gamma \sum_{j=1}^{n}V_{ij}v_j+U_{ii}\left(\hat{f}^{(-i)}(t_i)-y_i\right)+\gamma V_{ii}\left(\hat{f}'^{(-i)}(t_i)-v_i\right)-v_i\\
=&\left(\hat{f}'(t_i)-v_i\right)+U_{ii}\left(\hat{f}^{(-i)}(t_i)-y_i\right)+\gamma V_{ii}\left(\hat{f}'^{(-i)}(t_i)-v_i\right).
\end{split}
\end{equation}
Thus 
\begin{equation}\label{th3proofeq2}
\hat{f}'^{(-i)}(t_i)-v_i = \frac{\hat{f}'(t_i)-v_i}{1-\gamma V_{ii}}+ \frac{U_{ii}\left(\hat{f}^{(-i)}(t_i)-y_i\right)}{1-\gamma V_{ii}}.
\end{equation}
By substituting equation \eqref{th3proofeq2} into \eqref{th3proofeq1}, we obtain
\begin{equation*}
\hat{f}^{(-i)}(t_i)-y_i=\frac{\hat{f}(t_i)-y_i+\gamma \frac{T_{ii}}{1-\gamma V_{ii}}\left(\hat{f}'(t_i)-v_i\right)}{1-S_{ii}-\gamma\frac{T_{ii}}{1-\gamma V_{ii}}U_{ii}}.
\end{equation*}
Consequently, 
\begin{equation*}
CV(\lambda,\gamma)=\frac{1}{n}\sum_{i=1}^{n}\left( \frac{\hat{f}(t_i)-y_i+\gamma \frac{T_{ii}}{1-\gamma V_{ii}}\left(\hat{f}'(t_i)-v_i\right)}{1-S_{ii}-\gamma\frac{T_{ii}}{1-\gamma V_{ii}}U_{ii}}\right)^2.
\end{equation*}
\end{proof}

\bibliographystyle{apalike}
%\bibliography{meth}
\bibliography{WorkingAll}

\end{document}